\documentclass[aps,prx,twocolumn,floatfix,superscriptaddress,longbibliography,notitlepage]{revtex4-2}

\usepackage[utf8]{inputenc}
\usepackage[T1]{fontenc}
\usepackage{qcircuit}

\usepackage{upgreek} % for \upmu

\usepackage{changes}
\setdeletedmarkup{\color{blue}\sout{#1}}
\usepackage{cancel}

\usepackage{mathrsfs}
\usepackage[normalem]{ulem}
\usepackage{graphicx, color, graphpap}% Include figure files
\usepackage{enumitem}
\usepackage{amssymb}
\usepackage{amsthm}
\usepackage{multirow}
\usepackage[colorlinks=true,citecolor=blue,linkcolor=magenta]{hyperref}
\usepackage[T1]{fontenc}
\usepackage{verbatim}
\usepackage{mathtools}
\usepackage{titlesec}

\long\def\ca#1\cb{} %Use for commenting out: \ca...\cb

\newcommand{\abs}[2][]{#1| #2 #1|}

\newcommand{\ketbra}[2]{| \hspace{1pt} #1 \rangle \langle #2 \hspace{1pt} |}

\newcommand{\bramatket}[3]{\langle #1 \hspace{1pt} | #2 | \hspace{1pt} #3 \rangle}

\newcommand{\ket}[1]{|#1\rangle}               %ket
\newcommand{\bra}[1]{\langle #1|}              %bra
\newcommand{\dya}[1]{\ket{#1}\!\bra{#1}}
\newcommand{\rank}{\text{rank}}

\newcommand{\Tr}{{\rm Tr}}

\newcommand{\supp}{\text{supp}}

\newcommand{\ave}[1]{\langle #1\rangle}               %average
\renewcommand{\geq}{\geqslant}
\renewcommand{\leq}{\leqslant}

\renewcommand{\vec}[1]{\boldsymbol{#1}}  % Bold vectors instead of arrow vectors

\newcommand{\ad}{^\dagger}
\newcommand*{\id}{\openone}

\newcommand{\sinc}{\text{sinc}}

\newtheorem{theorem}{Theorem}

\newtheorem{definition}{Definition}

\begin{document}
\title{Quantum Inception Score}
%\date{\today}

\author{Akira Sone}
\email{akira.sone@umb.edu}
\affiliation{Department of Physics, University of Massachusetts, Boston, Massachusetts 02125, USA}
\affiliation{{The NSF AI Institute for Artificial Intelligence and Fundamental Interactions}}

\author{Akira Tanji}
%\email{}
\affiliation{
Department of Applied Physics and Physico-Informatics, Keio University,
Hiyoshi 3-14-1, Kohoku, Yokohama 223-8522, Japan}

\author{Naoki Yamamoto}
\email{yamamoto@appi.keio.ac.jp}
\affiliation{
Department of Applied Physics and Physico-Informatics, Keio University,
Hiyoshi 3-14-1, Kohoku, Yokohama 223-8522, Japan}
\affiliation{Keio Quantum Computing Center, Keio University,
Hiyoshi 3-14-1, Kohoku, Yokohama 223-8522, Japan}

\begin{abstract}
Motivated by the great success of classical generative models in machine learning, enthusiastic exploration of their quantum version has recently started. To depart on this journey, it is important to develop a relevant metric to evaluate the quality of quantum generative models; in the classical case, one such example is the (classical) inception score (cIS). 
In this paper, as a natural extension of cIS, we propose the quantum inception score (qIS) for quantum generators. 
Importantly, qIS relates the quality to the Holevo information of the quantum channel that classifies a given dataset. 
In this context, we show several properties of qIS. 
First, qIS is greater than or equal to the corresponding cIS, which is defined through projection measurements on the system output. 
Second, the difference between qIS and cIS arises from the presence of quantum coherence, as characterized by the resource theory of asymmetry. 
Third, when a set of entangled generators is prepared, there exists a classifying process leading to the further enhancement of qIS. 
Fourth, we harness the quantum fluctuation theorem to characterize the physical limitation of qIS.
Finally, we apply qIS to assess the quality of the one-dimensional spin chain model as a quantum generative model, with the quantum convolutional neural network as a quantum classifier, for the phase classification problem in the quantum many-body physics.
\end{abstract}

\maketitle

\section{Introduction}
\label{sec:intro}

A burgeoning advancement of the artificial intelligence (AI) is a hallmark of the contemporary information age. 
The principal objective of AI research is the creation of machines which exhibit human-like capabilities of performing complicated tasks including learning, analyzing and reasoning. Presently, the prevailing landscape of AI technologies is predominantly underpinned by the machine learning algorithms, such as convolutional neural network and support vector machine, 
which have been embracing various applications~\cite{goodfellow2016deep,aggarwal2018neural,finlayson2019adversarial,biggio2018wild,tanaka2021deep}. 
However, recently, the machine learning has been warned to adhere to a Moore's law-like trend concerning the size of datasets~\cite{sarma2019machine}, namely the curse of the dimensionality. 
Considering that people seek to encode data into larger feature spaces to facilitate pattern discovery, there is a pressing need for an alternative approach to reduce complexity in the context of big data analysis.

Quantum computers hold significant potential to overcome such challenge due to their intrinsically greater information storage and information processing capacities compared to the classical devices~\cite{Nielsen}. 
Therefore, there have been  enthusiastic explorations into quantum-enhanced machine learning~\cite{schutt2020machine,biamonte2017quantum,jerbi2023quantum,torlai2020machine,schuld2021machine,gili2022generalization}. 
In this paper, our primary focus is on the generative modeling~\cite{turhan2018recent,xu2015overview}, which is one of the most successful framework in classical unsupervised learning  \cite{barlow1989unsupervised,dike2018unsupervised}. 
Various quantum generative models have been proposed, and most of which employ the variational approach, e.g., Refs.~\cite{gao2018quantum,gao2022enhancing,tezuka2024generative}. 
In particular we find quantum generative models motivated from the classical architectures \cite{Kingma_Book_Autoencoder,goodfellow2020generative}, such as variational quantum autoencoders (VQAEs)~\cite{Romero17,khoshaman2018quantum,pepper2019experimental,ma2020compression}, the quantum generative adversarial networks (qGANs)~\cite{lloyd2018quantum,dallaire2018quantum,huang2021experimental} and quantum Boltzmann machines~\cite{amin2018quantum,zoufal2021variational,kieferova2017tomography}. Furthermore, quantum generative models have recently found promising applications in physics domains, such as quantum many-body physics~\cite{lian2019machine,zhang2022experimental,sels2021quantum}.

\begin{figure*}[htp!]
\centering
\includegraphics[width=2\columnwidth]{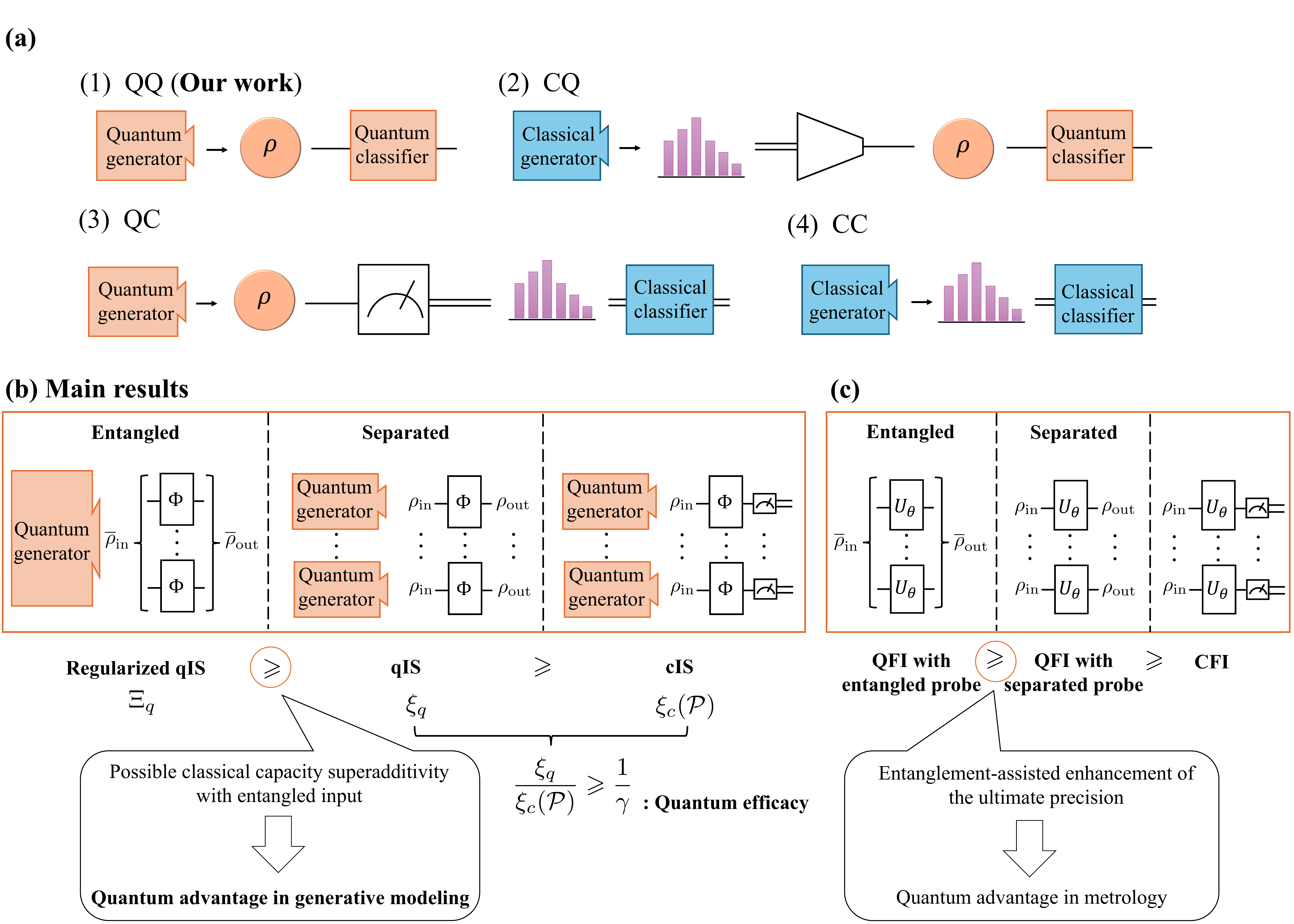}
\caption{
(a) Four types of protocols in the generative modeling. Generative modeling can involve quantumness in both the generator and classifier (QQ protocol), only in the classifier (CQ protocol), only in the generator (QC protocol), or in neither (CC protocol). This paper studies the QQ protocols. (b) Summary of the main results. 
We propose a definition of the qIS (Definition~\ref{def:main1}) and maximally achievable qIS (defined as regularized qIS) by employing the Holevo information and regularized classical capacity, respectively, as a measures of the quality of the quantum generative models. Particularly, when the entangled input $\overline{\rho}_{\rm{in}}$ is allowed, there exists a classifying process $\Phi$ leading to the further quality enhancement. This can be regarded as the quantum advantage in the generative modeling. We also demonstrate that cIS can be recovered via the projection measurements and qIS is larger than or equal to cIS (Theorem~\ref{th:main2}) due to the presence of quantum coherence (Theorem~\ref{th:main3}) characterized by the resource theory of asymmetry.
Furthermore, by employing the quantum fluctuation theorem approach, we illustrate that  {the difference between qIS and cIS decreases} due to pure dephasing in the quantum classifier, which can be characterized by the quantum efficacy (Theorem~\ref{th:main4}).  
(c) An analogous view in the quantum metrology. The ultimate precision limit characterized by the quantum Fisher information (QFI) surpasses the classical limit characterized by the classical Fisher information (CFI). Further enhancement of the QFI can be achieved by the entangled probes, which is the quantum advantage in the metrology.}
\label{fig:summary}
\end{figure*}

To assess the quality of generative models, we need a meaningful and practically computable metric. 
Actually there exists such a metric in the classical regime, which is known as the (classical) {\it inception score} (cIS) \cite{barratt2018note,salimans2016improved}. 
The cIS is a single metric that measures both the diversity of the generated data and the sharpness of each data, the latter of which can be quantified by a classifier. 
In this paper, as a natural extension of cIS, we propose the \textit{quantum inception score} (qIS) for quantum generators. 
Importantly, qIS relates the quality to the Holevo information \cite{holevo1973bounds,schumacher2001optimal}. 
As in the classical case, qIS measures the quality of quantum generative models with the help of quantum classifier. 
Note here that, in analogy to quantum metrology protocols~\cite{giovannetti2006quantum}, the quantum generative modeling problem can involve quantumness in both the generator and classifier (QQ protocol), only in the classifier (CQ protocol), only in the generator (QC protocol), or in neither (CC protocol), as shown in Fig.~\ref{fig:summary}(a). 
The protocol considered in this paper is the QQ protocol. 
With this analogy, our results regarding quantum generative models can also be seen as analogous to the quantum advantage in metrology illustrated in Fig.~\ref{fig:summary}(c), where the ultimate precision limit, characterized by the quantum Fisher information (QFI), surpasses the classical limit characterized by the classical Fisher information (CFI) and is further enhanced by entangled probes. This analogy holds because the protocols considered in quantum metrology in Ref.~\cite{giovannetti2006quantum} are analogous to those in quantum communications~\cite{bennett1998quantum}.

Our main results are Definition~\ref{def:main1} and Theorems~\ref{th:main2}, \ref{th:main3} and \ref{th:main4}, which are summarized in Fig.~\ref{fig:summary}(b). 
First, qIS is greater than or equal to the corresponding cIS, which is defined through projection measurements on the system output; this is analogous to the result in quantum metrology where QFI is bigger than CFI. 
Second, the proposed qIS allows us to demonstrate that the presence of quantum coherence, preserved by the quantum classifier, is the resource  {generating the difference between qIS and cIS.} 
 {Second, the proposed qIS demonstrates that measurements on the quantum classifier can be one factor leading to quality degradation due to the destruction of quantum coherence preserved by the classifier.} 
Third, the entanglement of the generator's output could lead to the further quality enhancement, which is coming from the the potential superadditivity of the communication capacity of the quantum classifier channel~\cite{hastings2009superadditivity}. This can be regarded as the quantum advantage in quantum generative models. 
Forth, leveraging the qIS, we employ the information-theoretic fluctuation theorem~\cite{vedral2012information,sone2023jarzynski,kafri2012holevo} to characterize the physical limitation of the quality of quantum generative models based on the concept of quantum efficacy~\cite{kafri2012holevo,Sagawa10,fujitani2010jarzynski}.
Finally, we provide an example of applying the qIS to assess the quality of the one-dimensional (1D) spin chain model combined with the quantum convolutional neural network (QCNN)~\cite{cong2019quantum,grant2018hierarchical}, for the phase classification problem. 
These results corroborate the significance of exploring the quantum foundation and communication approach to study the quantum machine learning protocols.

We here remark that our study does not involve comparisons between quantum generators and classical generators, nor between quantum classifiers and classical classifiers. 
This is in stark contrast to the study of Gao \textit{et al.} \cite{gao2022enhancing} and Huang \textit{et al.} \cite{huang2022quantum}; the former compares quantum and classical generators, and the latter compares quantum and classical classifiers. 
More precisely, Gao \textit{et al.}~\cite{gao2022enhancing} study the expressivity of quantum and classical generators, to demonstrate that quantum generators have better expressivity than classical ones due to entanglement. 
On the other hand, Huang \textit{et al.}~\cite{huang2022quantum} study quantum and classical classifiers with a fixed quantum generator, and they demonstrate that the quantum classifier performs better if the quantum generator produces entanglement. 
 {Different from these studies, focusing on the quality of the quantum generative models in QQ protocol, we explore the roles of quantum resources such as quantum coherence and entanglement, analyze their physical limitations using a quantum thermodynamics approach, and examine their applications in quantum phase classification.}

This paper is organized as follows. 
In Sec.~\ref{sec:IS}, we briefly review the concept of IS in the classical generative models and introduce the qIS in Sec.~\ref{sec:qIS}. 
Then, we {discuss the relation between} the qIS and cIS in Sec.~\ref{sec:qISvscIS}, and further discuss the role of quantum coherence in Sec.~\ref{sec:coherence}. 
Furthermore, we utilize the quantum fluctuation theorem to characterize the degradation of the quality in Sec.~\ref{sec:limitation}. 
Finally, we show the examples of using the qIS to assess the quantum generative models for the phase classification of the 1D spin-$1/2$ chain in Sec.~\ref{sec:example}, followed by the conclusion in Sec.~\ref{sec:conclusion}.

\section{Inception score}
\label{sec:IS}
Let us first review the concept of IS based on Refs.~\cite{barratt2018note,salimans2016improved}. 
Let $\mathcal{X}\equiv\{x_0,x_1,x_2\cdots,x_{r-1}\}$ be a dataset produced from an unknown probability distribution. 
The aim of generative modeling is to construct a model {$p_{\rm{in}}(x)$} that approximates the unknown probability distribution producing $\mathcal{X}$. 
Given a data $x_i$, we also aim to construct a classifier that produces a relevant label $y_j\in\mathcal{Y}\equiv\{y_0,y_1,y_2,\cdots,y_{\ell-1}\}$, where $\ell$ denotes the number of the labels; we model the classifier via the conditional probability {$q_{\rm{out}}(y|x)$}, which is usually constructed via a neural network. Note here that, in the generative models, we always have $\ell\leq r$. 
This is because we usually assume that the number of labels characterizing the patterns of the encoded dataset is smaller than that of the number of the input data. 
The marginal probability for the label data, $p_{\rm{out}}(y_j)$, is given by $p_{\rm{out}}(y_j)=\sum_{i=0}^{r-1}q_{\rm{out}}(y_j|x_i)p_{\rm{in}}(x_i)$. 
Therefore, the \textit{input} and \textit{output} of the classifier are  {$p_{\rm{in}}(x_i)$} and {$p_{\rm{out}}(y_j)$}, respectively.

Then, the IS of  $p_{\rm{in}}(x)$ relative to  $q_{\rm{out}}(y|x)$ is defined by
\begin{align}
{\xi\equiv\exp\left(\sum_{i=0}^{r-1}p_{\rm{in}}(x_i)D(q_{\rm{out}}(y|x_i)\,\|\,p_{\rm{out}}(y))\right),}
\label{eq:IS}
\end{align}
where 
\begin{align}
    {D(q_{\rm{out}}(y|x_i)\,\|\,p_{\rm{out}}(y))\equiv\sum_{j=0}^{\ell-1}q_{\rm{out}}(y_j|x_i)\ln\frac{q_{\rm{out}}(y_j|x_i)}{p_{\rm{out}}(y_j)}}
\label{eq:KL}
\end{align}
is the Kullback–Leibler (KL) divergence of the conditional label probability $q_{\rm{out}}(y_j|x_i)$ with respect to  $p_{\rm{out}}(y_j)$. Note that $\ln\xi$ can be expressed as  {$\ln\xi = \sum_i p_{\rm{in}}(x_i) \sum_j q_{\rm{out}}(y_j|x_i)\ln q_{\rm{out}}(y_j|x_i) -\sum_j p_{\rm{out}}(y_j) \ln p_{\rm{out}}(y_j)$}, that is, the sum of the expected negative Shannon entropy of  {$q_{\rm{out}}(y|x)$} and the entropy of  $p_{\rm{out}}(y)$. 
The former quantifies the accuracy of the label assigned to the data and the latter does the diversity of the data, implying that a generative model with bigger IS may cast as a high-quality data generator. 
Therefore, in this scenario, the primary objective is to construct generative models, with the help of relevant design of the classifier, that achieves the higher IS; for this reason, IS has been often used for GAN.

\section{Quantum inception score}
\label{sec:qIS}

We consider a general setup of quantum generative model with a system described by $d$-dimensional Hilbert space $\mathcal{H}_S$. 
Let $\mathcal{B}(\mathcal{H}_S)$ denote the set of density operators acting on $\mathcal{H}_S$. 
Suppose that the generator encodes a classical data (or a latent variable) $x_i\in\mathcal{X}$ to  a quantum state $\rho_{\rm{in}}(x_i)\in\mathcal{B}(\mathcal{H}_S)$, for $i=0,1,2,\cdots,(r-1)$. Here, note that $\{\rho_{\rm{in}}(x_i)\}_{i=0}^{r-1}$ are the input states of the quantum classifier. 

As in the usual machine learning scenario that tries to solve problems by encoding data into a larger dimensional space, we assume $d\geq r$. 
The encoded quantum states are then processed by a completely positive and trace-preserving (CPTP) map $\Phi:\mathcal{B}(\mathcal{H}_S)\to\mathcal{B}(\mathcal{H}_S')$, where $\mathcal{H}_S'$ denotes the $d'$-dimensional Hilbert space of the output system. 
Thus the output state $\rho_{\rm{out}}\in\mathcal{B}(\mathcal{H}_S')$ is related to the input as $\rho_{\rm{out}}=\Phi(\rho_{\text{in}})$.

Now, let us write the input state $\rho_{\text{in}}$ as an ensemble of $\rho_{\text{in}}(x_i)\in\mathcal{B}(\mathcal{H}_S)$ encoding the classical input data $x_i\in\mathcal{X}$ sampled from a probability distribution  $p_{\rm{in}}(x_i)$, i.e.,
\begin{align}
\rho_{\text{in}}=\sum_{i=0}^{r-1}{p_{\rm{in}}(x_i)}\rho_{\text{in}}(x_i).
\label{eq:input}
\end{align}
Note that in general $\rho_{\text{in}}(x_i)$ does not necessarily commute with $\rho_{\text{in}}(x_{i'})$ when $i\neq i'$. Then, the output state becomes 
\begin{align}
    \rho_{\rm{out}}=\sum_{i=0}^{r-1}{p_{\rm{in}}(x_i)}\rho_{\rm{out}}^{(i)}\,,
\label{eq:OutputEnsemble}
\end{align}
where
\begin{align}
\rho_{\rm{out}}^{(i)}\equiv\Phi(\rho_{\text{in}}(x_i)).
\label{eq:EachOutput}
\end{align}

Let us write the ensemble of the output state as 
\begin{align}
\mathcal{E}\equiv\left\{ {p_{\rm{in}}(x_i)},\rho_{\rm{out}}^{(i)}\right\}_{i=0}^{r-1}\,.
\end{align}
Because the supports of $\rho_{\rm{out}}^{(i)}$ and $\rho_{\rm{out}}$ satisfy  $\supp(\rho_{\rm{out}}^{(i)})\subseteq\supp(\rho_{\rm{out}})$, the quantum relative entropy $S(\rho_{\rm{out}}^{(i)}\,\|\,\rho_{\rm{out}})\equiv \Tr[\rho_{\rm{out}}^{(i)}\ln\rho_{\rm{out}}^{(i)}]-\Tr[\rho_{\rm{out}}^{(i)}\ln\rho_{\rm{out}}]$ takes a finite value. 
Then the Holevo information $\chi(\mathcal{E};\rho_{\rm{out}})$ of the output ensemble is defined as~\cite{holevo1973bounds,schumacher2001optimal}
\begin{align}
\chi(\mathcal{E};\rho_{\rm{out}})&\equiv \sum_{i=0}^{r-1} {p_{\rm{in}}(x_i)}S(\rho_{\rm{out}}^{(i)}\,\|\,\rho_{\rm{out}})\,.
\label{eq:Holevo}
\end{align}
The Holevo-Schumacher-Westmoreland (HSW) capacity (or Holevo capacity)~\cite{hausladen1996classical,schumacher1997sending,holevo1998capacity,king2001capacity,holevo1999capacity,siudzinska2020classical,lee2015classical} is defined as the maximization of the Holevo information over all the possible  {$p_{\rm{in}}(x_i)$} and $\rho_{\text{in}}(x_i)$:
\begin{align}
    \chi_{\max}(\Phi)\equiv\max_{\{{p_{\rm{in}}(x_i)},\rho_{\text{in}}(x_i)\}} \chi(\mathcal{E};\rho_{\rm{out}})\,.
\label{eq:HolevoCapacity}
\end{align}
In the most generic scenario, the maximal transmittable amount of classical information through the quantum channel $\Phi$ is measured by the regularized classical capacity $C(\Phi)$~\cite{holevo1999capacity,siudzinska2020classical,lee2015classical,gyongyosi2018survey,hastings2009superadditivity} 
\begin{align}
    C(\Phi) \equiv \lim_{n\to \infty}\frac{1}{n}\chi_{\max}(\Phi^{\otimes n})\,,
\label{eq:RegularizedClassicalCapacity}
\end{align}
which is asymptotically achievable with infinite  {copies} of $\Phi$ {with joint measurement}, and we always have ~\cite{siudzinska2020classical,lee2015classical,gyongyosi2018survey,hastings2009superadditivity}
\begin{align}
    C(\Phi) \geq \chi_{\max}(\Phi)\,.
\label{eq:ClassicalCapacityInequality}
\end{align}
From these facts, as our \textit{first main result}, we define the qIS as follows.
\begin{definition}[Quantum Inception Score]
\label{def:main1}
The quantum inception score (qIS) is defined as
\begin{align}
    \xi_q\equiv \exp\left(\chi(\mathcal{E};\rho_{\rm{out}})\right)\,.
\label{eq:QIS}
\end{align}
{Additionally, the achievable maximum qIS is defined as 
\begin{align}
\Xi_q \equiv \exp(C(\Phi))\,,
\label{def:regularizedIS}
\end{align}
which we refer to as the regularized qIS. 
}
\end{definition}

Suppose now that the task of quantum generator is to generate the optimal set of $\{{p_{\rm{in}}(x_i)},\rho_{\text{in}}(x_i)\}$ to maximize $\xi_q$. 
Then, because the asymptotically achievable maximal qIS is given by $\exp(C(\Phi))$, from Eqs.~\eqref{eq:RegularizedClassicalCapacity} and \eqref{eq:ClassicalCapacityInequality}, $\exp(C(\Phi))$ is regarded as an indicator of the best achievable quality by the quantum generator. 
Adopting an example of the image recognition, these results are consistent to the intuition that having a greater volume of information enables the generator to achieve better balanced accuracy and diversity. Particularly, in the asymptotic setup, 
when {entangled} input $\overline{\rho}_{\rm in}\in\mathcal{B}(\mathcal{H}_S^{\otimes n})$ is allowed, there exists a CPTP map $\Phi$ leading to the superadditivity of the classical capacity $C(\Phi)>\chi_{\max}(\Phi)$ ~\cite{hastings2009superadditivity}. 
From Eq.~\eqref{def:regularizedIS}, this results in the quantum advantage in the generative modeling achieved by using the entanglement as a resource.

\section{{Relation between the qIS and cIS}}
\label{sec:qISvscIS}

Here, we explore the relation between the qIS and cIS.  
By introducing projective measurements, we can recover the cIS as follow. 
To classify the data into $\ell$ labels, we assume $\ell=d'=\dim(\mathcal{H}_S')$. Let $\{\ket{y_j}\}_{j=0}^{\ell-1}$ be an orthonormal basis for $\mathcal{H}_S'$. Then, the rank-1 projective measurement onto the orthogonal state $\ket{y_j}$ is given by $\Pi_j\equiv\dya{y_j}$. Then, we have $\Pi_j\Pi_k=\Pi_j\delta_{jk}$ with $\delta_{jk}$ the Kronecker's $\delta$. Also, the completeness relation $\sum_{j=0}^{\ell-1}\Pi_j=\id$ holds, where $\id$ denotes the identity operator acting on $\mathcal{H}_S'$.  
Regarding the dimensions of the Hilbert spaces, hence, the meaningful setup is $d\geq r\geq\ell$. Therefore, $\Phi$ belongs to the class of tree-like quantum classifiers composed of the hierarchical quantum circuits~\cite{grant2018hierarchical,cong2019quantum,li2022recent,hur2022quantum,wrobel2021application,chen2023quantum}, where the compressed states are expected to carry the features of the encoded data $\mathcal{X}$.

For a state $\rho\in\mathcal{B}(\mathcal{H}_S')$, let us write the post-projection-measurement state as
\begin{align}
    \mathcal{P}(\rho)\equiv\sum_{j=0}^{\ell-1}\Pi_j\rho\Pi_j=\sum_{j=0}^{\ell-1}\bramatket{y_j}{\rho}{y_j}\dya{y_j}\,,
\end{align} 
which is also a dephasing map generating an incoherent state diagonal in the basis $\{\ket{y_j}\}_{j=0}^{\ell-1}$. 
Then, the probability that the output state takes $\ket{y_j}$ is given by $p_{\rm{out}}(y_j)\equiv\Tr[\rho_{\rm{out}}\Pi_j]=\bramatket{y_j}{\rho_{\rm{out}}}{y_j}$. 
From Eqs.~\eqref{eq:input}, \eqref{eq:OutputEnsemble}, and \eqref{eq:EachOutput}, we have ${q_{\rm{out}}(y_j|x_i)}\equiv\Tr[\rho_{\rm{out}}^{(i)}\Pi_j]=\bramatket{y_j}{\rho_{\rm{out}}^{(i)}}{y_j}$. 
Because $\mathcal{P}(\rho_{\rm{out}}^{(i)})$ and $\mathcal{P}(\rho_{\rm{out}})$ have the identical eigenbasis $\{\ket{y_j}\}_{j=0}^{\ell-1}$, we have $S(\mathcal{P}(\rho_{\rm{out}}^{(i)})\,\|\,\mathcal{P}(\rho_{\rm{out}}))={D(q_{\rm{out}}(y|x_i)\,\|\,p_{\rm{out}}(y))}$, meaning that the cIS in Eq.~\eqref{eq:IS} can be recovered from Eq.~\eqref{eq:QIS} by projective measurements. Therefore,  {when we write} the Holevo information of the projected output ensemble due to $\mathcal{P}$ as 
\begin{align}
    \chi_{\mathcal{P}}(\mathcal{E};\rho_{\rm{out}})\equiv \sum_{i=0}^{r-1}{p_{\rm{in}}(x_i)}S(\mathcal{P}(\rho_{\rm{out}}^{(i)})\,\|\,\mathcal{P}(\rho_{\rm{out}}))\,,
\label{eq:ProjectHolevo}
\end{align}
the cIS dependent on the choice of $\mathcal{P}$ can be written as
\begin{align}
\xi_c(\mathcal{P})\equiv \exp\left(\chi_{\mathcal{P}}(\mathcal{E};\rho_{\rm{out}})\right)\,.
\label{eq:CIS}
\end{align}

Extending to the positive operator-valued measures (POVMs) $\mathcal{M}\equiv\{E_j\}_{j=0}^{\ell-1}$, where $E_j\geq0$ is a POVM element satisfying $\sum_{j=0}^{\ell-1}E_j=\id$. Also, note that with $j$th POVM element, the probabilities are ${p_{\rm{out}}(y_j)}\equiv\Tr[\rho_{\rm{out}}E_j]$ and ${q_{\rm{out}}(y_j|x_i)}\equiv\Tr[\rho_{\rm{out}}^{(i)}E_j]$. The maximum cIS can be given by using the \textit{accessible information}~\cite{holevo1973bounds,jozsa1994lower,davies1978information,WatrousBook18,vedral2002role}
\begin{align}
\begin{split}
I_{\rm{acc}}&(\mathcal{E};\rho_{\rm{out}})\\
&\equiv \max_{\{\mathcal{M}\}}\Bigg[\sum_{j=0}^{\ell-1}\sum_{i=0}^{r-1}p_{\rm{in}}(x_i)\Tr[\rho_{\rm{out}}^{(i)}E_j]\ln \Tr[\rho_{\rm{out}}^{(i)}E_j]\\
&\quad\quad\quad\quad-\sum_{j=0}^{\ell-1}\Tr[\rho_{\rm{out}}E_j]\ln\Tr[\rho_{\rm{out}}E_j]\Bigg]\\
&=\max_{\{\mathcal{M}\}}\left[\sum_{i=0}^{r-1}p_{\rm{in}}(x_i)D_{\mathcal{M}}(q_{\rm{out}}(y|x_i)\,\|\,p_{\rm{out}}(y))\right]\,,
\end{split}
\label{eq:accessible}
\end{align}
which is the maximization of the classical mutual information over all possible POVMs $\{\mathcal{M}\}$. Here, we intentionally write $D_{\mathcal{M}}$ to emphasize the dependence of the KL divergence on the choice of the POVMs. The Holevo theorem states~\cite{holevo1973bounds} 
\begin{align}
    \chi(\mathcal{E};\rho_{\rm{out}})\geq I_{\rm{acc}}(\mathcal{E};\rho_{\rm{out}})
\end{align}
with the equality if and only if 
\begin{align}
[\rho_{\rm{out}}^{(n)},\rho_{\rm{out}}^{(m)}]=0~(\forall n,\forall m)\,,
\end{align}
implying that $\rho_{\rm{out}}^{(n)}$ and $\rho_{\rm{out}}^{(m)}$ can be simultaneously diagonalized. 

By defining the \textit{accessible IS} as 
\begin{align}
    \xi_{\rm{acc}} \equiv \exp(I_{\rm{acc}}(\mathcal{E};\rho_{\rm{out}}))\,,
\end{align}
we can obtain our second main result as follows
\begin{theorem}
\label{th:main2}
The inception scores, $\xi_q$, $\xi_{\rm{acc}}$, and $\xi_c(\mathcal{P})$, satisfy
\begin{align}
\ell\geq\xi_q\geq\xi_{\rm{acc}}\geq
\xi_c(\mathcal{P})~(\forall\mathcal{P})\,.
\label{eq:QIS>CIS}
\end{align}
\end{theorem}
\begin{proof}
First, $\xi_q\geq\xi_{\rm{acc}}$ is the Holevo theorem itself. Second, because the projective measurements belong to the POVMs, from Eq.~\eqref{eq:accessible}, we must have $\xi_{\rm{acc}}\geq \xi_c(\mathcal{P})$. {Finally, from $\ln\ell\geq \chi(\mathcal{E};\rho_{\rm{out}})$ and Eq.~\eqref{eq:QIS}, we have $\ell\geq \xi_q$.} Therefore, we can obtain Eq.~\eqref{eq:QIS>CIS}.
\end{proof}
 {Theorem~\ref{th:main2} implies that performing the measurements on the classifier is a primary factor causing the difference between qIS and cIS.}

\section{Quantum coherence and quality}
\label{sec:coherence}

 {From Theorem~\ref{th:main2}, we can intuitively infer that the difference between qIS and cIS arises due to the destruction of quantum coherence by the measurements. Indeed, this inference is correct.} 
In the following, we employ the resource theory of asymmetry (RTA) to provide a rigorous and formal analysis.

\subsection{Review of resource theory of asymmetry}

First, we briefly review the {RTA}~\cite{lostaglio2015quantum,lostaglio2015description, marvian2014asymmetry,marvian2020coherence,marvian2016quantum,takagi2019skew,yamaguchi2023smooth,ahmadi2013wigner,
gour2008resource,marvian2013theory,marvian2014extending,gour2009measuring,marvian2016clocks,vaccaro2008tradeoff,gour2008resource}, where the quantum coherence is regarded as a resource of breaking the group symmetry. Let $G$ be a symmetry group, and $g$ be the group element with its unitary representation $U_g:G\to\mathcal{B}(\mathcal{H})$ acting on a $D$-dimensional Hilbert space $\mathcal{H}$. Let 
\begin{align}
\mathcal{I}_G(\mathcal{H})\equiv\{\sigma\,|\, U_g\sigma U_g\ad = \sigma, \forall g\in G, \sigma\in\mathcal{B}(\mathcal{H})\}
\end{align}
be the set of the free states in the RTA, which are invariant under any unitary operations with respect to $g$. The free state satisfies the commutation relation $[\rho,U_g]=0~(\forall g)$ and is called the \textit{symmetric} state with respect to $G$, and \textit{asymmetric} state otherwise, which becomes a resource state in the RTA. 

A relevant set of free operations are the covariant operations $\Lambda$~\cite{bartlett2003entanglement} with respect to $G$, which satisfies 
\begin{align}
\Lambda(U_g\rho U_g\ad)=U_g\Lambda(\rho)U_g^{\dagger}~~(\forall g\in G,\forall\rho\in\mathcal{B}(\mathcal{H}))\,.
\end{align}
The covariant operation cannot generate the asymmetric states from the symmetric state and transform one asymmetric state to the other. 

To quantify the asymmetry of a given quantum state $\rho$ with respect to $G$, a asymmetry measure $A(\rho;G)$ needs to satisfy the following conditions:
The asymmetry measure must satisfy 
\begin{enumerate}
    \item{$A(\rho;G)\geq 0~(\forall\rho\in\mathcal{B}(\mathcal{H}))$ and $A(\rho;G) = 0 \Longleftrightarrow \rho\in\mathcal{I}_G(\mathcal{H})$. }
    \item{For all covariant operations $\{\Lambda\}$, we must have $A(\rho;G)\geq A(\Lambda(\rho);G)~(\forall \rho\in\mathcal{B}(\mathcal{H}))$.}
\end{enumerate}
One of the asymmetry measures is the \textit{relative entropy of asymmetry}~\cite{gour2009measuring}. To define the relative entropy of asymmetry, let us first introduce the $G$-twirling operation~\cite{gour2008resource,marvian2013theory,marvian2014extending,gour2009measuring,marvian2016clocks,bartlett2003entanglement,vaccaro2008tradeoff}, which is defined as
\begin{align}
\mathcal{G}(\rho)\equiv\int_{G}dg U_g\rho U_g\ad
\end{align}
averaging over the unitary operations with the Haar measure $dg$~\footnote{For a finite group, we have $\mathcal{G}(\rho)\equiv\frac{1}{\abs{G}}\sum_{g\in G}U_g\rho U_g\ad$, where $\abs{G}$ denotes the order of the group.}. When $G$ is a finite or compact Lie group, the relative entropy of asymmetry with respect to $G$ is defined as~\cite{gour2009measuring}
\begin{align}
    A(\rho;G)\equiv S(\rho \,\|\, \mathcal{G}(\rho))\,.
\end{align}

Now, let us consider the case $G=U(1)$ generated by an observable $H$, whose unitary representations form a set of time translations $\{e^{-itH}\,|\,\forall t\in\mathbb{R}\}$. In this case, $\rho$ is a symmetric state if and only if $[\rho,H]=0$ and an asymmetric state if and only if $[\rho,H]\neq 0$ ~\cite{yamaguchi2023smooth}. Therefore, when $\rho$ is a symmetric state, $\rho$ can be diagonalized by the eigenbasis of $H$. Let us write $A(\rho;H)$ as the relative entropy of asymmetry for this case. When $H$ has $L$ distinct eigenvalues, the explicit form of the corresponding $U(1)$-twirling operation with respect to $H$ is given by ~\cite{marvian2016clocks,gour2009measuring} (see the Appendix~\ref{app:twirling} for detailed explanations) 
\begin{align}
\!\!\!\!\mathcal{G}_H(\rho)\!\equiv\!\lim_{T\to\infty}\left[\frac{1}{2T}\int_{-T}^{T}dte^{-iHt}\rho e^{iHt}\right]\!=\!\sum_{n=1}^{L}\Pi_n\rho\Pi_n,
\label{eq:twirling}
\end{align}
where $L\leq D$ and $\{\Pi_n\}_{n=1}^{L}$ are the projectors onto the subspace of the eigenbasis of $H$. Then, the relative entropy of asymmetry $A(\rho,H)\equiv S(\rho\,\|\,\mathcal{G}_H(\rho))$ coincides with the so-called relative entropy of superposition with respect to the orthogonal decomposition of the Hilbert space~\cite{aberg2006quantifying}. Particularly, when $H$ is nondegenerate (i.e.,  $L=D$), we have $\rank(\Pi_n)=1~(\forall n)$, and the relative entropy of asymmetry coincides with the relative entropy of coherence with respect to the eigenabsis of $H$~\cite{baumgratz2014quantifying}.

\subsection{Asymmetry and quantum inception score}
Now, we are ready to discuss the relation between the quantum inception score and the quantum coherence captured by the asymmetry with respect to $U(1)$ group generated by $\rho_{\rm{out}}^{(i)}$ the constituent states of the output state. 

For the output state $\rho_{\rm{out}}=\sum_{i=0}^{r-1} {p_{\rm{in}}(x_i)} \rho_{\rm{out}}^{(i)}$ with its fixed ensemble $\mathcal{E}=\{{p_{\rm{in}}(x_i)},\rho_{\rm{out}}^{(i)}\}_{i=0}^{r-1}$, for each $i$, we consider the set of time translations $\{e^{-i\rho_{\rm{out}}^{(i)}t}\,|\,\forall t\in\mathbb{R}\}$. When $[\rho_{\rm{out}}^{(n)},\rho_{\rm{out}}^{(m)}]=0~(\forall n, \forall m)$, we particularly call the corresponding ensemble as \textit{symmetric ensemble}
\begin{align}
    \!\!\!\!\mathcal{E}_{S}\equiv\left\{{p_{\rm{in}}(x_i)},\rho_{\rm{out}}^{(i)}\,\Big|\,[\rho_{\rm{out}}^{(n)},\rho_{\rm{out}}^{(m)}]=0~(\forall n,\forall m)\right\}_{i=0}^{r-1}\,.
\label{eq:symmetric}
\end{align}
However, we define the \textit{asymmetric ensemble} as 
\begin{align}
    \!\!\!\!\!\overline{\mathcal{E}_{S}}\equiv\left\{{p_{\rm{in}}(x_i)},\rho_{\rm{out}}^{(i)}\,\Big|\,[\rho_{\rm{out}}^{(n)},\rho_{\rm{out}}^{(m)}]\neq0~(\exists n,\exists m)\right\}_{i=0}^{r-1}\,,
\end{align}
which is the complement of $\mathcal{E}_{S}$. To explore the relation between the quantum inception score and the asymmetry measure, we define the average relative entropy of asymmetry of the output state $\rho_{\rm{out}}$ as 
\begin{align}
\ave{A(\rho_{\rm{out}};H)} \equiv \sum_{i=0}^{r-1}{p_{\rm{in}}(x_i)}A(\rho_{\rm{out}}^{(i)};H)\,,
\label{eq:AveCoherence}
\end{align}
which measures the average amount of coherence contained by $\rho_{\rm{out}}$, which is characterized by the asymmetry with respect to the $U(1)$ group generated by $H$. Then, we can obtain our \textit{third main result} as follows
\begin{theorem}
\label{th:main3}
For the output state $\rho_{\rm{out}}=\sum_{i=0}^{r-1}{p_{\rm{in}}(x_i)}\rho_{\rm{out}}^{(i)}$ with its fixed ensemble $\mathcal{E}=\{{p_{\rm{in}}(x_i)},\rho_{\rm{out}}^{(i)}\}_{i=0}^{r-1}$, we have 
\begin{align}  \!\!\!\xi_q>\xi_{\rm{acc}}\Longleftrightarrow \mathcal{E}= \overline{\mathcal{E}_{S}}\Longleftrightarrow \ave{A(\rho_{\rm{out}};\rho_{\rm{out}}^{(k)})}\neq 0~(\exists k)\,.
\label{eq:AsymmetryResource}
\end{align}
\end{theorem}
\begin{proof}
We prove by taking the contraposition of the following statement
\begin{align}  \!\!\!\xi_q=\xi_{\rm{acc}}\Longleftrightarrow \mathcal{E}= \mathcal{E}_{S}\Longleftrightarrow \ave{A(\rho_{\rm{out}};\rho_{\rm{out}}^{(k)})}= 0~(\forall k)\,.
\label{eq:NoAsymmetry}
\end{align}
For the first part,  $\xi_q=\xi_{\rm{acc}}\Longleftrightarrow \mathcal{E}=\mathcal{E}_{S}$ is the Holevo theorem. 

For the second part, We consider a set of time translations 
$\{e^{-it \rho_{\rm{out}}^{(k)}}\,|\,\forall t\in\mathbb{R}\}$ generated by the density operator of the constituent state $\rho_{\rm{out}}^{(k)}$.

Let us prove the sufficiency. 
When $\mathcal{E}=\mathcal{E}_{S}$, we have $[\rho_{\rm{out}}^{(i)},\rho_{\rm{out}}^{(k)}]=0~(\forall i, \forall k)$. Therefore, the relative entropy of asymmetry of $\rho_{\rm{out}}^{(i)}$ with respect to $\rho_{\rm{out}}^{(k)}$ must vanish, i.e.,  $A(\rho_{\rm{out}}^{(i)};\rho_{\rm{out}}^{(k)})=0~(\forall i,\forall k)$, because of the condition of the asymmetry measure. Therefore, from Eq.~\eqref{eq:AveCoherence}, we have $\ave{A(\rho_{\rm{out}};\rho_{\rm{out}}^{(k)})}=0~(\forall k)$. Therefore, the sufficiency holds
\begin{align}
\mathcal{E}=\mathcal{E}_{S}\Longrightarrow   \ave{A(\rho_{\rm{out}};\rho_{\rm{out}}^{(k)})}=0~(\forall k)\,.
 \label{eq:suff}
\end{align}

Next, let us prove the necessity. Let us focus on Eq.~\eqref{eq:AveCoherence}. Here, we have ${p_{\rm{in}}(x_i)}>0~(\forall i)$ and the nonnegativity of the quantum relative entropy $S(\rho_{\rm{out}}^{(i)}\,\|\,\mathcal{G}_k(\rho_{\rm{out}}^{(i)}))\geq0~(\forall i,\forall k)$, where $\mathcal{G}_k$ denotes the $U(1)$-twirling operation with respect to $\rho_{\rm{out}}^{(k)}$. Therefore, when we have $\ave{A(\rho_{\rm{out}};\rho_{\rm{out}}^{(k)})}=0~(\forall k)$, we must have $S(\rho_{\rm{out}}^{(i)}\,\|\,\mathcal{G}_k(\rho_{\rm{out}}^{(i)}))=0~(\forall i,\forall k)$. From the condition of the asymmetry measure, this implies $[\rho_{\rm{out}}^{(i)},\rho_{\rm{out}}^{(k)}]=0~(\forall i,\forall k)$, namely  $\mathcal{E}=\mathcal{E}_{S}$. Therefore, the necessity holds:
\begin{align}\ave{A(\rho_{\rm{out}};\rho_{\rm{out}}^{(k)})}=0~(\forall k)\Longrightarrow \mathcal{E}=\mathcal{E}_{S}\,.
\label{eq:nece}
\end{align}

From Eqs.~\eqref{eq:suff} and \eqref{eq:nece}, we can obtain Eq.~\eqref{eq:NoAsymmetry}. By taking its contraposition, we obtain Eq.~\eqref{eq:AsymmetryResource}, which proves Theorem~\ref{th:main3}. 
\end{proof}

This theorem demonstrates that the qIS is larger than {the maximum cIS (i.e., the accessible IS)} if and only if $\rho_{\rm{out}}$ has an asymmetry with respect to the $U(1)$ group generated by some constituent states $\rho_{\rm{out}}^{(i)}$. This means that the quantum coherence preserved during the classification process, captured by the asymmetry, is the resource  {generating the difference between qIS and cIS.} 
 {This theorem not only proves that this quality degradation occurs due to the destruction of quantum coherence by measurements, but it also captures the specific characteristics of this quantum coherence. It shows that the coherence destroyed by the optimal POVM, resulting in minimal quality degradation, is characterized by the asymmetry of $\rho_{\rm{out}}$ with respect to the constituent states $\rho_{\rm{out}}^{(i)}$.}
Also, note that this theorem holds for any asymmetry measures, such as skew informations~\cite{marvian2016quantum,takagi2019skew,yamaguchi2023smooth,ahmadi2013wigner}.

 {Here we remark that we are not comparing quantum generative models with and without measurements. 
Yet the IS has been used as a measure for model selection in the classical regime. 
In our case, the idea is as follows. 
Suppose that we have two models $A$ and $B$. 
Then, it makes sense to compare their accessible IS (i.e.,  maximum cIS) $\xi_{\rm{acc}}(A)$ versus $\xi_{\rm{acc}}(B)$ for the purpose of choosing the better model. 
We note that, since a gap can arise between $\xi_q$ and $\xi_{\rm{acc}}$, comparing their qISs $\xi_q(A)$ and $\xi_q(B)$ is not always appropriate, 
but when the gap vanishes, the qIS serves as a fundamental quantum-limited quantity for model selection, and Theorem~\ref{th:main3} clarifies such situation. 
}

\section{ {Impact of decoherence on quality}}
\label{sec:limitation}

Finally, we discuss  {another primary factor causing the quality degradation of the quantum generative models.} 
From Theorem~\ref{th:main3}, we can also expect that the pure dephasing (or decoherence) contributes to degrading the quality.  
Here, we demonstrate that the quality degradation mechanism can be captured by the information-theoretic fluctuation theorems~\cite{vedral2012information,sone2023jarzynski,kafri2012holevo} based on the concept of quantum efficacy~\cite{kafri2012holevo,Sagawa10,fujitani2010jarzynski}. 
Following Ref.~\cite{kafri2012holevo}, when $\Delta a$ is a random variable whose average is $\ave{\Delta a} = \chi(\mathcal{E};\rho_{\rm{out}})-\chi_{\mathcal{P}}(\mathcal{E};\rho_{\rm{out}})$, the corresponding quantum fluctuation theorem is given as $\ave{\exp(-\Delta a)}=\gamma$, where $\gamma$ is called quantum efficacy~\footnote{{For the origin of $\gamma$ regarding the Holevo information, refer to Eqs.~(11)\,--\,(14) in Ref.~\cite{kafri2012holevo}.}} and satisfies $0<\gamma\leq 1$.

From the Jensen's inequality, we have $\chi(\mathcal{E};\rho_{\rm{out}})-\chi_{\mathcal{P}}(\mathcal{E};\rho_{\rm{out}})\geq -\ln\gamma$.
Here, note that $\gamma$ is strictly dependent on $\Phi$ and the choice of the projective measurement $\mathcal{P}$. Then, we can obtain our \textit{fourth main result} as follows.
\begin{theorem}
\label{th:main4}
The quantum inception score can be lower bounded by using quantum efficacy $\gamma$ as
\begin{align}
    \xi_q\geq \frac{\xi_c(\mathcal{P})}{\gamma}~~(0<\gamma\leq1)\,.
\label{eq:Efficacy}
\end{align}
\end{theorem}
\begin{proof}
By definitions, we have $\xi_q\equiv \exp\left(\chi(\mathcal{E};\rho_{\rm{out}})\right)$ and $ \xi_c(\mathcal{P})\equiv \exp\left(\chi_{\mathcal{P}}(\mathcal{E};\rho_{\rm{out}})\right)$. Because $\exp(-\ln\gamma)=\gamma^{-1}$, we can obtain $\xi_q\geq \gamma^{-1}\xi_c(\mathcal{P})$ with 
$0<\gamma\leq1$.
\end{proof}
From Refs.~\cite{kafri2012holevo,ruskai2002inequalities}, when $\gamma=1$, the equality of Eq.~\eqref{eq:Efficacy} holds if and only if all $\rho_{\rm{out}}^{(i)}$ have \textit{identical} eigenbasis and $\mathcal{P}$ is the dephasing map onto this common eigenbasis, i.e.,  $\mathcal{P}(\rho_{\rm{out}}^{(i)})=\rho_{\rm{out}}^{(i)}~~(\forall i)$. Therefore, we have  $\rho_{\rm{out}}=\mathcal{P}(\rho_{\rm{out}})$, which means that $\rho_{\rm{out}}$ is an incoherent state diagonal in this common eigenbasis. This implies that the quality degradation of the quantum generative model is primarily due to the pure dephasing process.
This analysis also enables us to interpret $\gamma$ as a quantity characterizing the \textit{physical limitation} of the quality due to the information loss. In particular, the negative logarithm $(-\ln\gamma)$ accounts for the information content contained by the quantum coherence, which is {preserved during the classification process}.

\section{{Application to a phase classification problem in quantum many-body physics}}
\label{sec:example}

As an example, we harness the quantum inception score to analyze the quality of a fixed quantum generator with a trainable quantum classifier. Here, the generator is the one-dimensional (1D) spin-$1/2$ chain that may experience phase transition, and the classifier is the quantum convolutional neural network (QCNN)~\cite{cong2019quantum}.
{Note that a main goal here is to test the gap between qIS and cIS for a fixed generator, rather than designing a good generator achieving a higher qIS like the situation of GAN.}

\subsection{QCNNs}
The architecture of the QCNNs were introduced in Refs.~\cite{cong2019quantum,grant2018hierarchical}, which are expected to be promising near-term quantum algorithms as quantum classifiers~\cite{li2022recent,hur2022quantum,wrobel2021application,chen2023quantum} because of their absence of barren plateaus~\cite{pesah2021absence} resulting in the trainability of the QCNNs~\cite{uvarov2020barren, mcclean2018barren,cerezo2021barren,wang2020noise, zhang2022quantum}. 
Following Ref.~\cite{pesah2021absence}, let us consider the following QCNN circuit (see Fig.~\ref{fig:QCNN} as an example for the 8-qubit case).

\begin{figure}[htp!]
\centering
\includegraphics[width=1.0\columnwidth]{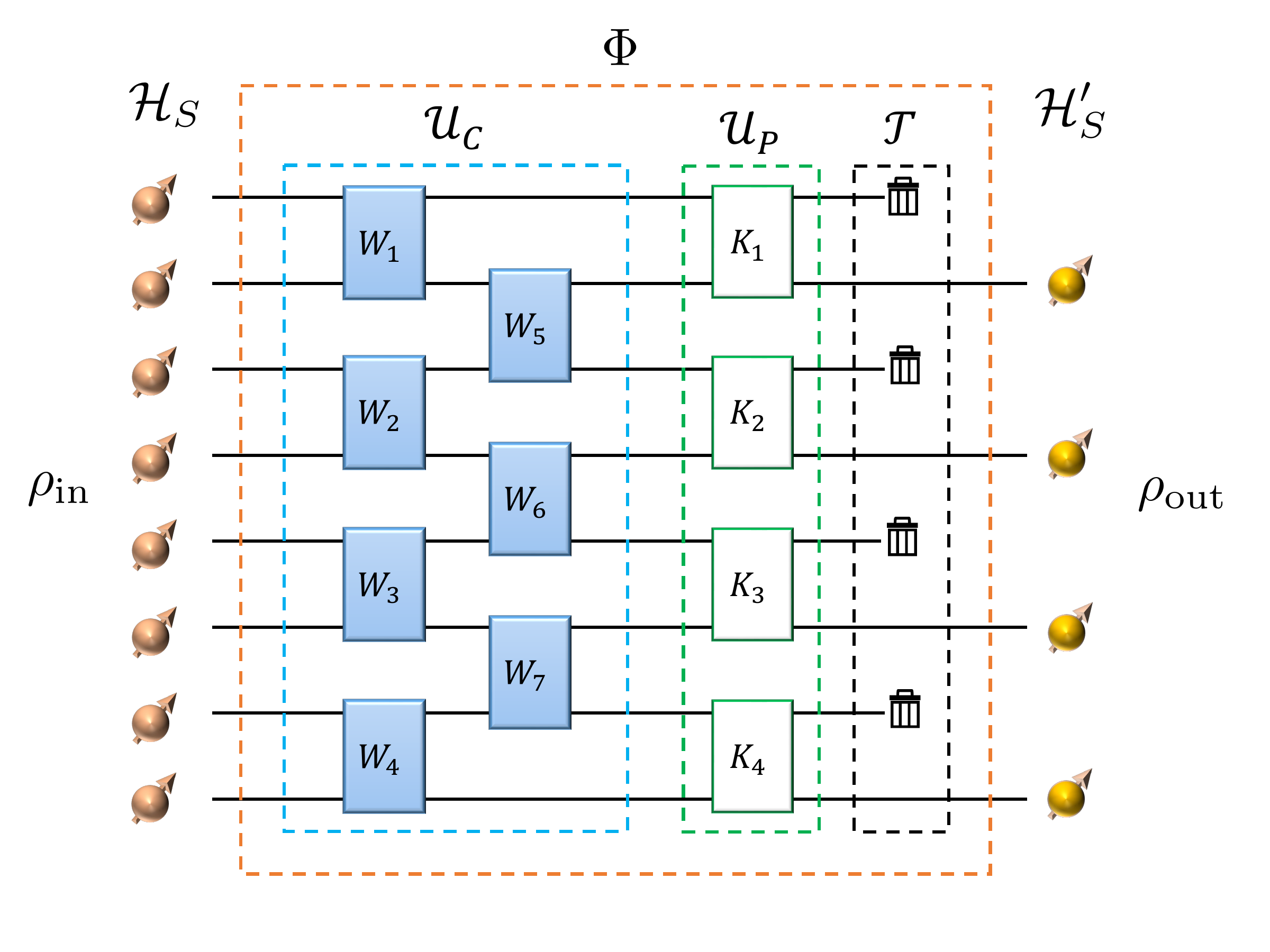}
\caption{The 8-qubit quantum convolutional neural network. $\mathcal{U}_{C}$ and $\mathcal{U}_P$ denote the convolutional and pooling layer. $\{W_1,\cdots, W_7\}$ are the 2-qubit unitaries acting on the pairs of qubits in an alternating manner, which are parameterized by some tunable parameters. 
$\{K_1,\cdots, K_4\}$ are the 2-qubit unitaries with the form of $K_j\equiv\dya{0}\otimes U_j+\dya{1}\otimes V_j$, where $U_j$ and $V_j$ are single-qubit unitaries parameterized by some tunable parameters. $\mathcal{T}\equiv \Tr_{\overline{\mathcal{H}}_S'}$ denotes the partial trace over $\overline{\mathcal{H}}_S'$ (the complement of $\mathcal{H}_S'$). Then, the quantum classifier channel is given by $\Phi\equiv\mathcal{T}\circ\mathcal{U}_P\circ\mathcal{U}_C$.  }
\label{fig:QCNN}
\end{figure}

The input state $\rho_{\text{in}}$ is an $N$-qubit state, which may encode the classical data, namely the output of a generator. 
Then, we send $\rho_{\text{in}}$ into the single layered convolutional circuit denoted by a unitary operation $\mathcal{U}_C$ and then the pooling circuit denoted by another unitary operation $\mathcal{U}_P$. 
The convolutional circuit $\mathcal{U}_C$ consists of two columns of the alternating 2-qubit gate parameterized by some tunable parameters. 
The pooling circuit $\mathcal{U}_P$ consists of the 2-qubit gates with the form of $K_j=\dya{0}\otimes U_j+\dya{1}\otimes V_j$, where $j$ denotes the $j$th pair of the qubits, and $U_j$ and $V_j$ are single-qubit gates parameterized by some tunable parameters. 
Then, we take the partial trace $\mathcal{T}\equiv\Tr_{\overline{\mathcal{H}}_S'}$ over $\overline{\mathcal{H}}_S'$, the complement of the output Hilbert space $\mathcal{H}_S'$, to obtain the output state $\rho_{\rm{out}}$. 
Therefore, the whole process is written as  
\begin{align}
\rho_{\rm{out}}=\Phi(\rho_{\text{in}})\equiv\mathcal{T}\circ\mathcal{U}_P\circ\mathcal{U}_C(\rho_{\text{in}})\,.
\end{align}

\subsection{Setup and problem formulation}

\subsubsection{Generator and quantum dataset}

Adopting the same model studied in Ref.~\cite{cong2019quantum}, we employ the QCNN to classify 2-class and 3-class quantum phases. 
The target quantum state to be classified, $\rho_{\text{in}}(x)=\dya{\psi_0(x)}$, is the ground state of an $N$-body Hamiltonian of a 1D spin-$1/2$ chain with open boundary conditions~\cite{jeyaretnam2021quantum,chen2014symmetry}:
\begin{align}
    H \!=\!-J\sum^{N-2}_{n=1}Z_{n}X_{n+1}Z_{n+2} 
              \!-\!h_{1}\sum^{N}_{n=1}X_{n}
              \!-\!h_{2}\sum^{N-1}_{n=1}X_{n}X_{n+1},
\label{eq:Hamiltonian}
\end{align}
where $\{X_n, Y_n, Z_n\}$ are the Pauli matrices of the $n$th spin. 
Also, $J$, $h_1$, and $h_2$ are the strength of the cluster coupling, the global transverse field, and the nearest-neighboring Ising coupling, respectively. 
These parameters take several values, which accordingly lead to several ground states of $H$. In particular, we collect those parameters into the vector
\begin{align}
x \equiv \begin{pmatrix}h_1/J\\ h_2/J\end{pmatrix}\,,
\end{align}
and we write the corresponding Hamiltonian in Eq.~\eqref{eq:Hamiltonian} as $H(x)$.

Here, we consider the case of $N=9$. 
The parameters $h_1/J$ and $h_2/J$ take equally separated 64 values in the intervals $h_1/J\in[0,1.6]$ and $h_2/J\in[-1.6,1.6]$, respectively; therefore, we consider totally $64\times64=4096$ points, i.e.,  $\{x_i\}_{i=0}^{4095}=\{\vec{z}_{n,m}\}_{(n,m)=(0,0)}^{(63,63)}$, where 
\begin{align}
\vec{z}_{n,m}\equiv 
\begin{pmatrix}
\frac{1.6}{63}n\\
-1.6+\frac{3.2}{63}m
\end{pmatrix}~(n,m=0,1,\cdots,63)\,.
\end{align}
{Later, we will show the phase diagram of the ground states $\{\rho_{\rm in}(x_i)\}_{i=0}^{4095}$.}

Apart from these parameter points, we take 40 ground states $\{\rho_{\rm in}(x_i)\}_{i=0}^{39}$ as the training quantum data, where the parameter vectors $\{x_0,x_1,\cdots,x_{39}\}$ are taken as $h_1/J=1$ and $h_2/J\in[-1.6,1.6]$, namely 
\begin{align}
\label{training data}
    x_i \equiv 
    \begin{pmatrix}
    1\\
    -1.6+\frac{3.2}{39}i
    \end{pmatrix}\,~(i=0,1,\cdots,39)\,.
\end{align}
In our simulation, the ground state $\rho_{\text{in}}(x_i)$ is obtained by diagonalizing $H(x_i)$.

\begin{table}[htp!]
\begin{tabular}{|c||c|c|}
\hline
\multicolumn{3}{|c|}{(a) 2-class phase classification}\\
\hline
Label &  
$h_2/J$& 
Phase\\     
\hline
0 & $-1.15<h_2/J<0$ & SPT \\
\hline
1 & $h_2/J<-1.15$ or $h_2/J>0$  & Other phases\\
\hline
\multicolumn{3}{|c|}{(b) 3-class phase classification} \\
\hline
Label &  
$h_2/J$& 
Phase\\     
\hline
0 & $-1.15<h_2/J<0$ & SPT \\
\hline
1 & N/A & Nothing\\
\hline
2& $h_2/J>0$& Paramagnetic\\
\hline
3& $h_2/J<-1.15$& Antiferromagnetic\\
\hline
\end{tabular}
\caption{The (a) 2-class and (b) 3-class phase of the ground state of the Hamiltonian \eqref{eq:Hamiltonian} when $h_1/J=1$. 
{Note that for the 3-class case, we employ a 4-valued POVM, and the label $``1"$ has the meaning of ``fail in classification.'' }} 
 \label{table:2-3-class}
\end{table}

Table~\ref{table:2-3-class} summarizes the labels (the phase) for the 2-class and 3-class cases when $h_1/J=1$. 
For the 2-class classification problem, we assign the label $``0"$ to the $\mathbb{Z}_2\times\mathbb{Z}_2$ symmetry-protected topological (SPT) phase~\cite{haldane1983nonlinear,gu2009tensor,pollmann2012symmetry} when $-1.15<h_2/J<0$ and $``1"$ to the other phases. 
For the 3-class classification problem, we assign the label $``0"$ to the SPT phase when $-1.15<h_2/J<0$, $``2"$ to the paramagnetic phase when $h_2/J>0$, and $``3"$ to the antiferromagnetic phase when $h_2/J<-1.15$. {Note that, later in the discussion, we introduce a four-valued POVM to classify the phase, where the additional label $``1"$ has the meaning of ``fail in classification''. 
Also, the general phase other than the case $h_1/J=1$ will be shown later.  
}

\subsubsection{Training of QCNN}

The QCNN circuit used for this quantum phase classification problem is shown in Fig.~\ref{fig:QCNN_for_PhaseClassifications}.
\begin{figure}[htp!]
\centering
\includegraphics[width=1\columnwidth]{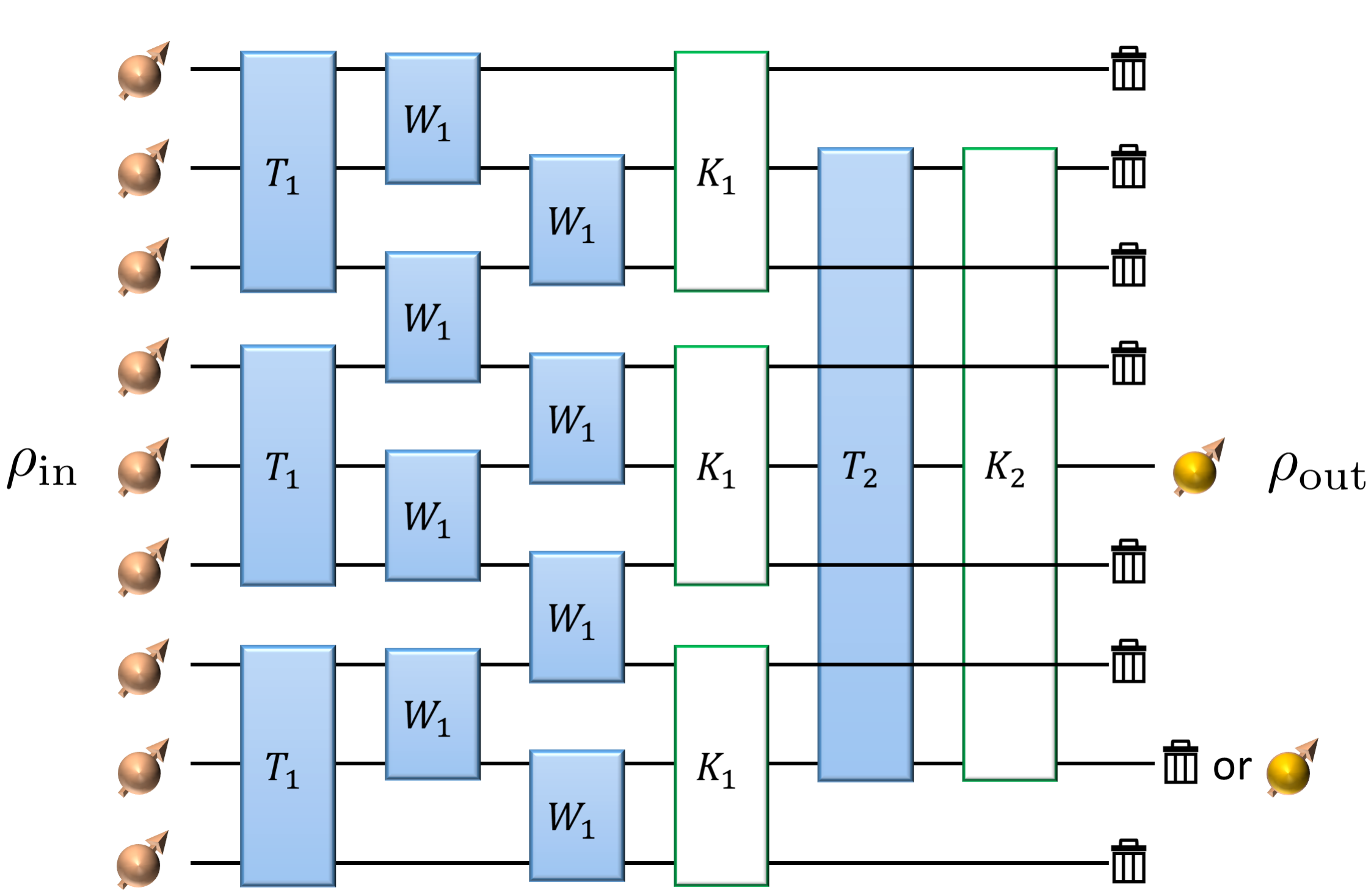}
\caption{The 9-qubit QCNN circuit for the 2-class and 3-class quantum phase classification problems. $T_j$ and $K_j$ denote the 3-qubit gate, and $W_j$ denotes the 2-qubit gate. We take partial trace over all qubits except for the fifth qubit for 2-class and except for both the fifth and eighth qubits for the 3-class classification, respectively.
} 
\label{fig:QCNN_for_PhaseClassifications}
\end{figure}
This circuit consists of $N=9$ qubits, and it contains 117 learning parameters for the 2-class classification and 156 for the 3-class classification problems, respectively. The gates with the same name and index share the same parameters.
The convolutional layer consists of $W_j$ and $T_j$. $W_j$ is the 2-qubit convolutional gate with the form of $W_j = (V\otimes V)e^{-i\vec{\theta}\cdot(ZZ,YY,XX)}(V\otimes V)$, where $\vec{\theta}$ is the 3-dimensional vector and $V$ is the single-qubit gate with the form of $V = e^{-i\phi_1Z}e^{-i\phi_2Y}e^{-i\phi_3X}$ which has 3 parameters. $W_j$ has 15 parameters since each $V$ in $W_j$ has different parameters. $T_j$ is the 3-qubit convolutional gate with the form of  $T_j = W^{(3,1)}W^{(2,3)}W^{(1,2)}$, where $W^{(a,b)}$ acts on qubits indexed by $a$ and $b$. $T_j$ has 45 parameters since each $W$ in $T_j$ has different parameters.
The pooling layer consists of $K_j$. For the 2-class classification, $K_j$ is the 3-qubit pooling gate with the form of $K_j = (\id\otimes\id\otimes\dya{0}+\id\otimes V\otimes\dya{1})(\dya{0}\otimes\id\otimes\id+\dya{1}\otimes V\otimes\id)$ which has 6 parameters. 
For the 3-class classification, $K_1$ is the same as $K_j$ described above, while $K_2$ is the same form as $T_j$.

As for the dimensions of the Hilbert space of the output state, we have $\ell=2$ for the 2-class classification and $\ell=4$ for the 3-class classification problems, respectively. 
{This corresponds to the partial trace operation after applying $K_2$; that is, we take the partial trace over all qubits except for the 5-th qubit for the 2-class classification and except for both 5-th and 8-th qubits for the 3-class classification, respectively. } 
Let us define $\ket{+}$ and $\ket{-}$ as $\ket{\pm} \equiv (\ket{0}\pm\ket{1})/\sqrt{2}$\,,
where $\{\ket{0},\ket{1}\}$ is the computational basis. Then, for the 2-class classification, we define $\Pi_{0}\equiv\dya{+}$ and $\Pi_{1}\equiv\dya{-}$ as the projectors corresponding to the labels $``0"$ and $``1,"$ respectively. For the 3-class classification, we define $\Pi_0\equiv\dya{+,+}$, $\Pi_{1}\equiv\dya{+,-}$, $\Pi_{2}\equiv\dya{-,+}$ and $\Pi_{3}\equiv\dya{-,-}$ as the projectors corresponding to the labels $``0,"$ $``1,"$ $``2"$ and $``3,"$ respectively.

In the QCNN training procedure, we minimize the cost function by updating the parameters of QCNN with the use of the simultaneous perturbation stochastic approximation (SPSA) optimizer~\cite{spall1998overview,spall1999stochastic}. 
Here, the cost function is the cross-entropy loss given by 
\begin{align}
\begin{split}
\!\!\!\mathcal{L}(\Phi)&=-\sum_{i,j}p_{\rm{true}}(x_i,y_j)\ln p_{\rm{train}}(x_i,y_j)\\
&=-\sum_{i,j}p_{\rm{true}}(y_j|x_i) {p_{\rm{in}}(x_i)}\ln (p_{\rm{train}}(y_j|x_i) {p_{\rm{in}}(x_i)})\\
&=-\frac{1}{r}\sum_{i,j}p_{\rm{true}}(y_j|x_i)\ln\frac{p_{\rm{train}}(y_j|x_i)}{r}\\
&=-\frac{1}{r}\sum_{i,j}p_{\rm{true}}(y_j|x_i)\ln p_{\rm{train}}(y_j|x_i)+\ln r,
\end{split}
\label{eq:costfunction}
\end{align}
where
\begin{align}
p_{\rm{train}}(y_j|x_i)\equiv \Tr\left[\Phi(\rho_{\text{in}}(x_i))\Pi_j(x_i)\right], 
\end{align}
with $\Phi$ the QCNN and $\Pi_j(x_i)$ the projector corresponding to the label $y_j$ of the training data $x_i$ encoded in $\rho_{\text{in}}(x_i)$. 
{
Also, $p_{\rm{true}}(y_j|x_i)$ is the true distribution, where $y_j$ is the index of phase assigned to $\rho_{\text{in}}(x_i)$. 
Recall that we are given $40$ training data \eqref{training data}, and we here assume that $\rho_{\rm in}(x_i)$ appears with equal probability ${p_{\rm{in}}(x_i)}=1/40$. 
}

\subsubsection{Prediction and quality evaluation}

{After training the QCNN, we apply the trained QCNN to predict the phase (label) of test data 
\footnote{
Specifically, we measure the output of the trained QCNN with some projectors and assign the label that corresponds to the highest probability of the measurement results. 
}
and then compute the qIS and cIS, for the 2-class and 3-class classifications problems in the following two scenarios. 
}
The first is the unbiased scenario where we randomly generate equal numbers of data for every label, and the second is the biased scenario where the numbers of randomly generated data are (largely) different for each label. 
Clearly, the former has a bigger diversity in the data, or equivalently the generator has a capability to produce a bigger variety of data; thus the values of both qIS and cIS for the former case would be bigger compared to the latter case.

In the simulation, the total number of test data is $r=1500$, where the labels are given as follows. 
For the 2-class classification problem, in the unbiased case we randomly select 750 data with label $``0"$ and 750 data with label $``1"$; in the biased case we randomly select 1480 data with label $``0"$ and 20 data with label $``1."$ 
For the 3-class classification problem, in the unbiased case we randomly select 500 data with label $``0,"$ 500 data with $``2,"$ and 500 data with $``3"$; in the biased case we randomly select 1480 data with label $``0,"$ 10 data with $``2"$ and 10 data with $``3."$ 
Table ~\ref{table:2-3-class-data} shows the summary of the setting described above. 
In all cases, {we assume that the generators are fixed and} each data is generated with equal probability; that is, we suppose 
\begin{align}
{p_{\rm{in}}(x_i)}=\frac{1}{1500}~(\forall i)\,.
\end{align}

\begin{table}[htp!]
\begin{tabular}{|c||c|c|}
\hline
\multicolumn{3}{|c|}{(a) 2-class phase classification} \\
\hline
Label &  
Unbiased& 
Biased\\     
\hline
0 & 750 & 1480 \\
\hline
1 & 750  & 20\\
\hline
\multicolumn{3}{|c|}{(b) 3-class phase classification} \\
\hline
Label &  
Unbiased& 
Biased\\     
\hline
0 & 500 & 1480 \\
\hline
1 & 0 & 0\\
\hline
2& 500 & 10\\
\hline
3& 500& 10\\
\hline
\end{tabular}
\caption{Number of randomly selected test data for the (a) 2-class and (b) 3-class classification problems. 
}
\label{table:2-3-class-data}
\end{table}

With the above setting, we can compute qIS simply using Eqs.~\eqref{eq:Holevo} and \eqref{eq:QIS} with $r=1500$. 
To compute the cIS $\xi_c(\mathcal{P})$, we further need to specify the measurement or equivalently the measurement process $\mathcal{P}$. 
In our simulation, for the 2-class classification problem, we take the following dephasing operations: 
\begin{align}
    \mathcal{P}_{2}(\rho) \equiv\bramatket{+}{\rho}{+}\dya{+}+\bramatket{-}{\rho}{-}\dya{-}. 
\end{align}
Also, for the 3-class classification problems, we take 
\begin{align}
\begin{split}
    \mathcal{P}_3(\rho)&\equiv
    \,~\bramatket{+,+}{\rho}{+,+}\dya{+,+}\\
    &~~+\bramatket{+,-}{\rho}{+,-}\dya{+,-}\\
    &~~+\bramatket{-,+}{\rho}{-,+}\dya{-,+}\\
    &~~+\bramatket{-,-}{\rho}{-,-}\dya{-,-}\,.
\end{split}
\end{align}
{
Furthermore, we consider the problem of calculating the accessible information $I_{\rm acc}$ given by Eq.~\eqref{eq:accessible} for the 2-class classification problem. 
In particular, instead of maximizing $\{\mathcal{M}\}$ over all possible POVMs, we here consider an optimization problem of the projectors in the following form: \begin{align}
\begin{split}
    \Pi_{0}(\theta,\phi)&\equiv\dya{\psi_{0}(\theta,\phi)}, \\
    \Pi_{1}(\theta,\phi)&\equiv\dya{\psi_{1}(\theta,\phi)}\,,
\end{split}
\label{eq:ParameterizedProjectors}
\end{align}
where
\begin{align}
\begin{split}
 {\ket{\psi_{0}(\theta,\phi)}} &\equiv\cos(\theta/2)\ket{0}+e^{i\phi}\sin(\theta/2)\ket{1},\\
 {\ket{\psi_1(\theta,\phi)}}&\equiv\sin(\theta/2)\ket{0}-e^{i\phi}\cos(\theta/2)\ket{1}\,.
\end{split}
\end{align}
That is, we optimize the parameters $(\theta, \phi)$ so that the cIS is maximized. 
}

For the 3-class classification problem, we use the projectors in the following form:
\begin{align}
\begin{split}
    \Pi_{0}(\vec{\theta})&\equiv U(\vec{\theta})\dya{0,0}U^\dagger(\vec{\theta})\,, \\
    \Pi_{1}(\vec{\theta})&\equiv U(\vec{\theta})\dya{0,1}U^\dagger(\vec{\theta})\,, \\
    \Pi_{2}(\vec{\theta})&\equiv U(\vec{\theta})\dya{1,0}U^\dagger(\vec{\theta})\,, \\
    \Pi_{3}(\vec{\theta})&\equiv U(\vec{\theta})\dya{1,1}U^\dagger(\vec{\theta})\,,
\end{split}
\label{eq:ParameterizedProjectors3class}
\end{align}
where $\vec\theta$ is the 15-dimensional vector and 
\begin{align}
\begin{split}
U(\vec{\theta})\equiv &(e^{-i\theta_1Z}e^{-i\theta_2Y}e^{-i\theta_3X}\otimes e^{-i\theta_4Z}e^{-i\theta_5Y}e^{-i\theta_6X})\\
\times&e^{-i(\theta_7ZZ+\theta_8YY+\theta_9XX)}\\
\times&(e^{-i\theta_{10}Z}e^{-i\theta_{11}Y}e^{-i\theta_{12}X}\otimes e^{-i\theta_{13}Z}e^{-i\theta_{14}Y}e^{-i\theta_{15}X})\,.
\end{split}
\end{align}

\subsection{Simulation results}

The inception scores for the 2-class and 3-class phase classifications are shown in Table~\ref{table:2-class-IS}. 
In the 2-class case, we calculate $\xi_c$ with the projection measurement on the $X$-axis and $Z$-axis, which correspond to $(\theta,\phi)=(\pi/2,0)$ and $(\theta,\phi)=(0,\pi)$ in Eq.~\eqref{eq:ParameterizedProjectors}, respectively. 
{Also, we calculate $\xi_c$ with the high-accuracy-axis; that is, $(\theta,\phi)$ are chosen so that the classification accuracy, which is the ratio of the number of matches between the predicted labels and the correct labels to that of all data $\{x_i\}_{i=0}^{4095}$, is almost maximized.} 
Moreover, we calculate $\xi_c$ with the optimized-axis, meaning that we choose $(\theta,\phi)$ that maximizes cIS. 
In the 3-class case, $\xi_c$ with the $XX$-axis and the $ZZ$-axis are calculated by the projectors $\{\dya{+,+},\dya{+,-},\dya{-,+},\dya{-,-}\}$ and $\{\dya{0,0},\dya{0,1},\dya{1,0},\dya{1,1}\}$, respectively. 
For the high-accuracy-axis, we choose $\vec{\theta}$ in Eq.~\eqref{eq:ParameterizedProjectors3class} so that the classification accuracy is almost maximized. 
For the optimized-axis, we choose $\vec{\theta}$ that maximizes cIS.
{As expected, the qIS of the unbiased case (more diverse case) is larger than that of the biased case; interestingly, the bias-unbias gap for the 3-class case is bigger than that for the 2-class case. 
Also, the importance of appropriate choice of the measurement is clearly seen; in the 2-class classification for the unbiased case, the normalized error of $\xi_q-\xi_c$ is $(1.123-1.111)/2=6.0\times10^{-3}$, where $\xi_q\leq \ell=2$ is used; also for the 3-class case, the normalized error is $(1.553-1.534)/3\approx6.3\times 10^{-3}$. 
}

\begin{table}[htp!]
\begin{tabular}{|c||c|c|}
\hline
\multicolumn{3}{|c|}{(a) 2-class phase classification} \\
\hline
Inception Scores &  
Unbiased& 
Biased\\     
\hline
$\xi_q$ & 1.123 & 1.028 \\
\hline
$\xi_c$: $X$-axis & 1.095 & 1.018\\
\hline
$\xi_c$: $Z$-axis & 1.001 & 1.001\\
\hline
$\xi_c$: high-accuracy-axis & 1.022& 1.005\\
\hline
$\xi_c$: optimized-axis & 1.111 & 1.023\\
\hline
\multicolumn{3}{|c|}{(b) 3-class phase classification} \\
\hline
Inception Scores &  
Unbiased& 
Biased\\     
\hline
$\xi_q$ & 1.553 & 1.126 \\
\hline
$\xi_c$: $XX$-axis & 1.501 & 1.110 \\
\hline
$\xi_c$: $ZZ$-axis & 1.017 & 1.006\\
\hline
$\xi_c$: high-accuracy-axis & 1.352& 1.073\\
\hline
$\xi_c$: optimized-axis & 1.534 & 1.115\\
\hline
\end{tabular}
\caption{Inception scores $\xi_q$ and $\xi_c$ for the (a) 2-class and (b) 3-class classification problems. 
cIS is calculated with several types of projection measurements. }
\label{table:2-class-IS}
\end{table}

{For the 2-class classification problem, we can see the effect of appropriate choice of the measurement axis, using the Bloch sphere representation of the states.} 
Figure~\ref{fig:bloch_hist} shows the plots of the output states of the QCNN in the Bloch sphere, in Figs.~\ref{fig:bloch_hist}(a1) and ~\ref{fig:bloch_hist}(a2) the unbiased and Figs.~\ref{fig:bloch_hist}(b1) and ~\ref{fig:bloch_hist}(b2) the biased cases. 
The purple and yellow plots are the output states corresponding to the labels $``0"$ and $``1,"$ respectively. 
{
The classification task is to design a measurement axis (a single line passing through the origin) for best separating the yellow and purple points. 
Clearly, a measurement axis on the $xy$ plane better works than the $z$ axis. 
Next, from the Figs.~\ref{fig:bloch_hist}(a1) and ~\ref{fig:bloch_hist}(b1), separating the unbiased dataset by a single line on the $xy$ plane seems harder than the case for the biased case. 
This can be seen in the histogram of the projection results of the QCNN outputs onto the $X$-axis, shown in Fig.~\ref{fig:bloch_hist}(a3) for the unbiased case and Fig.~\ref{fig:bloch_hist}(b3) the biased case, respectively; the horizontal and vertical axes are the expectation value $\ave{X}$ of the QCNN output and the number of data belonging to the bin, respectively. 
That is, in Fig.~\ref{fig:bloch_hist}(a3) there is an overlap between the two dataset, while in Fig.~\ref{fig:bloch_hist}(b3) the dataset are clearly separated, meaning that the accuracy in the biased case is higher than the unbiased case.  
However, due to the lack of diversity of the biased dataset, the resulting qIS for the unbiased case takes a higher value than that for the biased one.  
Apart from these observation, it is interesting that the dataset $\{ \rho_{\rm out}(x_i)\}$ construct a near-1D manifold while $\{\rho_{\rm out}(x_i)\}$ is 2D distributed in $(\mathbb{C}^2)^{\otimes 9}$; thus, the classifier (QCNN) $\Phi$ works so that the single line (i.e., the measurement axis) could best separate the two dataset. 
}

\begin{figure}[htp!]
\centering
\includegraphics[width=1\columnwidth]{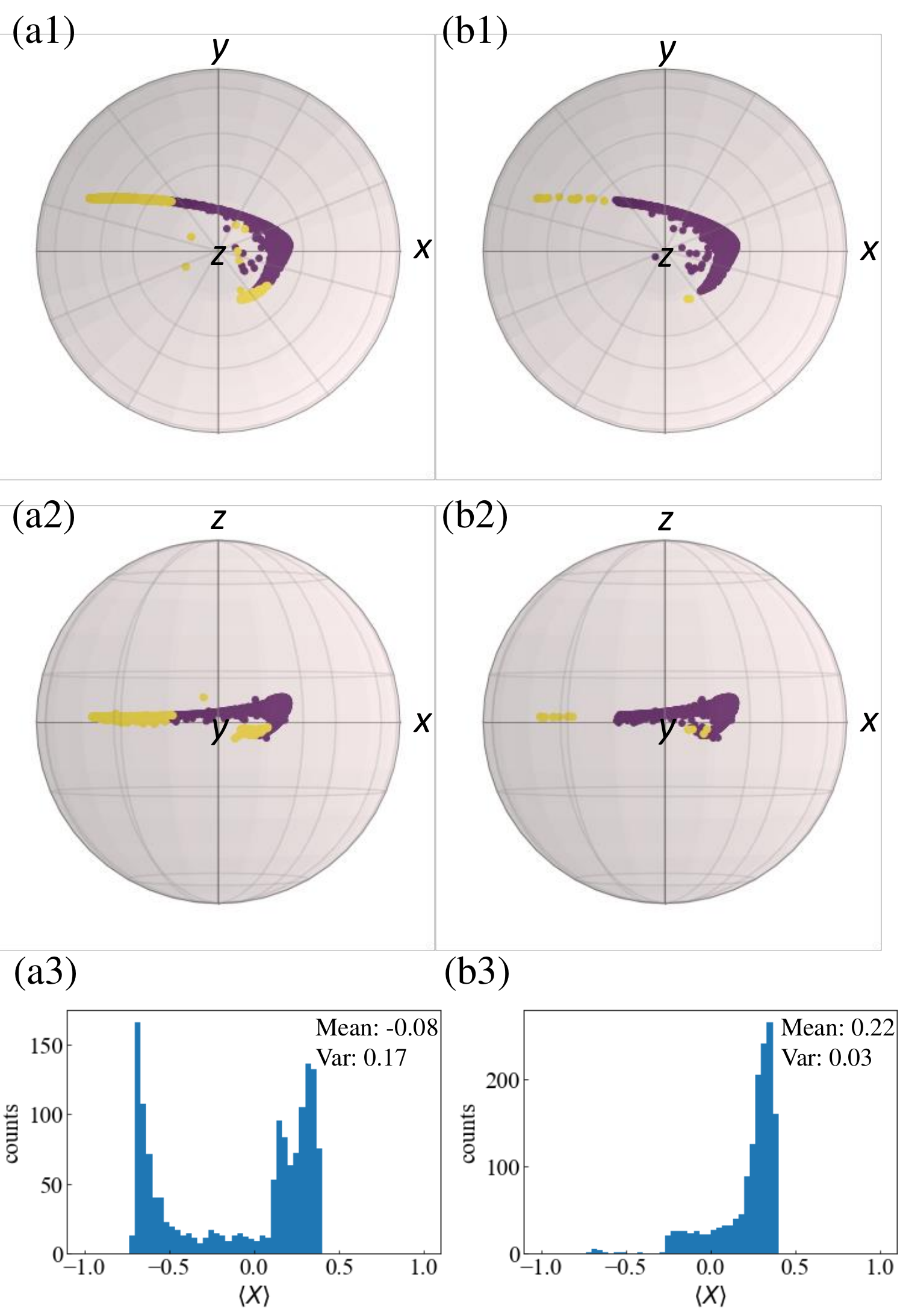}
\caption{The Bloch sphere representation and the histograms of the outputs of the QCNN; panels (a1, a2, a3) show the unbiased case and panels (b1, b2, b3) show the biased case of the 2-class phase classification problem.}
\label{fig:bloch_hist}
\end{figure}

{
Finally, Fig.~\ref{fig:2_class_phase_diagram} 
\begin{figure}[htp!]
\centering
\includegraphics[width=1\columnwidth]{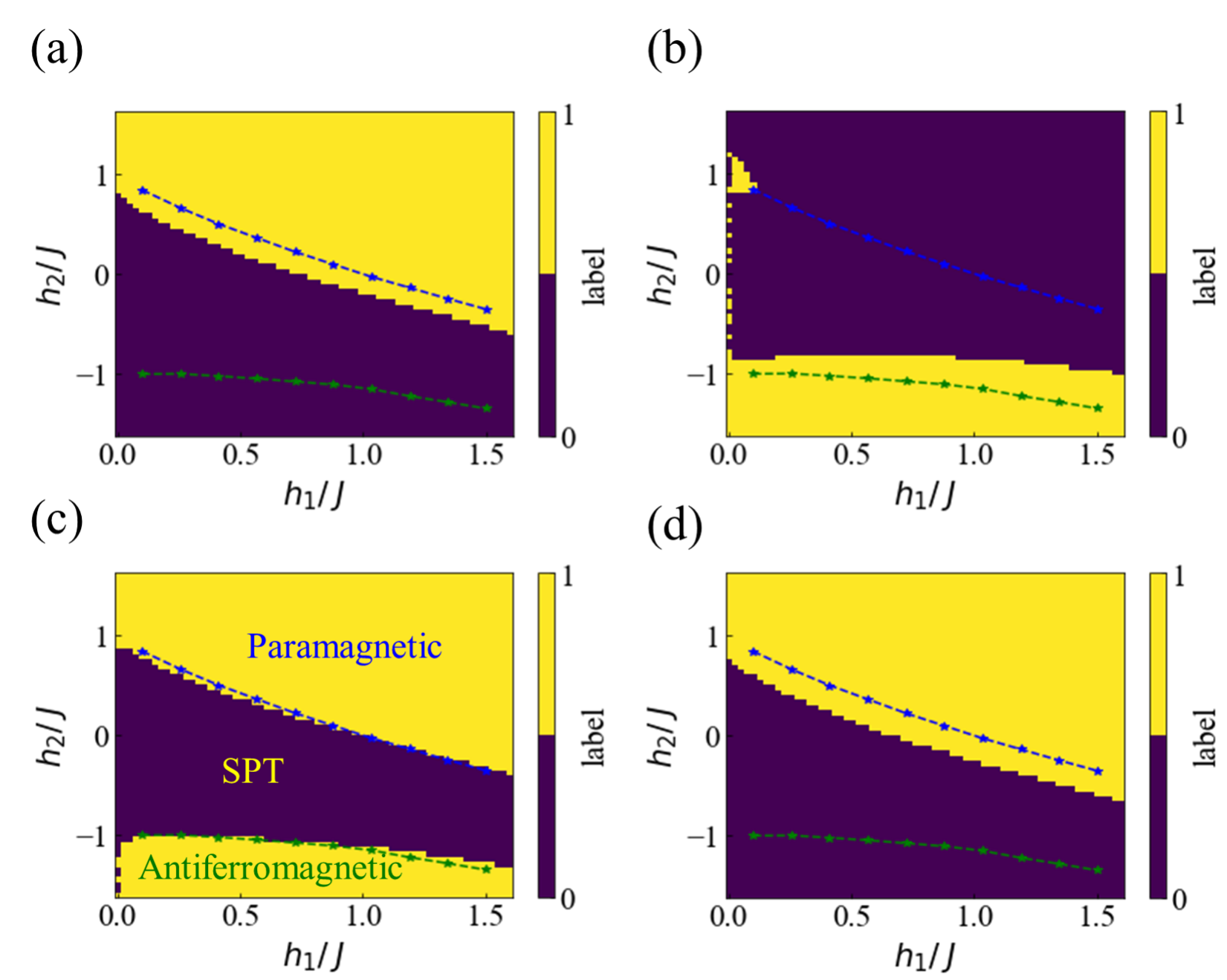}
\caption{Phase diagram predicted by the QCNN with the projection measurement onto the (a) $X$ axis, (b) $Z$ axis, (c) high-accuracy axis, and (d) optimized axis, for the 2-class classification problem. 
{
%(c) and (d) are for the unbiased case.
We use the axes (c) and (d) obtained in the unbiased case. 
}
}
\label{fig:2_class_phase_diagram}
\end{figure} 
shows the phase diagrams predicted by the trained QCNN with the projection measurement onto the (a) $X$ axis, (b) $Z$ axis, (c) high-accuracy axis, and (d) optimized axis, for the 2-class classification problem; recall that we determined the high-accuracy and optimized axis by appropriately choosing $(\theta,\phi)$ in Eq.~\eqref{eq:ParameterizedProjectors}. 
{
Regarding Fig.~\ref{fig:2_class_phase_diagram}(c) the high-accuracy axis and Fig.~\ref{fig:2_class_phase_diagram}(d) the optimized axis, we use those obtained in the unbiased case. 
}
The ground states of the Hamiltonian \eqref{eq:Hamiltonian} are classified to the paramagnetic phase (upper the blue line), the antiferromagnetic phase (below the green line), or the SPT phase between the two lines, where the blue and green lines are the exact phase boundaries obtained by using the infinite size DMRG numerical simulator. 
The purple and yellow regions correspond to the SPT phase with the label $``0"$ and the paramagnetic/antiferromagnetic phases with the label $``1,"$ respectively. 
That is, for the 2-class case, we do not distinguish the paramagnetic and antiferromagnetic phases. 
The value of classification accuracy achieved in each measurement methods are: (a) 0.78, (b) 0.49, (c) 0.98, and (d) 0.76. 
Interestingly, the optimal strategy in Fig.~\ref{fig:2_class_phase_diagram}(d) maximizes $\xi_c$ at the price of not detecting the antiferromagnetic phase, while the QCNN has the ability for detecting that phase as shown in Fig.~\ref{fig:2_class_phase_diagram}(c). 
This is consistent to the concept of IS, which quantifies the quality of the
generator, defined by the balance between the accuracy and diversity. 
}

{
Similarly, Fig.~\ref{fig:3_class_phase_diagram}
\begin{figure}[htp!]
\centering
\includegraphics[width=1\columnwidth]{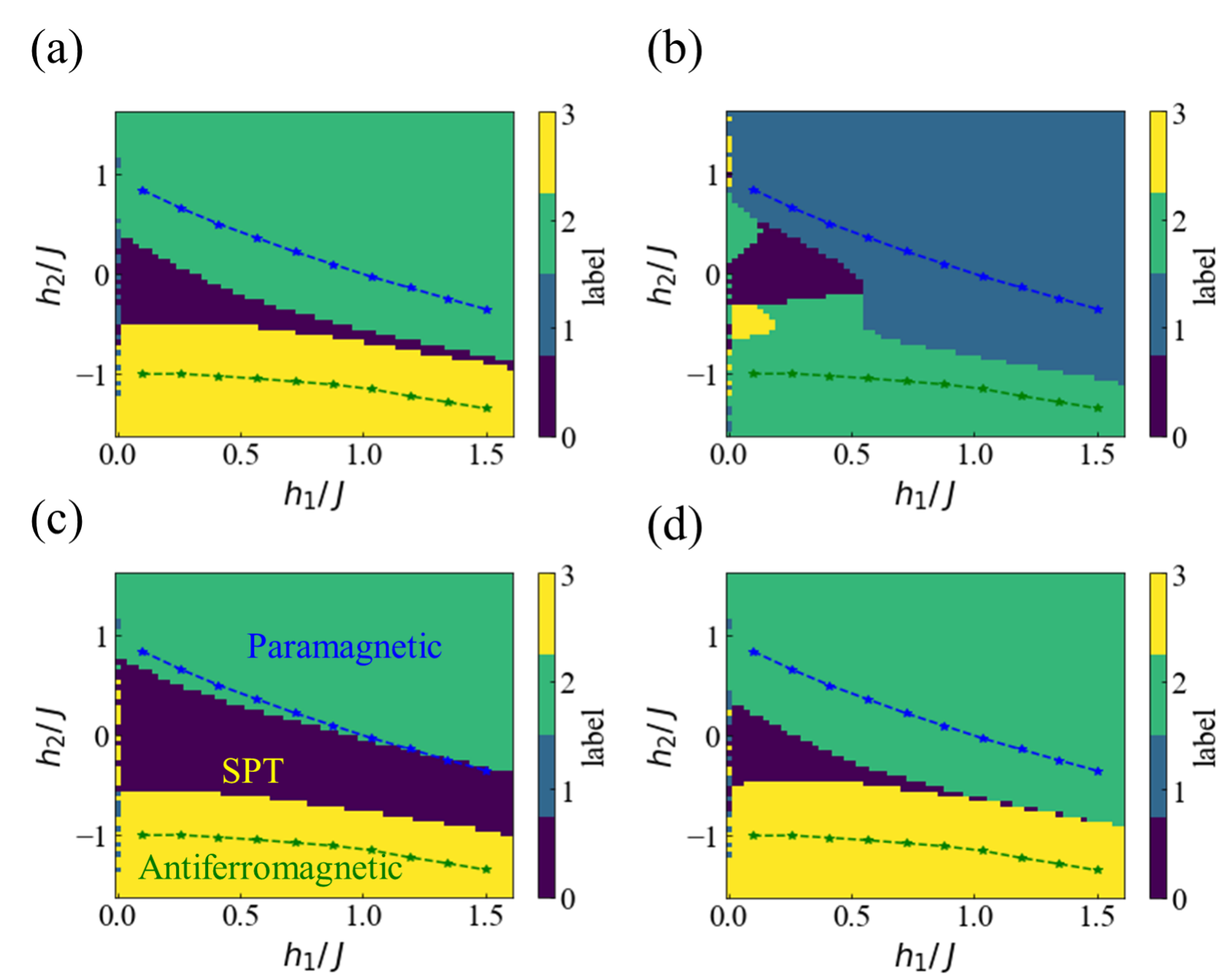}
\caption{
Phase diagram predicted by the QCNN with the projection measurement onto the (a) $XX$ axis, (b) $ZZ$ axis,  (c) high-accuracy axis, and (d) optimized axis, for the 3-class classification problem. 
{
We use the axes (c) and (d) obtained in the unbiased case. 
}
}
\label{fig:3_class_phase_diagram}
\end{figure} 
shows the phase diagrams predicted by the trained QCNN with the projection measurement onto the (a) $XX$ axis, (b) $ZZ$ axis, (c) high-accuracy axis, and (d) optimized axis, for the 3-class classification problem. 
The purple, blue, green, and yellow regions correspond to the phases of SPT, Nothing, Paramagnetic, Antiferromagnetic, with the labels $``0,"$ $``1,"$ $``2,"$ and $``3,"$ respectively. 
The value of classification accuracy in each measurement methods are (a) 0.70, (b) 0.05, (c) 0.87, and (d) 0.68. 
As in the 2-class case, the optimal measurement in Fig.~\ref{fig:3_class_phase_diagram}(d) achieves less accurate detection of the phase, compared to Fig.~\ref{fig:3_class_phase_diagram}(c) and even Fig.~\ref{fig:3_class_phase_diagram}(a). 
}

\section{Conclusion}
\label{sec:conclusion}

We have proposed the quantum inception score as a quality measure of quantum generative models, and obtained the following three main claims by connecting to the Holevo information. 
 {First, the performing measurements is one of the primary factors resulting in the quality 
difference of the quantum generative models in QQ protocol in Fig.~\ref{fig:summary} because of the destruction of the quantum coherence preserved in the quantum classifier.}
Second, the best quality can be further achieved with the entanglement of the generator's output {due to the potential superadditivity of the communication capacity. This can be regarded as the quantum advantage in the generative modeling by using the entanglement as a resource}. 
Third, the quality  {difference} of quantum generative models is due to the decoherence in the classifier, 
which can be quantified by the quantum efficacy emerging from the quantum fluctuation theorem. 
We also show examples of utilizing the quantum inception score to evaluate the quality of the 1D spin-$1/2$ chain as a generator, for the 2-class and 3-class classification of the quantum phase in the quantum many-body physics.

Finally we remark that, in the classical regime, {inception score} is not currently a widely used measure for assessing generative models because of its potential constraints in facilitating useful model comparison as pointed out in Ref.~\cite{barratt2018note}. 
These constraints could extend to the qIS in the similar manner. 
Nonetheless, the results obtained in this paper based on qIS enable us to emphasize 
the significance of further investigation of characterizing the quantum machine learning protocols from the fundamental perspectives, such as quantum information transmission and information thermodynamics.

\section*{Acknowledgement}

We are grateful to Ryuji Takagi, Haidong Yuan, Paola Cappellaro, Mingda Li, Sebastian Deffner, Marco Cerezo, Quntao Zhuang, Elton Yechao Zhu and Tharon Holdsworth for insightful and helpful discussions. A.S. is supported by the NSF under Grant No. MPS-2328774 {and Cooperative Agreement PHY-2019786}. N.Y. is supported by MEXT Quantum Leap Flagship Program Grants No. JPMXS0118067285 and No. JPMXS0120319794 and JST Grant No. JPMJPF2221.

\appendix 

\section{Derivation of Eq.~\eqref{eq:twirling}}
\label{app:twirling}
In the Appendix, we provide the proof details of Eq.~\eqref{eq:twirling}. We discuss two cases: nondegenerate and degenerate $H$. 

\subsection{Nondegenerate $H$}

First, let us focus on the case that $H$ is nondegenerate. $H$ can be diagonalized as 
\begin{align}
H = \sum_{n=1}^{D}\omega_n\dya{\omega_n}\,,
\label{eq:Hdiagonal}
\end{align}
where $D$ denotes the dimension of the Hilbert space and $\{\omega_n\}_{n=1}^{D}$ are all different from each other. Here, $\{\ket{\omega_n}\}_{n=1}^{D}$ is an orthonormal basis of $H$. Defining, 
\begin{align}
    \Delta\omega_{nm}\equiv\omega_{n}-\omega_{m},
\end{align}
we have 
\begin{align}
\begin{split}
    \frac{1}{2T}&\int_{-T}^{T}dte^{-iHt}\rho e^{iHt}\\
    =&\frac{1}{2T}\sum_{n,m}\int_{-T}^{T}dte^{-i\Delta\omega_{nm}t}\bramatket{\omega_n}{\rho}{\omega_m}\ketbra{\omega_n}{\omega_m}\\
    =&\sum_{n,m} \sinc(\Delta\omega_{nm}T)\bramatket{\omega_n}{\rho}{\omega_m}\ketbra{\omega_n}{\omega_m}\\
    =&\sum_{n=1}^{D}\bramatket{\omega_n}{\rho}{\omega_n}\dya{\omega_n}\\
    &+\sum_{n\neq m}\sinc(T\Delta\omega_{nm})\bramatket{\omega_n}{\rho}{\omega_m}\ketbra{\omega_n}{\omega_m}\,,
\end{split}
\end{align}
where we used $\sinc(0)=1$. Then, by utilizing $\lim_{T\to\infty} \sinc(T\Delta\omega_{nm})=0~(\Delta\omega_{nm}\neq 0)$, the second term vanishes when $T\to\infty$; therefore, we obtain 
\begin{align}
\begin{split}
    \mathcal{G}_H(\rho)&= \lim_{T\to\infty}\left[\frac{1}{2T}\int_{-T}^{T}dt e^{-it H}\rho e^{it H}\right]\\
    &=\sum_{n=1}^{D}\bramatket{\omega_n}{\rho}{\omega_n}\dya{\omega_n}=\sum_{n=1}^{D}\Pi_n\rho\Pi_n\,,
\end{split}
\label{eq:nondegenerate}
\end{align}
where $\Pi_n\equiv\dya{\omega_n}$ is a rank-1 projector onto the eigenstate $\ket{\omega_n}$ of $H$. 
This implies that $\mathcal{G}_H$ corresponds to the dephasing map transferring $\rho$ into a fully incoherent state~\cite{baumgratz2014quantifying} diagonal in the eigenbasis $\{\ket{\omega_n}\}_{n=1}^{D}$ of $H$. Then, the relative entropy of asymmetry $A(\rho;H)\equiv S\left(\rho\,\|\,\mathcal{G}_H(\rho)\right)$ coincides with the relative entropy of coherence~\cite{baumgratz2014quantifying}.

\subsection{Degenerate $H$}
Next, let us discuss the case that $H$ is degenerate. Again, $H$ can be diagonalized as Eq.~\eqref{eq:Hdiagonal}. Suppose that $H$ has $\lambda$ degenerate eigenvalues 
\begin{align}
\omega_{\alpha_1},\omega_{\alpha_2},\cdots,\omega_{\alpha_{\lambda}}\,.
\end{align}
For each degenerate eigenvalue $\omega_{\alpha_{\mu}}$, suppose that we have $k_{\mu}$ orthonormal eigenbasis 
\begin{align}
\{\ket{\omega_{\alpha_{\mu}}^{(1)}},\ket{\omega_{\alpha_{\mu}}^{(2)}},\cdots,\ket{\omega_{\alpha_{\mu}}^{(k_{\mu})}}\}\,.
\end{align}
Defining 
\begin{align}
\overline{K}\equiv\sum_{\mu=1}^{\lambda}k_{\mu}\,,
\end{align}
the number of distinct nondegenerate eigenvalues is given by 
\begin{align}
K=D-\overline{K}\,.
\end{align}
Then, $L=K+\lambda <D$ is the total number of distinct eigenvalues of $H$, and the Hilbert space can be written as 
\begin{align}
    \mathcal{H} = \bigoplus_{\alpha=1}^{L}\mathcal{H}_{\alpha}\,,
\end{align}
where $\mathcal{H}_{\alpha}$ denotes the subspace spanned by the eigenstates of $H$ corresponding to the eigenvalue $\omega_{\alpha}$. Here, for the nondegenerate eigenvalue $\omega_{\alpha}$, the corresponding projector is the rank-1 projector 
\begin{align}
\Pi_{\omega_{\alpha}}\equiv\dya{\omega_{\alpha}}\,.
\end{align}
For the degenerate eigenvalue $\omega_{\alpha_{\mu}}$, the projector onto the corresponding subspace of the eigenbasis is 
\begin{align}
\Pi_{\alpha_{\mu}}=\sum_{\nu=1}^{k_{\mu}}\dya{\omega_{\alpha_{\mu}}^{(\nu)}}\,,
\end{align}
so that the rank of the projector is $\rank(\Pi_{\alpha_{\mu}})=k_{\mu}>1$. 
Now, let us define the set of the nondegenerate eigenvalues as 
\begin{align}
\Omega\equiv\{\omega_{\alpha}\,|\,\omega_{\alpha}\neq\omega_{\beta},~\alpha\neq\beta\}
\end{align}
and the set of degenerate eigenvalues as
\begin{align}
    \Gamma \equiv \{\omega_{\alpha}\,|\, \omega_{\alpha}=\omega_{\beta}, \alpha\neq\beta\}\,. 
\end{align}
Then, we have 
\begin{align}
\begin{split}
    \frac{1}{2T}&\int_{-T}^{T}dte^{-iHt}\rho e^{iHt}\\
    =&\sum_{n,m} \sinc(\Delta\omega_{nm}T)\bramatket{\omega_n}{\rho}{\omega_m}\ketbra{\omega_n}{\omega_m}\\
     =&\sum_{\omega_{\alpha}\in\Omega}\bramatket{\omega_{\alpha}}{\rho}{\omega_{\alpha}}\dya{\omega_{\alpha}}
     +\sum_{\omega_{\alpha},\omega_{\beta}\in\Gamma}\bramatket{\omega_{\alpha}}{\rho}{\omega_{\beta}}\ketbra{\omega_{\alpha}}{\omega_{\beta}}\\
     &+\sum_{\omega_{\alpha},\omega_{\beta}\in\Omega}\sinc(T\Delta\omega_{\alpha\beta})\bramatket{\omega_{\alpha}}{\rho}{\omega_{\beta}}\ketbra{\omega_{\alpha}}{\omega_{\beta}}\\
     =&\sum_{\omega_{\alpha}\in\Omega}\bramatket{\omega_{\alpha}}{\rho}{\omega_{\alpha}}\dya{\omega_{\alpha}}\\
     &+\sum_{\mu=1}^{\lambda}\sum_{\nu=1}^{k_{\mu}}\sum_{\nu'=1}^{k_{\mu}}\bramatket{\omega_{\alpha_{\mu}}^{(\nu)}}{\rho}{\omega_{\alpha_{\mu}}^{(\nu')}}\ketbra{\omega_{\alpha_{\mu}}^{(\nu)}}{\omega_{\alpha_{\mu}}^{(\nu')}}\\
     &+\sum_{\omega_{\alpha},\omega_{\beta}\in\Omega}\sinc(T\Delta\omega_{\alpha\beta})\bramatket{\omega_{\alpha}}{\rho}{\omega_{\beta}}\ketbra{\omega_{\alpha}}{\omega_{\beta}}\\
     =&\sum_{\omega_{\alpha}\in\Omega}\Pi_{\omega_{\alpha}}\rho\Pi_{\omega_{\alpha}}+\sum_{\mu=1}^{\lambda}\Pi_{\alpha_{\mu}}\rho\Pi_{\alpha_{\mu}}\\
     &+\sum_{\omega_{\alpha},\omega_{\beta}\in\Omega}\sinc(T\Delta\omega_{\alpha\beta})\bramatket{\omega_{\alpha}}{\rho}{\omega_{\beta}}\ketbra{\omega_{\alpha}}{\omega_{\beta}}\,,
\end{split}
\end{align}
where the third term vanishes as $T\to\infty$. 
Therefore, defining the set of projectors 
\begin{align}
\{\Pi_{n}\}_{n=1}^{L}\in\{\Pi_{\omega_{\alpha}}\}_{\omega_{\alpha}\in\Omega}\cup\{\Pi_{\alpha_{\mu}}\}_{\mu=1}^{\lambda}
\end{align}
onto the subspace of the eigenbasis of $H$, we obtain   
\begin{align}
\mathcal{G}_H(\rho)= \lim_{T\to\infty}\left[\frac{1}{2T}\int_{-T}^{T}dt e^{-it H}\rho e^{it H}\right]=\sum_{n=1}^{L}\Pi_{n}\rho \Pi_n\,,
\label{eq:degenerate}
\end{align}
which is a block diagonal state. In this case, the relative entropy of asymmetry $A(\rho;H)\equiv S\left(\rho\,\|\,\mathcal{G}_H(\rho)\right)$ coincides with the relative entropy of superposition~\cite{aberg2006quantifying}. Obviously, Eq.~\eqref{eq:degenerate} is reduced to Eq.~\eqref{eq:nondegenerate} when we have $\rank(\Pi_{\alpha_{\mu}})=1~(\forall \mu)$, which corresponds to the case that $H$ is nondegenerate.

\bibliography{ref.bib}

%apsrev4-2.bst 2019-01-14 (MD) hand-edited version of apsrev4-1.bst
%Control: key (0)
%Control: author (8) initials jnrlst
%Control: editor formatted (1) identically to author
%Control: production of article title (0) allowed
%Control: page (0) single
%Control: year (1) truncated
%Control: production of eprint (0) enabled
\begin{thebibliography}{100}%
\makeatletter
\providecommand \@ifxundefined [1]{%
 \@ifx{#1\undefined}
}%
\providecommand \@ifnum [1]{%
 \ifnum #1\expandafter \@firstoftwo
 \else \expandafter \@secondoftwo
 \fi
}%
\providecommand \@ifx [1]{%
 \ifx #1\expandafter \@firstoftwo
 \else \expandafter \@secondoftwo
 \fi
}%
\providecommand \natexlab [1]{#1}%
\providecommand \enquote  [1]{``#1''}%
\providecommand \bibnamefont  [1]{#1}%
\providecommand \bibfnamefont [1]{#1}%
\providecommand \citenamefont [1]{#1}%
\providecommand \href@noop [0]{\@secondoftwo}%
\providecommand \href [0]{\begingroup \@sanitize@url \@href}%
\providecommand \@href[1]{\@@startlink{#1}\@@href}%
\providecommand \@@href[1]{\endgroup#1\@@endlink}%
\providecommand \@sanitize@url [0]{\catcode `\\12\catcode `\$12\catcode
  `\&12\catcode `\#12\catcode `\^12\catcode `\_12\catcode `\%12\relax}%
\providecommand \@@startlink[1]{}%
\providecommand \@@endlink[0]{}%
\providecommand \url  [0]{\begingroup\@sanitize@url \@url }%
\providecommand \@url [1]{\endgroup\@href {#1}{\urlprefix }}%
\providecommand \urlprefix  [0]{URL }%
\providecommand \Eprint [0]{\href }%
\providecommand \doibase [0]{https://doi.org/}%
\providecommand \selectlanguage [0]{\@gobble}%
\providecommand \bibinfo  [0]{\@secondoftwo}%
\providecommand \bibfield  [0]{\@secondoftwo}%
\providecommand \translation [1]{[#1]}%
\providecommand \BibitemOpen [0]{}%
\providecommand \bibitemStop [0]{}%
\providecommand \bibitemNoStop [0]{.\EOS\space}%
\providecommand \EOS [0]{\spacefactor3000\relax}%
\providecommand \BibitemShut  [1]{\csname bibitem#1\endcsname}%
\let\auto@bib@innerbib\@empty
%</preamble>
\bibitem [{\citenamefont {Goodfellow}\ \emph {et~al.}(2016)\citenamefont
  {Goodfellow}, \citenamefont {Bengio},\ and\ \citenamefont
  {Courville}}]{goodfellow2016deep}%
  \BibitemOpen
  \bibfield  {author} {\bibinfo {author} {\bibfnamefont {I.}~\bibnamefont
  {Goodfellow}}, \bibinfo {author} {\bibfnamefont {Y.}~\bibnamefont {Bengio}},\
  and\ \bibinfo {author} {\bibfnamefont {A.}~\bibnamefont {Courville}},\
  }\href@noop {} {\emph {\bibinfo {title} {Deep {L}earning}}}\ (\bibinfo
  {publisher} {MIT press},\ \bibinfo {year} {2016})\BibitemShut {NoStop}%
\bibitem [{\citenamefont {Aggarwal}(2023)}]{aggarwal2018neural}%
  \BibitemOpen
  \bibfield  {author} {\bibinfo {author} {\bibfnamefont {C.~C.}\ \bibnamefont
  {Aggarwal}},\ }\href@noop {} {\emph {\bibinfo {title} {{N}eural {N}etworks
  and {D}eep {L}earning, 2nd ed.}}}\ (\bibinfo  {publisher} {Springer
  International Publishing},\ \bibinfo {year} {2023})\BibitemShut {NoStop}%
\bibitem [{\citenamefont {Finlayson}\ \emph {et~al.}(2019)\citenamefont
  {Finlayson}, \citenamefont {Bowers}, \citenamefont {Ito}, \citenamefont
  {Zittrain}, \citenamefont {Beam},\ and\ \citenamefont
  {Kohane}}]{finlayson2019adversarial}%
  \BibitemOpen
  \bibfield  {author} {\bibinfo {author} {\bibfnamefont {S.~G.}\ \bibnamefont
  {Finlayson}}, \bibinfo {author} {\bibfnamefont {J.~D.}\ \bibnamefont
  {Bowers}}, \bibinfo {author} {\bibfnamefont {J.}~\bibnamefont {Ito}},
  \bibinfo {author} {\bibfnamefont {J.~L.}\ \bibnamefont {Zittrain}}, \bibinfo
  {author} {\bibfnamefont {A.~L.}\ \bibnamefont {Beam}},\ and\ \bibinfo
  {author} {\bibfnamefont {I.~S.}\ \bibnamefont {Kohane}},\ }\bibfield  {title}
  {\bibinfo {title} {Adversarial attacks on medical machine learning},\ }\href
  {https://doi.org/10.1126/science.aaw4399} {\bibfield  {journal} {\bibinfo
  {journal} {Science}\ }\textbf {\bibinfo {volume} {363}},\ \bibinfo {pages}
  {1287} (\bibinfo {year} {2019})}\BibitemShut {NoStop}%
\bibitem [{\citenamefont {Biggio}\ and\ \citenamefont
  {Roli}(2018)}]{biggio2018wild}%
  \BibitemOpen
  \bibfield  {author} {\bibinfo {author} {\bibfnamefont {B.}~\bibnamefont
  {Biggio}}\ and\ \bibinfo {author} {\bibfnamefont {F.}~\bibnamefont {Roli}},\
  }\bibfield  {title} {\bibinfo {title} {Wild patterns: Ten years after the
  rise of adversarial machine learning},\ }in\ \href
  {https://doi.org/10.1145/3243734.3264418} {\emph {\bibinfo {booktitle} {Proc.
  ACM Conf. Comput. Commun. Secur.}}}\ (\bibinfo {year} {2018})\ pp.\ \bibinfo
  {pages} {2154 -- 2156}\BibitemShut {NoStop}%
\bibitem [{\citenamefont {Tanaka}\ \emph {et~al.}(2021)\citenamefont {Tanaka},
  \citenamefont {Tomiya},\ and\ \citenamefont {Hashimoto}}]{tanaka2021deep}%
  \BibitemOpen
  \bibfield  {author} {\bibinfo {author} {\bibfnamefont {A.}~\bibnamefont
  {Tanaka}}, \bibinfo {author} {\bibfnamefont {A.}~\bibnamefont {Tomiya}},\
  and\ \bibinfo {author} {\bibfnamefont {K.}~\bibnamefont {Hashimoto}},\
  }\href@noop {} {\emph {\bibinfo {title} {Deep {L}earning and {P}hysics}}}\
  (\bibinfo  {publisher} {Springer},\ \bibinfo {year} {2021})\BibitemShut
  {NoStop}%
\bibitem [{\citenamefont {Sarma}\ \emph {et~al.}(2019)\citenamefont {Sarma},
  \citenamefont {Deng},\ and\ \citenamefont {Duan}}]{sarma2019machine}%
  \BibitemOpen
  \bibfield  {author} {\bibinfo {author} {\bibfnamefont {S.~D.}\ \bibnamefont
  {Sarma}}, \bibinfo {author} {\bibfnamefont {D.-L.}\ \bibnamefont {Deng}},\
  and\ \bibinfo {author} {\bibfnamefont {L.-M.}\ \bibnamefont {Duan}},\
  }\bibfield  {title} {\bibinfo {title} {Machine learning meets quantum
  physics},\ }\href {https://doi.org/10.1063/PT.3.4164} {\bibfield  {journal}
  {\bibinfo  {journal} {Phys. Today}\ }\textbf {\bibinfo {volume} {72}},\
  \bibinfo {pages} {48} (\bibinfo {year} {2019})}\BibitemShut {NoStop}%
\bibitem [{\citenamefont {Nielsen}\ and\ \citenamefont
  {Chuang}(2010)}]{Nielsen}%
  \BibitemOpen
  \bibfield  {author} {\bibinfo {author} {\bibfnamefont {M.~A.}\ \bibnamefont
  {Nielsen}}\ and\ \bibinfo {author} {\bibfnamefont {I.~L.}\ \bibnamefont
  {Chuang}},\ }\href@noop {} {\emph {\bibinfo {title} {Quantum Computation and
  Quantum Information: 10th Anniversary Edition}}},\ \bibinfo {edition} {10th}\
  ed.\ (\bibinfo  {publisher} {Cambridge University Press},\ \bibinfo {address}
  {New York, NY, USA},\ \bibinfo {year} {2010})\BibitemShut {NoStop}%
\bibitem [{\citenamefont {Sch{\"u}tt}\ \emph {et~al.}(2020)\citenamefont
  {Sch{\"u}tt}, \citenamefont {Chmiela}, \citenamefont {von Lilienfeld},
  \citenamefont {Tkatchenko}, \citenamefont {Tsuda},\ and\ \citenamefont
  {M{\"u}ller}}]{schutt2020machine}%
  \BibitemOpen
  \bibfield  {author} {\bibinfo {author} {\bibfnamefont {K.~T.}\ \bibnamefont
  {Sch{\"u}tt}}, \bibinfo {author} {\bibfnamefont {S.}~\bibnamefont {Chmiela}},
  \bibinfo {author} {\bibfnamefont {O.~A.}\ \bibnamefont {von Lilienfeld}},
  \bibinfo {author} {\bibfnamefont {A.}~\bibnamefont {Tkatchenko}}, \bibinfo
  {author} {\bibfnamefont {K.}~\bibnamefont {Tsuda}},\ and\ \bibinfo {author}
  {\bibfnamefont {K.-R.}\ \bibnamefont {M{\"u}ller}},\ }\bibfield  {title}
  {\bibinfo {title} {Machine learning meets quantum physics},\ }\href@noop {}
  {\bibfield  {journal} {\bibinfo  {journal} {Lecture Notes in Physics}\ }
  (\bibinfo {year} {2020})}\BibitemShut {NoStop}%
\bibitem [{\citenamefont {Biamonte}\ \emph {et~al.}(2017)\citenamefont
  {Biamonte}, \citenamefont {Wittek}, \citenamefont {Pancotti}, \citenamefont
  {Rebentrost}, \citenamefont {Wiebe},\ and\ \citenamefont
  {Lloyd}}]{biamonte2017quantum}%
  \BibitemOpen
  \bibfield  {author} {\bibinfo {author} {\bibfnamefont {J.}~\bibnamefont
  {Biamonte}}, \bibinfo {author} {\bibfnamefont {P.}~\bibnamefont {Wittek}},
  \bibinfo {author} {\bibfnamefont {N.}~\bibnamefont {Pancotti}}, \bibinfo
  {author} {\bibfnamefont {P.}~\bibnamefont {Rebentrost}}, \bibinfo {author}
  {\bibfnamefont {N.}~\bibnamefont {Wiebe}},\ and\ \bibinfo {author}
  {\bibfnamefont {S.}~\bibnamefont {Lloyd}},\ }\bibfield  {title} {\bibinfo
  {title} {Quantum machine learning},\ }\href
  {https://doi.org/10.1038/nature23474} {\bibfield  {journal} {\bibinfo
  {journal} {Nature}\ }\textbf {\bibinfo {volume} {549}},\ \bibinfo {pages}
  {195} (\bibinfo {year} {2017})}\BibitemShut {NoStop}%
\bibitem [{\citenamefont {Jerbi}\ \emph {et~al.}(2023)\citenamefont {Jerbi},
  \citenamefont {Fiderer}, \citenamefont {Poulsen~Nautrup}, \citenamefont
  {K{\"u}bler}, \citenamefont {Briegel},\ and\ \citenamefont
  {Dunjko}}]{jerbi2023quantum}%
  \BibitemOpen
  \bibfield  {author} {\bibinfo {author} {\bibfnamefont {S.}~\bibnamefont
  {Jerbi}}, \bibinfo {author} {\bibfnamefont {L.~J.}\ \bibnamefont {Fiderer}},
  \bibinfo {author} {\bibfnamefont {H.}~\bibnamefont {Poulsen~Nautrup}},
  \bibinfo {author} {\bibfnamefont {J.~M.}\ \bibnamefont {K{\"u}bler}},
  \bibinfo {author} {\bibfnamefont {H.~J.}\ \bibnamefont {Briegel}},\ and\
  \bibinfo {author} {\bibfnamefont {V.}~\bibnamefont {Dunjko}},\ }\bibfield
  {title} {\bibinfo {title} {Quantum machine learning beyond kernel methods},\
  }\href {https://doi.org/10.1038/s41467-023-36159-y} {\bibfield  {journal}
  {\bibinfo  {journal} {Nat. Commun.}\ }\textbf {\bibinfo {volume} {14}},\
  \bibinfo {pages} {517} (\bibinfo {year} {2023})}\BibitemShut {NoStop}%
\bibitem [{\citenamefont {Torlai}\ and\ \citenamefont
  {Melko}(2020)}]{torlai2020machine}%
  \BibitemOpen
  \bibfield  {author} {\bibinfo {author} {\bibfnamefont {G.}~\bibnamefont
  {Torlai}}\ and\ \bibinfo {author} {\bibfnamefont {R.~G.}\ \bibnamefont
  {Melko}},\ }\bibfield  {title} {\bibinfo {title} {Machine-learning quantum
  states in the {NISQ} era},\ }\href
  {https://doi.org/10.1146/annurev-conmatphys-031119-050651} {\bibfield
  {journal} {\bibinfo  {journal} {Annu. Rev. Condens. Matter Phys.}\ }\textbf
  {\bibinfo {volume} {11}},\ \bibinfo {pages} {325} (\bibinfo {year}
  {2020})}\BibitemShut {NoStop}%
\bibitem [{\citenamefont {Schuld}\ and\ \citenamefont
  {Petruccione}(2021)}]{schuld2021machine}%
  \BibitemOpen
  \bibfield  {author} {\bibinfo {author} {\bibfnamefont {M.}~\bibnamefont
  {Schuld}}\ and\ \bibinfo {author} {\bibfnamefont {F.}~\bibnamefont
  {Petruccione}},\ }\href@noop {} {\emph {\bibinfo {title} {Machine {L}earning
  with {Q}uantum {C}omputers}}}\ (\bibinfo  {publisher} {Springer},\ \bibinfo
  {year} {2021})\BibitemShut {NoStop}%
\bibitem [{\citenamefont {Gili}\ \emph {et~al.}(2024)\citenamefont {Gili},
  \citenamefont {Mauri},\ and\ \citenamefont
  {Perdomo-Ortiz}}]{gili2022generalization}%
  \BibitemOpen
  \bibfield  {author} {\bibinfo {author} {\bibfnamefont {K.}~\bibnamefont
  {Gili}}, \bibinfo {author} {\bibfnamefont {M.}~\bibnamefont {Mauri}},\ and\
  \bibinfo {author} {\bibfnamefont {A.}~\bibnamefont {Perdomo-Ortiz}},\
  }\bibfield  {title} {\bibinfo {title} {Generalization metrics for practical
  quantum advantage in generative models},\ }\href
  {https://doi.org/10.1103/PhysRevApplied.21.044032} {\bibfield  {journal}
  {\bibinfo  {journal} {Phys. Rev. Applied}\ }\textbf {\bibinfo {volume}
  {21}},\ \bibinfo {pages} {044032} (\bibinfo {year} {2024})}\BibitemShut
  {NoStop}%
\bibitem [{\citenamefont {Turhan}\ and\ \citenamefont
  {Bilge}(2018)}]{turhan2018recent}%
  \BibitemOpen
  \bibfield  {author} {\bibinfo {author} {\bibfnamefont {C.~G.}\ \bibnamefont
  {Turhan}}\ and\ \bibinfo {author} {\bibfnamefont {H.~S.}\ \bibnamefont
  {Bilge}},\ }\bibfield  {title} {\bibinfo {title} {Recent trends in deep
  generative models: a review},\ }in\ \href
  {https://doi.org/10.1109/UBMK.2018.8566353} {\emph {\bibinfo {booktitle}
  {2018 3rd International Conference on Computer Science and Engineering
  (UBMK)}}}\ (\bibinfo {organization} {IEEE},\ \bibinfo {year} {2018})\ pp.\
  \bibinfo {pages} {574--579}\BibitemShut {NoStop}%
\bibitem [{\citenamefont {Xu}\ \emph {et~al.}(2015)\citenamefont {Xu},
  \citenamefont {Li},\ and\ \citenamefont {Zhou}}]{xu2015overview}%
  \BibitemOpen
  \bibfield  {author} {\bibinfo {author} {\bibfnamefont {J.}~\bibnamefont
  {Xu}}, \bibinfo {author} {\bibfnamefont {H.}~\bibnamefont {Li}},\ and\
  \bibinfo {author} {\bibfnamefont {S.}~\bibnamefont {Zhou}},\ }\bibfield
  {title} {\bibinfo {title} {An overview of deep generative models},\ }\href
  {https://doi.org/10.1080/02564602.2014.987328} {\bibfield  {journal}
  {\bibinfo  {journal} {IETE Tech. Rev.}\ }\textbf {\bibinfo {volume} {32}},\
  \bibinfo {pages} {131} (\bibinfo {year} {2015})}\BibitemShut {NoStop}%
\bibitem [{\citenamefont {Barlow}(1989)}]{barlow1989unsupervised}%
  \BibitemOpen
  \bibfield  {author} {\bibinfo {author} {\bibfnamefont {H.~B.}\ \bibnamefont
  {Barlow}},\ }\bibfield  {title} {\bibinfo {title} {Unsupervised learning},\
  }\href@noop {} {\bibfield  {journal} {\bibinfo  {journal} {Neural
  computation}\ }\textbf {\bibinfo {volume} {1}},\ \bibinfo {pages} {295}
  (\bibinfo {year} {1989})}\BibitemShut {NoStop}%
\bibitem [{\citenamefont {Dike}\ \emph {et~al.}(2018)\citenamefont {Dike},
  \citenamefont {Zhou}, \citenamefont {Deveerasetty},\ and\ \citenamefont
  {Wu}}]{dike2018unsupervised}%
  \BibitemOpen
  \bibfield  {author} {\bibinfo {author} {\bibfnamefont {H.~U.}\ \bibnamefont
  {Dike}}, \bibinfo {author} {\bibfnamefont {Y.}~\bibnamefont {Zhou}}, \bibinfo
  {author} {\bibfnamefont {K.~K.}\ \bibnamefont {Deveerasetty}},\ and\ \bibinfo
  {author} {\bibfnamefont {Q.}~\bibnamefont {Wu}},\ }\bibfield  {title}
  {\bibinfo {title} {Unsupervised learning based on artificial neural network:
  A review},\ }in\ \href {https://doi.org/10.1109/CBS.2018.8612259} {\emph
  {\bibinfo {booktitle} {2018 IEEE International Conference on Cyborg and
  Bionic Systems (CBS)}}}\ (\bibinfo {organization} {IEEE},\ \bibinfo {year}
  {2018})\ pp.\ \bibinfo {pages} {322--327}\BibitemShut {NoStop}%
\bibitem [{\citenamefont {Gao}\ \emph {et~al.}(2018)\citenamefont {Gao},
  \citenamefont {Zhang},\ and\ \citenamefont {Duan}}]{gao2018quantum}%
  \BibitemOpen
  \bibfield  {author} {\bibinfo {author} {\bibfnamefont {X.}~\bibnamefont
  {Gao}}, \bibinfo {author} {\bibfnamefont {Z.-Y.}\ \bibnamefont {Zhang}},\
  and\ \bibinfo {author} {\bibfnamefont {L.-M.}\ \bibnamefont {Duan}},\
  }\bibfield  {title} {\bibinfo {title} {A quantum machine learning algorithm
  based on generative models},\ }\href {https://doi.org/10.1126/sciadv.aat9004}
  {\bibfield  {journal} {\bibinfo  {journal} {Sci. Adv.}\ }\textbf {\bibinfo
  {volume} {4}},\ \bibinfo {pages} {eaat9004} (\bibinfo {year}
  {2018})}\BibitemShut {NoStop}%
\bibitem [{\citenamefont {Gao}\ \emph {et~al.}(2022)\citenamefont {Gao},
  \citenamefont {Anschuetz}, \citenamefont {Wang}, \citenamefont {Cirac},\ and\
  \citenamefont {Lukin}}]{gao2022enhancing}%
  \BibitemOpen
  \bibfield  {author} {\bibinfo {author} {\bibfnamefont {X.}~\bibnamefont
  {Gao}}, \bibinfo {author} {\bibfnamefont {E.~R.}\ \bibnamefont {Anschuetz}},
  \bibinfo {author} {\bibfnamefont {S.-T.}\ \bibnamefont {Wang}}, \bibinfo
  {author} {\bibfnamefont {J.~I.}\ \bibnamefont {Cirac}},\ and\ \bibinfo
  {author} {\bibfnamefont {M.~D.}\ \bibnamefont {Lukin}},\ }\bibfield  {title}
  {\bibinfo {title} {Enhancing generative models via quantum correlations},\
  }\href {https://doi.org/10.1103/PhysRevX.12.021037} {\bibfield  {journal}
  {\bibinfo  {journal} {Phys. Rev. X}\ }\textbf {\bibinfo {volume} {12}},\
  \bibinfo {pages} {021037} (\bibinfo {year} {2022})}\BibitemShut {NoStop}%
\bibitem [{\citenamefont {Tezuka}\ \emph {et~al.}(2024)\citenamefont {Tezuka},
  \citenamefont {Uno},\ and\ \citenamefont {Yamamoto}}]{tezuka2024generative}%
  \BibitemOpen
  \bibfield  {author} {\bibinfo {author} {\bibfnamefont {H.}~\bibnamefont
  {Tezuka}}, \bibinfo {author} {\bibfnamefont {S.}~\bibnamefont {Uno}},\ and\
  \bibinfo {author} {\bibfnamefont {N.}~\bibnamefont {Yamamoto}},\ }\bibfield
  {title} {\bibinfo {title} {Generative model for learning quantum ensemble
  with optimal transport loss},\ }\href
  {https://doi.org/10.1007/s42484-024-00142-7} {\bibfield  {journal} {\bibinfo
  {journal} {Quantum Mach. Intell.}\ }\textbf {\bibinfo {volume} {6}},\
  \bibinfo {pages} {6} (\bibinfo {year} {2024})}\BibitemShut {NoStop}%
\bibitem [{\citenamefont {Kingma}\ and\ \citenamefont
  {Welling}(2019)}]{Kingma_Book_Autoencoder}%
  \BibitemOpen
  \bibfield  {author} {\bibinfo {author} {\bibfnamefont {D.~P.}\ \bibnamefont
  {Kingma}}\ and\ \bibinfo {author} {\bibfnamefont {M.}~\bibnamefont
  {Welling}},\ }\bibfield  {title} {\bibinfo {title} {{A}n {I}ntroduction to
  {V}ariational {A}utoencoders},\ }\href {https://doi.org/10.1561/2200000056}
  {\bibfield  {journal} {\bibinfo  {journal} {Found. Trends Mach. Learn.}\
  }\textbf {\bibinfo {volume} {12}},\ \bibinfo {pages} {307} (\bibinfo {year}
  {2019})}\BibitemShut {NoStop}%
\bibitem [{\citenamefont {Goodfellow}\ \emph {et~al.}(2020)\citenamefont
  {Goodfellow}, \citenamefont {Pouget-Abadie}, \citenamefont {Mirza},
  \citenamefont {Xu}, \citenamefont {Warde-Farley}, \citenamefont {Ozair},
  \citenamefont {Courville},\ and\ \citenamefont
  {Bengio}}]{goodfellow2020generative}%
  \BibitemOpen
  \bibfield  {author} {\bibinfo {author} {\bibfnamefont {I.}~\bibnamefont
  {Goodfellow}}, \bibinfo {author} {\bibfnamefont {J.}~\bibnamefont
  {Pouget-Abadie}}, \bibinfo {author} {\bibfnamefont {M.}~\bibnamefont
  {Mirza}}, \bibinfo {author} {\bibfnamefont {B.}~\bibnamefont {Xu}}, \bibinfo
  {author} {\bibfnamefont {D.}~\bibnamefont {Warde-Farley}}, \bibinfo {author}
  {\bibfnamefont {S.}~\bibnamefont {Ozair}}, \bibinfo {author} {\bibfnamefont
  {A.}~\bibnamefont {Courville}},\ and\ \bibinfo {author} {\bibfnamefont
  {Y.}~\bibnamefont {Bengio}},\ }\bibfield  {title} {\bibinfo {title}
  {Generative adversarial networks},\ }\href {https://doi.org/10.1145/3422622}
  {\bibfield  {journal} {\bibinfo  {journal} {Commun. ACM.}\ }\textbf {\bibinfo
  {volume} {63}},\ \bibinfo {pages} {139} (\bibinfo {year} {2020})}\BibitemShut
  {NoStop}%
\bibitem [{\citenamefont {Romero}\ \emph {et~al.}(2017)\citenamefont {Romero},
  \citenamefont {Olson},\ and\ \citenamefont {Aspuru-Guzik}}]{Romero17}%
  \BibitemOpen
  \bibfield  {author} {\bibinfo {author} {\bibfnamefont {J.}~\bibnamefont
  {Romero}}, \bibinfo {author} {\bibfnamefont {J.~P.}\ \bibnamefont {Olson}},\
  and\ \bibinfo {author} {\bibfnamefont {A.}~\bibnamefont {Aspuru-Guzik}},\
  }\bibfield  {title} {\bibinfo {title} {{Quantum autoencoders for efficient
  compression of quantum data}},\ }\href
  {https://doi.org/10.1088/2058-9565/aa8072} {\bibfield  {journal} {\bibinfo
  {journal} {Quantum Sci. Technol.}\ }\textbf {\bibinfo {volume} {2}},\
  \bibinfo {pages} {045001} (\bibinfo {year} {2017})}\BibitemShut {NoStop}%
\bibitem [{\citenamefont {Khoshaman}\ \emph {et~al.}(2018)\citenamefont
  {Khoshaman}, \citenamefont {Vinci}, \citenamefont {Denis}, \citenamefont
  {Andriyash}, \citenamefont {Sadeghi},\ and\ \citenamefont
  {Amin}}]{khoshaman2018quantum}%
  \BibitemOpen
  \bibfield  {author} {\bibinfo {author} {\bibfnamefont {A.}~\bibnamefont
  {Khoshaman}}, \bibinfo {author} {\bibfnamefont {W.}~\bibnamefont {Vinci}},
  \bibinfo {author} {\bibfnamefont {B.}~\bibnamefont {Denis}}, \bibinfo
  {author} {\bibfnamefont {E.}~\bibnamefont {Andriyash}}, \bibinfo {author}
  {\bibfnamefont {H.}~\bibnamefont {Sadeghi}},\ and\ \bibinfo {author}
  {\bibfnamefont {M.~H.}\ \bibnamefont {Amin}},\ }\bibfield  {title} {\bibinfo
  {title} {Quantum variational autoencoder},\ }\href
  {https://doi.org/10.1088/2058-9565/aada1f} {\bibfield  {journal} {\bibinfo
  {journal} {Quantum Sci. Technol.}\ }\textbf {\bibinfo {volume} {4}},\
  \bibinfo {pages} {014001} (\bibinfo {year} {2018})}\BibitemShut {NoStop}%
\bibitem [{\citenamefont {Pepper}\ \emph {et~al.}(2019)\citenamefont {Pepper},
  \citenamefont {Tischler},\ and\ \citenamefont
  {Pryde}}]{pepper2019experimental}%
  \BibitemOpen
  \bibfield  {author} {\bibinfo {author} {\bibfnamefont {A.}~\bibnamefont
  {Pepper}}, \bibinfo {author} {\bibfnamefont {N.}~\bibnamefont {Tischler}},\
  and\ \bibinfo {author} {\bibfnamefont {G.~J.}\ \bibnamefont {Pryde}},\
  }\bibfield  {title} {\bibinfo {title} {Experimental realization of a quantum
  autoencoder: The compression of qutrits via machine learning},\ }\href
  {https://doi.org/10.1103/PhysRevLett.122.060501} {\bibfield  {journal}
  {\bibinfo  {journal} {Phys. Rev. Lett.}\ }\textbf {\bibinfo {volume} {122}},\
  \bibinfo {pages} {060501} (\bibinfo {year} {2019})}\BibitemShut {NoStop}%
\bibitem [{\citenamefont {Ma}\ \emph {et~al.}(2023)\citenamefont {Ma},
  \citenamefont {Huang}, \citenamefont {Chen}, \citenamefont {Dong},
  \citenamefont {Wang}, \citenamefont {Wu},\ and\ \citenamefont
  {Xiang}}]{ma2020compression}%
  \BibitemOpen
  \bibfield  {author} {\bibinfo {author} {\bibfnamefont {H.}~\bibnamefont
  {Ma}}, \bibinfo {author} {\bibfnamefont {C.-J.}\ \bibnamefont {Huang}},
  \bibinfo {author} {\bibfnamefont {C.}~\bibnamefont {Chen}}, \bibinfo {author}
  {\bibfnamefont {D.}~\bibnamefont {Dong}}, \bibinfo {author} {\bibfnamefont
  {Y.}~\bibnamefont {Wang}}, \bibinfo {author} {\bibfnamefont {R.-B.}\
  \bibnamefont {Wu}},\ and\ \bibinfo {author} {\bibfnamefont {G.-Y.}\
  \bibnamefont {Xiang}},\ }\bibfield  {title} {\bibinfo {title} {On compression
  rate of quantum autoencoders: Control design, numerical and experimental
  realization},\ }\href {https://doi.org/10.1016/j.automatica.2022.110659}
  {\bibfield  {journal} {\bibinfo  {journal} {Automatica}\ }\textbf {\bibinfo
  {volume} {147}},\ \bibinfo {pages} {110659} (\bibinfo {year}
  {2023})}\BibitemShut {NoStop}%
\bibitem [{\citenamefont {Lloyd}\ and\ \citenamefont
  {Weedbrook}(2018)}]{lloyd2018quantum}%
  \BibitemOpen
  \bibfield  {author} {\bibinfo {author} {\bibfnamefont {S.}~\bibnamefont
  {Lloyd}}\ and\ \bibinfo {author} {\bibfnamefont {C.}~\bibnamefont
  {Weedbrook}},\ }\bibfield  {title} {\bibinfo {title} {Quantum generative
  adversarial learning},\ }\href
  {https://doi.org/10.1103/PhysRevLett.121.040502} {\bibfield  {journal}
  {\bibinfo  {journal} {Phys. Rev. Lett.}\ }\textbf {\bibinfo {volume} {121}},\
  \bibinfo {pages} {040502} (\bibinfo {year} {2018})}\BibitemShut {NoStop}%
\bibitem [{\citenamefont {Dallaire-Demers}\ and\ \citenamefont
  {Killoran}(2018)}]{dallaire2018quantum}%
  \BibitemOpen
  \bibfield  {author} {\bibinfo {author} {\bibfnamefont {P.-L.}\ \bibnamefont
  {Dallaire-Demers}}\ and\ \bibinfo {author} {\bibfnamefont {N.}~\bibnamefont
  {Killoran}},\ }\bibfield  {title} {\bibinfo {title} {Quantum generative
  adversarial networks},\ }\href {https://doi.org/10.1103/PhysRevA.98.012324}
  {\bibfield  {journal} {\bibinfo  {journal} {Phys. Rev. A}\ }\textbf {\bibinfo
  {volume} {98}},\ \bibinfo {pages} {012324} (\bibinfo {year}
  {2018})}\BibitemShut {NoStop}%
\bibitem [{\citenamefont {Huang}\ \emph {et~al.}(2021)\citenamefont {Huang},
  \citenamefont {Du}, \citenamefont {Gong}, \citenamefont {Zhao}, \citenamefont
  {Wu}, \citenamefont {Wang}, \citenamefont {Li}, \citenamefont {Liang},
  \citenamefont {Lin}, \citenamefont {Xu} \emph
  {et~al.}}]{huang2021experimental}%
  \BibitemOpen
  \bibfield  {author} {\bibinfo {author} {\bibfnamefont {H.-L.}\ \bibnamefont
  {Huang}}, \bibinfo {author} {\bibfnamefont {Y.}~\bibnamefont {Du}}, \bibinfo
  {author} {\bibfnamefont {M.}~\bibnamefont {Gong}}, \bibinfo {author}
  {\bibfnamefont {Y.}~\bibnamefont {Zhao}}, \bibinfo {author} {\bibfnamefont
  {Y.}~\bibnamefont {Wu}}, \bibinfo {author} {\bibfnamefont {C.}~\bibnamefont
  {Wang}}, \bibinfo {author} {\bibfnamefont {S.}~\bibnamefont {Li}}, \bibinfo
  {author} {\bibfnamefont {F.}~\bibnamefont {Liang}}, \bibinfo {author}
  {\bibfnamefont {J.}~\bibnamefont {Lin}}, \bibinfo {author} {\bibfnamefont
  {Y.}~\bibnamefont {Xu}}, \emph {et~al.},\ }\bibfield  {title} {\bibinfo
  {title} {Experimental quantum generative adversarial networks for image
  generation},\ }\href {https://doi.org/10.1103/PhysRevApplied.16.024051}
  {\bibfield  {journal} {\bibinfo  {journal} {Phys. Rev. Appl.}\ }\textbf
  {\bibinfo {volume} {16}},\ \bibinfo {pages} {024051} (\bibinfo {year}
  {2021})}\BibitemShut {NoStop}%
\bibitem [{\citenamefont {Amin}\ \emph {et~al.}(2018)\citenamefont {Amin},
  \citenamefont {Andriyash}, \citenamefont {Rolfe}, \citenamefont
  {Kulchytskyy},\ and\ \citenamefont {Melko}}]{amin2018quantum}%
  \BibitemOpen
  \bibfield  {author} {\bibinfo {author} {\bibfnamefont {M.~H.}\ \bibnamefont
  {Amin}}, \bibinfo {author} {\bibfnamefont {E.}~\bibnamefont {Andriyash}},
  \bibinfo {author} {\bibfnamefont {J.}~\bibnamefont {Rolfe}}, \bibinfo
  {author} {\bibfnamefont {B.}~\bibnamefont {Kulchytskyy}},\ and\ \bibinfo
  {author} {\bibfnamefont {R.}~\bibnamefont {Melko}},\ }\bibfield  {title}
  {\bibinfo {title} {Quantum {B}oltzmann machine},\ }\href
  {https://doi.org/10.1103/PhysRevX.8.021050} {\bibfield  {journal} {\bibinfo
  {journal} {Phys. Rev. X}\ }\textbf {\bibinfo {volume} {8}},\ \bibinfo {pages}
  {021050} (\bibinfo {year} {2018})}\BibitemShut {NoStop}%
\bibitem [{\citenamefont {Zoufal}\ \emph {et~al.}(2021)\citenamefont {Zoufal},
  \citenamefont {Lucchi},\ and\ \citenamefont
  {Woerner}}]{zoufal2021variational}%
  \BibitemOpen
  \bibfield  {author} {\bibinfo {author} {\bibfnamefont {C.}~\bibnamefont
  {Zoufal}}, \bibinfo {author} {\bibfnamefont {A.}~\bibnamefont {Lucchi}},\
  and\ \bibinfo {author} {\bibfnamefont {S.}~\bibnamefont {Woerner}},\
  }\bibfield  {title} {\bibinfo {title} {Variational quantum {B}oltzmann
  machines},\ }\href {https://doi.org/10.1007/s42484-020-00033-7} {\bibfield
  {journal} {\bibinfo  {journal} {Quantum Mach. Intell.}\ }\textbf {\bibinfo
  {volume} {3}},\ \bibinfo {pages} {1} (\bibinfo {year} {2021})}\BibitemShut
  {NoStop}%
\bibitem [{\citenamefont {Kieferov{\'a}}\ and\ \citenamefont
  {Wiebe}(2017)}]{kieferova2017tomography}%
  \BibitemOpen
  \bibfield  {author} {\bibinfo {author} {\bibfnamefont {M.}~\bibnamefont
  {Kieferov{\'a}}}\ and\ \bibinfo {author} {\bibfnamefont {N.}~\bibnamefont
  {Wiebe}},\ }\bibfield  {title} {\bibinfo {title} {Tomography and generative
  training with quantum {B}oltzmann machines},\ }\href
  {https://doi.org/10.1103/PhysRevA.96.062327} {\bibfield  {journal} {\bibinfo
  {journal} {Phys. Rev. A}\ }\textbf {\bibinfo {volume} {96}},\ \bibinfo
  {pages} {062327} (\bibinfo {year} {2017})}\BibitemShut {NoStop}%
\bibitem [{\citenamefont {Lian}\ \emph {et~al.}(2019)\citenamefont {Lian},
  \citenamefont {Wang}, \citenamefont {Lu}, \citenamefont {Huang},
  \citenamefont {Wang}, \citenamefont {Yuan}, \citenamefont {Zhang},
  \citenamefont {Ouyang}, \citenamefont {Wang}, \citenamefont {Huang} \emph
  {et~al.}}]{lian2019machine}%
  \BibitemOpen
  \bibfield  {author} {\bibinfo {author} {\bibfnamefont {W.}~\bibnamefont
  {Lian}}, \bibinfo {author} {\bibfnamefont {S.-T.}\ \bibnamefont {Wang}},
  \bibinfo {author} {\bibfnamefont {S.}~\bibnamefont {Lu}}, \bibinfo {author}
  {\bibfnamefont {Y.}~\bibnamefont {Huang}}, \bibinfo {author} {\bibfnamefont
  {F.}~\bibnamefont {Wang}}, \bibinfo {author} {\bibfnamefont {X.}~\bibnamefont
  {Yuan}}, \bibinfo {author} {\bibfnamefont {W.}~\bibnamefont {Zhang}},
  \bibinfo {author} {\bibfnamefont {X.}~\bibnamefont {Ouyang}}, \bibinfo
  {author} {\bibfnamefont {X.}~\bibnamefont {Wang}}, \bibinfo {author}
  {\bibfnamefont {X.}~\bibnamefont {Huang}}, \emph {et~al.},\ }\bibfield
  {title} {\bibinfo {title} {Machine learning topological phases with a
  solid-state quantum simulator},\ }\href
  {https://doi.org/10.1103/PhysRevLett.122.210503} {\bibfield  {journal}
  {\bibinfo  {journal} {Phys. Rev. Lett.}\ }\textbf {\bibinfo {volume} {122}},\
  \bibinfo {pages} {210503} (\bibinfo {year} {2019})}\BibitemShut {NoStop}%
\bibitem [{\citenamefont {Zhang}\ \emph
  {et~al.}(2022{\natexlab{a}})\citenamefont {Zhang}, \citenamefont {Jiang},
  \citenamefont {Wang}, \citenamefont {Zhang}, \citenamefont {Huang},
  \citenamefont {Ouyang}, \citenamefont {Yu}, \citenamefont {Liu},
  \citenamefont {Deng},\ and\ \citenamefont {Duan}}]{zhang2022experimental}%
  \BibitemOpen
  \bibfield  {author} {\bibinfo {author} {\bibfnamefont {H.}~\bibnamefont
  {Zhang}}, \bibinfo {author} {\bibfnamefont {S.}~\bibnamefont {Jiang}},
  \bibinfo {author} {\bibfnamefont {X.}~\bibnamefont {Wang}}, \bibinfo {author}
  {\bibfnamefont {W.}~\bibnamefont {Zhang}}, \bibinfo {author} {\bibfnamefont
  {X.}~\bibnamefont {Huang}}, \bibinfo {author} {\bibfnamefont
  {X.}~\bibnamefont {Ouyang}}, \bibinfo {author} {\bibfnamefont
  {Y.}~\bibnamefont {Yu}}, \bibinfo {author} {\bibfnamefont {Y.}~\bibnamefont
  {Liu}}, \bibinfo {author} {\bibfnamefont {D.-L.}\ \bibnamefont {Deng}},\ and\
  \bibinfo {author} {\bibfnamefont {L.-M.}\ \bibnamefont {Duan}},\ }\bibfield
  {title} {\bibinfo {title} {Experimental demonstration of adversarial examples
  in learning topological phases},\ }\href
  {https://doi.org/10.1038/s41467-022-32611-7} {\bibfield  {journal} {\bibinfo
  {journal} {Nat. Commun.}\ }\textbf {\bibinfo {volume} {13}},\ \bibinfo
  {pages} {4993} (\bibinfo {year} {2022}{\natexlab{a}})}\BibitemShut {NoStop}%
\bibitem [{\citenamefont {Sels}\ and\ \citenamefont
  {Demler}(2021)}]{sels2021quantum}%
  \BibitemOpen
  \bibfield  {author} {\bibinfo {author} {\bibfnamefont {D.}~\bibnamefont
  {Sels}}\ and\ \bibinfo {author} {\bibfnamefont {E.}~\bibnamefont {Demler}},\
  }\bibfield  {title} {\bibinfo {title} {Quantum generative model for sampling
  many-body spectral functions},\ }\href
  {https://doi.org/10.1103/PhysRevB.103.014301} {\bibfield  {journal} {\bibinfo
   {journal} {Phys. Rev. B}\ }\textbf {\bibinfo {volume} {103}},\ \bibinfo
  {pages} {014301} (\bibinfo {year} {2021})}\BibitemShut {NoStop}%
\bibitem [{\citenamefont {Barratt}\ and\ \citenamefont
  {Sharma}()}]{barratt2018note}%
  \BibitemOpen
  \bibfield  {author} {\bibinfo {author} {\bibfnamefont {S.}~\bibnamefont
  {Barratt}}\ and\ \bibinfo {author} {\bibfnamefont {R.}~\bibnamefont
  {Sharma}},\ }\bibfield  {title} {\bibinfo {title} {A note on the inception
  score},\ }\href {https://arxiv.org/pdf/1801.01973.pdf} {\bibinfo  {journal}
  {arXiv preprint arXiv:1801.01973}\ }\BibitemShut {NoStop}%
\bibitem [{\citenamefont {Salimans}\ \emph {et~al.}(2016)\citenamefont
  {Salimans}, \citenamefont {Goodfellow}, \citenamefont {Zaremba},
  \citenamefont {Cheung}, \citenamefont {Radford},\ and\ \citenamefont
  {Chen}}]{salimans2016improved}%
  \BibitemOpen
\bibfield  {journal} {  }\bibfield  {author} {\bibinfo {author} {\bibfnamefont
  {T.}~\bibnamefont {Salimans}}, \bibinfo {author} {\bibfnamefont
  {I.}~\bibnamefont {Goodfellow}}, \bibinfo {author} {\bibfnamefont
  {W.}~\bibnamefont {Zaremba}}, \bibinfo {author} {\bibfnamefont
  {V.}~\bibnamefont {Cheung}}, \bibinfo {author} {\bibfnamefont
  {A.}~\bibnamefont {Radford}},\ and\ \bibinfo {author} {\bibfnamefont
  {X.}~\bibnamefont {Chen}},\ }\bibfield  {title} {\bibinfo {title} {Improved
  techniques for training {GAN}s},\ }\href
  {https://proceedings.neurips.cc/paper/2016/hash/8a3363abe792db2d8761d6403605aeb7-Abstract.html}
  {\bibfield  {journal} {\bibinfo  {journal} {Adv. Neural. Inf. Process.
  Syst.}\ }\textbf {\bibinfo {volume} {29}},\ \bibinfo {pages} {2234 }
  (\bibinfo {year} {2016})}\BibitemShut {NoStop}%
\bibitem [{\citenamefont {Holevo}(1973)}]{holevo1973bounds}%
  \BibitemOpen
  \bibfield  {author} {\bibinfo {author} {\bibfnamefont {A.~S.}\ \bibnamefont
  {Holevo}},\ }\bibfield  {title} {\bibinfo {title} {Bounds for the quantity of
  information transmitted by a quantum communication channel},\ }\href
  {http://www.mathnet.ru/php/archive.phtml?wshow=paper&jrnid=ppi&paperid=903&option_lang=eng}
  {\bibfield  {journal} {\bibinfo  {journal} {Probl. Peredachi. Inf.}\ }\textbf
  {\bibinfo {volume} {9}},\ \bibinfo {pages} {3 } (\bibinfo {year}
  {1973})}\BibitemShut {NoStop}%
\bibitem [{\citenamefont {Schumacher}\ and\ \citenamefont
  {Westmoreland}(2001)}]{schumacher2001optimal}%
  \BibitemOpen
  \bibfield  {author} {\bibinfo {author} {\bibfnamefont {B.}~\bibnamefont
  {Schumacher}}\ and\ \bibinfo {author} {\bibfnamefont {M.~D.}\ \bibnamefont
  {Westmoreland}},\ }\bibfield  {title} {\bibinfo {title} {Optimal signal
  ensembles},\ }\href {https://doi.org/10.1103/PhysRevA.63.022308} {\bibfield
  {journal} {\bibinfo  {journal} {Phys. Rev. A}\ }\textbf {\bibinfo {volume}
  {63}},\ \bibinfo {pages} {022308} (\bibinfo {year} {2001})}\BibitemShut
  {NoStop}%
\bibitem [{\citenamefont {Giovannetti}\ \emph {et~al.}(2006)\citenamefont
  {Giovannetti}, \citenamefont {Lloyd},\ and\ \citenamefont
  {Maccone}}]{giovannetti2006quantum}%
  \BibitemOpen
  \bibfield  {author} {\bibinfo {author} {\bibfnamefont {V.}~\bibnamefont
  {Giovannetti}}, \bibinfo {author} {\bibfnamefont {S.}~\bibnamefont {Lloyd}},\
  and\ \bibinfo {author} {\bibfnamefont {L.}~\bibnamefont {Maccone}},\
  }\bibfield  {title} {\bibinfo {title} {Quantum metrology},\ }\href
  {https://journals.aps.org/prl/abstract/10.1103/PhysRevLett.96.010401}
  {\bibfield  {journal} {\bibinfo  {journal} {Phys. Rev. Lett.}\ }\textbf
  {\bibinfo {volume} {96}},\ \bibinfo {pages} {010401} (\bibinfo {year}
  {2006})}\BibitemShut {NoStop}%
\bibitem [{\citenamefont {Bennett}\ and\ \citenamefont
  {Shor}(1998)}]{bennett1998quantum}%
  \BibitemOpen
  \bibfield  {author} {\bibinfo {author} {\bibfnamefont {C.~H.}\ \bibnamefont
  {Bennett}}\ and\ \bibinfo {author} {\bibfnamefont {P.~W.}\ \bibnamefont
  {Shor}},\ }\bibfield  {title} {\bibinfo {title} {Quantum information
  theory},\ }\href {https://doi.org/10.1109/18.720553} {\bibfield  {journal}
  {\bibinfo  {journal} {IEEE Trans. Inf. Th.}\ }\textbf {\bibinfo {volume}
  {44}},\ \bibinfo {pages} {2724} (\bibinfo {year} {1998})}\BibitemShut
  {NoStop}%
\bibitem [{\citenamefont {Hastings}(2009)}]{hastings2009superadditivity}%
  \BibitemOpen
  \bibfield  {author} {\bibinfo {author} {\bibfnamefont {M.~B.}\ \bibnamefont
  {Hastings}},\ }\bibfield  {title} {\bibinfo {title} {Superadditivity of
  communication capacity using entangled inputs},\ }\href
  {https://doi.org/10.1038/nphys1224} {\bibfield  {journal} {\bibinfo
  {journal} {Nat. Phys.}\ }\textbf {\bibinfo {volume} {5}},\ \bibinfo {pages}
  {255} (\bibinfo {year} {2009})}\BibitemShut {NoStop}%
\bibitem [{\citenamefont {Vedral}(2012)}]{vedral2012information}%
  \BibitemOpen
  \bibfield  {author} {\bibinfo {author} {\bibfnamefont {V.}~\bibnamefont
  {Vedral}},\ }\bibfield  {title} {\bibinfo {title} {An information--theoretic
  equality implying the {J}arzynski relation},\ }\href
  {https://doi.org/10.1088/1751-8113/45/27/272001} {\bibfield  {journal}
  {\bibinfo  {journal} {J. Phys. A: Math. Theor.}\ }\textbf {\bibinfo {volume}
  {45}},\ \bibinfo {pages} {272001} (\bibinfo {year} {2012})}\BibitemShut
  {NoStop}%
\bibitem [{\citenamefont {Sone}\ \emph {et~al.}(2023)\citenamefont {Sone},
  \citenamefont {Yamamoto}, \citenamefont {Holdsworth},\ and\ \citenamefont
  {Narang}}]{sone2023jarzynski}%
  \BibitemOpen
  \bibfield  {author} {\bibinfo {author} {\bibfnamefont {A.}~\bibnamefont
  {Sone}}, \bibinfo {author} {\bibfnamefont {N.}~\bibnamefont {Yamamoto}},
  \bibinfo {author} {\bibfnamefont {T.}~\bibnamefont {Holdsworth}},\ and\
  \bibinfo {author} {\bibfnamefont {P.}~\bibnamefont {Narang}},\ }\bibfield
  {title} {\bibinfo {title} {Jarzynski-like {E}quality of {N}onequilibrium
  {I}nformation {P}roduction {B}ased on {Q}uantum {C}ross {E}ntropy},\ }\href
  {https://doi.org/10.1103/PhysRevResearch.5.023039} {\bibfield  {journal}
  {\bibinfo  {journal} {Phys. Rev. Research}\ }\textbf {\bibinfo {volume}
  {5}},\ \bibinfo {pages} {023039} (\bibinfo {year} {2023})}\BibitemShut
  {NoStop}%
\bibitem [{\citenamefont {Kafri}\ and\ \citenamefont
  {Deffner}(2012)}]{kafri2012holevo}%
  \BibitemOpen
  \bibfield  {author} {\bibinfo {author} {\bibfnamefont {D.}~\bibnamefont
  {Kafri}}\ and\ \bibinfo {author} {\bibfnamefont {S.}~\bibnamefont
  {Deffner}},\ }\bibfield  {title} {\bibinfo {title} {Holevo's bound from a
  general quantum fluctuation theorem},\ }\href
  {https://doi.org/10.1103/PhysRevA.86.044302} {\bibfield  {journal} {\bibinfo
  {journal} {Phys. Rev. A}\ }\textbf {\bibinfo {volume} {86}},\ \bibinfo
  {pages} {044302} (\bibinfo {year} {2012})}\BibitemShut {NoStop}%
\bibitem [{\citenamefont {Sagawa}\ and\ \citenamefont {Ueda}(2010)}]{Sagawa10}%
  \BibitemOpen
  \bibfield  {author} {\bibinfo {author} {\bibfnamefont {T.}~\bibnamefont
  {Sagawa}}\ and\ \bibinfo {author} {\bibfnamefont {M.}~\bibnamefont {Ueda}},\
  }\bibfield  {title} {\bibinfo {title} {Generalized {J}arzynski equality under
  nonequilibrium feedback control},\ }\href
  {https://doi.org/10.1103/PhysRevLett.104.090602} {\bibfield  {journal}
  {\bibinfo  {journal} {Phys. Rev. Lett.}\ }\textbf {\bibinfo {volume} {104}},\
  \bibinfo {pages} {090602} (\bibinfo {year} {2010})}\BibitemShut {NoStop}%
\bibitem [{\citenamefont {Fujitani}\ and\ \citenamefont
  {Suzuki}(2010)}]{fujitani2010jarzynski}%
  \BibitemOpen
  \bibfield  {author} {\bibinfo {author} {\bibfnamefont {Y.}~\bibnamefont
  {Fujitani}}\ and\ \bibinfo {author} {\bibfnamefont {H.}~\bibnamefont
  {Suzuki}},\ }\bibfield  {title} {\bibinfo {title} {Jarzynski equality
  modified in the linear feedback system},\ }\href
  {https://doi.org/10.1143/JPSJ.79.104003} {\bibfield  {journal} {\bibinfo
  {journal} {J. Phys. Soc. Jpn.}\ }\textbf {\bibinfo {volume} {79}},\ \bibinfo
  {pages} {104003} (\bibinfo {year} {2010})}\BibitemShut {NoStop}%
\bibitem [{\citenamefont {Cong}\ \emph {et~al.}(2019)\citenamefont {Cong},
  \citenamefont {Choi},\ and\ \citenamefont {Lukin}}]{cong2019quantum}%
  \BibitemOpen
  \bibfield  {author} {\bibinfo {author} {\bibfnamefont {I.}~\bibnamefont
  {Cong}}, \bibinfo {author} {\bibfnamefont {S.}~\bibnamefont {Choi}},\ and\
  \bibinfo {author} {\bibfnamefont {M.~D.}\ \bibnamefont {Lukin}},\ }\bibfield
  {title} {\bibinfo {title} {Quantum convolutional neural networks},\ }\href
  {https://doi.org/10.1038/s41567-019-0648-8} {\bibfield  {journal} {\bibinfo
  {journal} {Nat. Phys}\ }\textbf {\bibinfo {volume} {15}},\ \bibinfo {pages}
  {1273} (\bibinfo {year} {2019})}\BibitemShut {NoStop}%
\bibitem [{\citenamefont {Grant}\ \emph {et~al.}(2018)\citenamefont {Grant},
  \citenamefont {Benedetti}, \citenamefont {Cao}, \citenamefont {Hallam},
  \citenamefont {Lockhart}, \citenamefont {Stojevic}, \citenamefont {Green},\
  and\ \citenamefont {Severini}}]{grant2018hierarchical}%
  \BibitemOpen
  \bibfield  {author} {\bibinfo {author} {\bibfnamefont {E.}~\bibnamefont
  {Grant}}, \bibinfo {author} {\bibfnamefont {M.}~\bibnamefont {Benedetti}},
  \bibinfo {author} {\bibfnamefont {S.}~\bibnamefont {Cao}}, \bibinfo {author}
  {\bibfnamefont {A.}~\bibnamefont {Hallam}}, \bibinfo {author} {\bibfnamefont
  {J.}~\bibnamefont {Lockhart}}, \bibinfo {author} {\bibfnamefont
  {V.}~\bibnamefont {Stojevic}}, \bibinfo {author} {\bibfnamefont {A.~G.}\
  \bibnamefont {Green}},\ and\ \bibinfo {author} {\bibfnamefont
  {S.}~\bibnamefont {Severini}},\ }\bibfield  {title} {\bibinfo {title}
  {Hierarchical quantum classifiers},\ }\href
  {https://doi.org/10.1038/s41534-018-0116-9} {\bibfield  {journal} {\bibinfo
  {journal} {npj Quantum Inf.}\ }\textbf {\bibinfo {volume} {4}},\ \bibinfo
  {pages} {65} (\bibinfo {year} {2018})}\BibitemShut {NoStop}%
\bibitem [{\citenamefont {Huang}\ \emph {et~al.}(2022)\citenamefont {Huang},
  \citenamefont {Broughton}, \citenamefont {Cotler}, \citenamefont {Chen},
  \citenamefont {Li}, \citenamefont {Mohseni}, \citenamefont {Neven},
  \citenamefont {Babbush}, \citenamefont {Kueng}, \citenamefont {Preskill},\
  and\ \citenamefont {Macclean}}]{huang2022quantum}%
  \BibitemOpen
  \bibfield  {author} {\bibinfo {author} {\bibfnamefont {H.-Y.}\ \bibnamefont
  {Huang}}, \bibinfo {author} {\bibfnamefont {M.}~\bibnamefont {Broughton}},
  \bibinfo {author} {\bibfnamefont {J.}~\bibnamefont {Cotler}}, \bibinfo
  {author} {\bibfnamefont {S.}~\bibnamefont {Chen}}, \bibinfo {author}
  {\bibfnamefont {J.}~\bibnamefont {Li}}, \bibinfo {author} {\bibfnamefont
  {M.}~\bibnamefont {Mohseni}}, \bibinfo {author} {\bibfnamefont
  {H.}~\bibnamefont {Neven}}, \bibinfo {author} {\bibfnamefont
  {R.}~\bibnamefont {Babbush}}, \bibinfo {author} {\bibfnamefont
  {R.}~\bibnamefont {Kueng}}, \bibinfo {author} {\bibfnamefont
  {J.}~\bibnamefont {Preskill}},\ and\ \bibinfo {author} {\bibfnamefont
  {J.~R.}\ \bibnamefont {Macclean}},\ }\bibfield  {title} {\bibinfo {title}
  {Quantum advantage in learning from experiments},\ }\href
  {https://doi.org/10.1126/science.abn7293} {\bibfield  {journal} {\bibinfo
  {journal} {Science}\ }\textbf {\bibinfo {volume} {376}},\ \bibinfo {pages}
  {1182} (\bibinfo {year} {2022})}\BibitemShut {NoStop}%
\bibitem [{\citenamefont {Hausladen}\ \emph {et~al.}(1996)\citenamefont
  {Hausladen}, \citenamefont {Jozsa}, \citenamefont {Schumacher}, \citenamefont
  {Westmoreland},\ and\ \citenamefont {Wootters}}]{hausladen1996classical}%
  \BibitemOpen
  \bibfield  {author} {\bibinfo {author} {\bibfnamefont {P.}~\bibnamefont
  {Hausladen}}, \bibinfo {author} {\bibfnamefont {R.}~\bibnamefont {Jozsa}},
  \bibinfo {author} {\bibfnamefont {B.}~\bibnamefont {Schumacher}}, \bibinfo
  {author} {\bibfnamefont {M.}~\bibnamefont {Westmoreland}},\ and\ \bibinfo
  {author} {\bibfnamefont {W.~K.}\ \bibnamefont {Wootters}},\ }\bibfield
  {title} {\bibinfo {title} {Classical information capacity of a quantum
  channel},\ }\href {https://doi.org/10.1103/PhysRevA.54.1869} {\bibfield
  {journal} {\bibinfo  {journal} {Phys. Rev. A}\ }\textbf {\bibinfo {volume}
  {54}},\ \bibinfo {pages} {1869} (\bibinfo {year} {1996})}\BibitemShut
  {NoStop}%
\bibitem [{\citenamefont {Schumacher}\ and\ \citenamefont
  {Westmoreland}(1997)}]{schumacher1997sending}%
  \BibitemOpen
  \bibfield  {author} {\bibinfo {author} {\bibfnamefont {B.}~\bibnamefont
  {Schumacher}}\ and\ \bibinfo {author} {\bibfnamefont {M.~D.}\ \bibnamefont
  {Westmoreland}},\ }\bibfield  {title} {\bibinfo {title} {Sending classical
  information via noisy quantum channels},\ }\href
  {https://doi.org/10.1103/PhysRevA.56.131} {\bibfield  {journal} {\bibinfo
  {journal} {Phys. Rev. A}\ }\textbf {\bibinfo {volume} {56}},\ \bibinfo
  {pages} {131} (\bibinfo {year} {1997})}\BibitemShut {NoStop}%
\bibitem [{\citenamefont {Holevo}(1998)}]{holevo1998capacity}%
  \BibitemOpen
  \bibfield  {author} {\bibinfo {author} {\bibfnamefont {A.~S.}\ \bibnamefont
  {Holevo}},\ }\bibfield  {title} {\bibinfo {title} {The capacity of the
  quantum channel with general signal states},\ }\href
  {https://doi.org/10.1109/18.651037} {\bibfield  {journal} {\bibinfo
  {journal} {IEEE Trans. Info. Theo.}\ }\textbf {\bibinfo {volume} {44}},\
  \bibinfo {pages} {269} (\bibinfo {year} {1998})}\BibitemShut {NoStop}%
\bibitem [{\citenamefont {King}\ and\ \citenamefont
  {Ruskai}(2001)}]{king2001capacity}%
  \BibitemOpen
  \bibfield  {author} {\bibinfo {author} {\bibfnamefont {C.}~\bibnamefont
  {King}}\ and\ \bibinfo {author} {\bibfnamefont {M.~B.}\ \bibnamefont
  {Ruskai}},\ }\bibfield  {title} {\bibinfo {title} {Capacity of quantum
  channels using product measurements},\ }\href
  {https://doi.org/10.1063/1.1327598} {\bibfield  {journal} {\bibinfo
  {journal} {J. Math. Phys.}\ }\textbf {\bibinfo {volume} {42}},\ \bibinfo
  {pages} {87} (\bibinfo {year} {2001})}\BibitemShut {NoStop}%
\bibitem [{\citenamefont {Holevo}\ \emph {et~al.}(1999)\citenamefont {Holevo},
  \citenamefont {Sohma},\ and\ \citenamefont {Hirota}}]{holevo1999capacity}%
  \BibitemOpen
  \bibfield  {author} {\bibinfo {author} {\bibfnamefont {A.~S.}\ \bibnamefont
  {Holevo}}, \bibinfo {author} {\bibfnamefont {M.}~\bibnamefont {Sohma}},\ and\
  \bibinfo {author} {\bibfnamefont {O.}~\bibnamefont {Hirota}},\ }\bibfield
  {title} {\bibinfo {title} {Capacity of quantum {G}aussian channels},\ }\href
  {https://doi.org/10.1103/PhysRevA.59.1820} {\bibfield  {journal} {\bibinfo
  {journal} {Phys. Rev. A}\ }\textbf {\bibinfo {volume} {59}},\ \bibinfo
  {pages} {1820} (\bibinfo {year} {1999})}\BibitemShut {NoStop}%
\bibitem [{\citenamefont {Siudzi{\'n}ska}(2020)}]{siudzinska2020classical}%
  \BibitemOpen
  \bibfield  {author} {\bibinfo {author} {\bibfnamefont {K.}~\bibnamefont
  {Siudzi{\'n}ska}},\ }\bibfield  {title} {\bibinfo {title} {Classical capacity
  of generalized {P}auli channels},\ }\href
  {https://doi.org/10.1088/1751-8121/abb276} {\bibfield  {journal} {\bibinfo
  {journal} {J. Phys. A: Math. Theor.}\ }\textbf {\bibinfo {volume} {53}},\
  \bibinfo {pages} {445301} (\bibinfo {year} {2020})}\BibitemShut {NoStop}%
\bibitem [{\citenamefont {Lee}\ \emph {et~al.}(2015)\citenamefont {Lee},
  \citenamefont {Ji}, \citenamefont {Park},\ and\ \citenamefont
  {Nha}}]{lee2015classical}%
  \BibitemOpen
  \bibfield  {author} {\bibinfo {author} {\bibfnamefont {J.}~\bibnamefont
  {Lee}}, \bibinfo {author} {\bibfnamefont {S.-W.}\ \bibnamefont {Ji}},
  \bibinfo {author} {\bibfnamefont {J.}~\bibnamefont {Park}},\ and\ \bibinfo
  {author} {\bibfnamefont {H.}~\bibnamefont {Nha}},\ }\bibfield  {title}
  {\bibinfo {title} {Classical capacity of {G}aussian communication under a
  single noisy channel},\ }\href {https://doi.org/10.1103/PhysRevA.91.042336}
  {\bibfield  {journal} {\bibinfo  {journal} {Phys. Rev. A}\ }\textbf {\bibinfo
  {volume} {91}},\ \bibinfo {pages} {042336} (\bibinfo {year}
  {2015})}\BibitemShut {NoStop}%
\bibitem [{\citenamefont {Gyongyosi}\ \emph {et~al.}(2018)\citenamefont
  {Gyongyosi}, \citenamefont {Imre},\ and\ \citenamefont
  {Nguyen}}]{gyongyosi2018survey}%
  \BibitemOpen
  \bibfield  {author} {\bibinfo {author} {\bibfnamefont {L.}~\bibnamefont
  {Gyongyosi}}, \bibinfo {author} {\bibfnamefont {S.}~\bibnamefont {Imre}},\
  and\ \bibinfo {author} {\bibfnamefont {H.~V.}\ \bibnamefont {Nguyen}},\
  }\bibfield  {title} {\bibinfo {title} {A survey on quantum channel
  capacities},\ }\href {https://doi.org/10.1109/COMST.2017.2786748} {\bibfield
  {journal} {\bibinfo  {journal} {IEEE Commun. Surv. Tutor.}\ }\textbf
  {\bibinfo {volume} {20}},\ \bibinfo {pages} {1149} (\bibinfo {year}
  {2018})}\BibitemShut {NoStop}%
\bibitem [{\citenamefont {Li}\ and\ \citenamefont {Deng}(2022)}]{li2022recent}%
  \BibitemOpen
  \bibfield  {author} {\bibinfo {author} {\bibfnamefont {W.}~\bibnamefont
  {Li}}\ and\ \bibinfo {author} {\bibfnamefont {D.-L.}\ \bibnamefont {Deng}},\
  }\bibfield  {title} {\bibinfo {title} {Recent advances for quantum
  classifiers},\ }\href {https://doi.org/10.1007/s11433-021-1793-6} {\bibfield
  {journal} {\bibinfo  {journal} {Sci. China: Phys. Mech. Astron.}\ }\textbf
  {\bibinfo {volume} {65}},\ \bibinfo {pages} {220301} (\bibinfo {year}
  {2022})}\BibitemShut {NoStop}%
\bibitem [{\citenamefont {Hur}\ \emph {et~al.}(2022)\citenamefont {Hur},
  \citenamefont {Kim},\ and\ \citenamefont {Park}}]{hur2022quantum}%
  \BibitemOpen
  \bibfield  {author} {\bibinfo {author} {\bibfnamefont {T.}~\bibnamefont
  {Hur}}, \bibinfo {author} {\bibfnamefont {L.}~\bibnamefont {Kim}},\ and\
  \bibinfo {author} {\bibfnamefont {D.~K.}\ \bibnamefont {Park}},\ }\bibfield
  {title} {\bibinfo {title} {Quantum convolutional neural network for classical
  data classification},\ }\href {https://doi.org/10.1007/s42484-021-00061-x}
  {\bibfield  {journal} {\bibinfo  {journal} {Quantum Mach. Intell.}\ }\textbf
  {\bibinfo {volume} {4}},\ \bibinfo {pages} {3} (\bibinfo {year}
  {2022})}\BibitemShut {NoStop}%
\bibitem [{\citenamefont {Wrobel}\ \emph {et~al.}()\citenamefont {Wrobel},
  \citenamefont {Baul}, \citenamefont {Moreno},\ and\ \citenamefont
  {Tam}}]{wrobel2021application}%
  \BibitemOpen
  \bibfield  {author} {\bibinfo {author} {\bibfnamefont {N.}~\bibnamefont
  {Wrobel}}, \bibinfo {author} {\bibfnamefont {A.}~\bibnamefont {Baul}},
  \bibinfo {author} {\bibfnamefont {J.}~\bibnamefont {Moreno}},\ and\ \bibinfo
  {author} {\bibfnamefont {K.-M.}\ \bibnamefont {Tam}},\ }\bibfield  {title}
  {\bibinfo {title} {An application of quantum machine learning on quantum
  correlated systems: Quantum convolutional neural network as a classifier for
  many-body wavefunctions from the quantum variational eigensolver},\ }\href
  {https://arxiv.org/abs/2111.05076} {\bibinfo  {journal} {arXiv preprint
  arXiv:2111.05076}\ }\BibitemShut {NoStop}%
\bibitem [{\citenamefont {Chen}\ \emph {et~al.}(2023)\citenamefont {Chen},
  \citenamefont {Chen}, \citenamefont {Long}, \citenamefont {Zhu},
  \citenamefont {Yuan},\ and\ \citenamefont {Wu}}]{chen2023quantum}%
  \BibitemOpen
\bibfield  {journal} {  }\bibfield  {author} {\bibinfo {author} {\bibfnamefont
  {G.}~\bibnamefont {Chen}}, \bibinfo {author} {\bibfnamefont {Q.}~\bibnamefont
  {Chen}}, \bibinfo {author} {\bibfnamefont {S.}~\bibnamefont {Long}}, \bibinfo
  {author} {\bibfnamefont {W.}~\bibnamefont {Zhu}}, \bibinfo {author}
  {\bibfnamefont {Z.}~\bibnamefont {Yuan}},\ and\ \bibinfo {author}
  {\bibfnamefont {Y.}~\bibnamefont {Wu}},\ }\bibfield  {title} {\bibinfo
  {title} {Quantum convolutional neural network for image classification},\
  }\href {https://doi.org/10.1007/s10044-022-01113-z} {\bibfield  {journal}
  {\bibinfo  {journal} {Pattern Anal. Appl.}\ }\textbf {\bibinfo {volume}
  {26}},\ \bibinfo {pages} {655} (\bibinfo {year} {2023})}\BibitemShut
  {NoStop}%
\bibitem [{\citenamefont {Jozsa}\ \emph {et~al.}(1994)\citenamefont {Jozsa},
  \citenamefont {Robb},\ and\ \citenamefont {Wootters}}]{jozsa1994lower}%
  \BibitemOpen
  \bibfield  {author} {\bibinfo {author} {\bibfnamefont {R.}~\bibnamefont
  {Jozsa}}, \bibinfo {author} {\bibfnamefont {D.}~\bibnamefont {Robb}},\ and\
  \bibinfo {author} {\bibfnamefont {W.~K.}\ \bibnamefont {Wootters}},\
  }\bibfield  {title} {\bibinfo {title} {Lower bound for accessible information
  in quantum mechanics},\ }\href {https://doi.org/10.1103/PhysRevA.49.668}
  {\bibfield  {journal} {\bibinfo  {journal} {Phys. Rev. A}\ }\textbf {\bibinfo
  {volume} {49}},\ \bibinfo {pages} {668} (\bibinfo {year} {1994})}\BibitemShut
  {NoStop}%
\bibitem [{\citenamefont {Davies}(1978)}]{davies1978information}%
  \BibitemOpen
  \bibfield  {author} {\bibinfo {author} {\bibfnamefont {E.}~\bibnamefont
  {Davies}},\ }\bibfield  {title} {\bibinfo {title} {Information and quantum
  measurement},\ }\href {https://doi.org/10.1109/TIT.1978.1055941} {\bibfield
  {journal} {\bibinfo  {journal} {IEEE Trans. Inf. Theory}\ }\textbf {\bibinfo
  {volume} {24}},\ \bibinfo {pages} {596} (\bibinfo {year} {1978})}\BibitemShut
  {NoStop}%
\bibitem [{\citenamefont {Watrous}(2018)}]{WatrousBook18}%
  \BibitemOpen
  \bibfield  {author} {\bibinfo {author} {\bibfnamefont {J.}~\bibnamefont
  {Watrous}},\ }\href@noop {} {\emph {\bibinfo {title} {The Theory of Quantum
  Information}}}\ (\bibinfo  {publisher} {Cambridge University Press},\
  \bibinfo {year} {2018})\BibitemShut {NoStop}%
\bibitem [{\citenamefont {Vedral}(2002)}]{vedral2002role}%
  \BibitemOpen
  \bibfield  {author} {\bibinfo {author} {\bibfnamefont {V.}~\bibnamefont
  {Vedral}},\ }\bibfield  {title} {\bibinfo {title} {The role of relative
  entropy in quantum information theory},\ }\href
  {https://doi.org/10.1103/RevModPhys.74.197} {\bibfield  {journal} {\bibinfo
  {journal} {Rev. Mod. Phys.}\ }\textbf {\bibinfo {volume} {74}},\ \bibinfo
  {pages} {197} (\bibinfo {year} {2002})}\BibitemShut {NoStop}%
\bibitem [{\citenamefont {Lostaglio}\ \emph
  {et~al.}(2015{\natexlab{a}})\citenamefont {Lostaglio}, \citenamefont
  {Korzekwa}, \citenamefont {Jennings},\ and\ \citenamefont
  {Rudolph}}]{lostaglio2015quantum}%
  \BibitemOpen
  \bibfield  {author} {\bibinfo {author} {\bibfnamefont {M.}~\bibnamefont
  {Lostaglio}}, \bibinfo {author} {\bibfnamefont {K.}~\bibnamefont {Korzekwa}},
  \bibinfo {author} {\bibfnamefont {D.}~\bibnamefont {Jennings}},\ and\
  \bibinfo {author} {\bibfnamefont {T.}~\bibnamefont {Rudolph}},\ }\bibfield
  {title} {\bibinfo {title} {Quantum coherence, time-translation symmetry, and
  thermodynamics},\ }\href {https://doi.org/10.1103/PhysRevX.5.021001}
  {\bibfield  {journal} {\bibinfo  {journal} {Phys. Rev. X}\ }\textbf {\bibinfo
  {volume} {5}},\ \bibinfo {pages} {021001} (\bibinfo {year}
  {2015}{\natexlab{a}})}\BibitemShut {NoStop}%
\bibitem [{\citenamefont {Lostaglio}\ \emph
  {et~al.}(2015{\natexlab{b}})\citenamefont {Lostaglio}, \citenamefont
  {Jennings},\ and\ \citenamefont {Rudolph}}]{lostaglio2015description}%
  \BibitemOpen
  \bibfield  {author} {\bibinfo {author} {\bibfnamefont {M.}~\bibnamefont
  {Lostaglio}}, \bibinfo {author} {\bibfnamefont {D.}~\bibnamefont
  {Jennings}},\ and\ \bibinfo {author} {\bibfnamefont {T.}~\bibnamefont
  {Rudolph}},\ }\bibfield  {title} {\bibinfo {title} {Description of quantum
  coherence in thermodynamic processes requires constraints beyond free
  energy},\ }\href {https://doi.org/10.1038/ncomms7383} {\bibfield  {journal}
  {\bibinfo  {journal} {Nat. Commun.}\ }\textbf {\bibinfo {volume} {6}},\
  \bibinfo {pages} {6383} (\bibinfo {year} {2015}{\natexlab{b}})}\BibitemShut
  {NoStop}%
\bibitem [{\citenamefont {Marvian}\ and\ \citenamefont
  {Spekkens}(2014{\natexlab{a}})}]{marvian2014asymmetry}%
  \BibitemOpen
  \bibfield  {author} {\bibinfo {author} {\bibfnamefont {I.}~\bibnamefont
  {Marvian}}\ and\ \bibinfo {author} {\bibfnamefont {R.~W.}\ \bibnamefont
  {Spekkens}},\ }\bibfield  {title} {\bibinfo {title} {Asymmetry properties of
  pure quantum states},\ }\href {https://doi.org/10.1103/PhysRevA.90.014102}
  {\bibfield  {journal} {\bibinfo  {journal} {Phys. Rev. A}\ }\textbf {\bibinfo
  {volume} {90}},\ \bibinfo {pages} {014102} (\bibinfo {year}
  {2014}{\natexlab{a}})}\BibitemShut {NoStop}%
\bibitem [{\citenamefont {Marvian}(2020)}]{marvian2020coherence}%
  \BibitemOpen
  \bibfield  {author} {\bibinfo {author} {\bibfnamefont {I.}~\bibnamefont
  {Marvian}},\ }\bibfield  {title} {\bibinfo {title} {Coherence distillation
  machines are impossible in quantum thermodynamics},\ }\href
  {https://doi.org/10.1038/s41467-019-13846-3} {\bibfield  {journal} {\bibinfo
  {journal} {Nat. Commun.}\ }\textbf {\bibinfo {volume} {11}},\ \bibinfo
  {pages} {25} (\bibinfo {year} {2020})}\BibitemShut {NoStop}%
\bibitem [{\citenamefont {Marvian}\ \emph {et~al.}(2016)\citenamefont
  {Marvian}, \citenamefont {Spekkens},\ and\ \citenamefont
  {Zanardi}}]{marvian2016quantum}%
  \BibitemOpen
  \bibfield  {author} {\bibinfo {author} {\bibfnamefont {I.}~\bibnamefont
  {Marvian}}, \bibinfo {author} {\bibfnamefont {R.~W.}\ \bibnamefont
  {Spekkens}},\ and\ \bibinfo {author} {\bibfnamefont {P.}~\bibnamefont
  {Zanardi}},\ }\bibfield  {title} {\bibinfo {title} {Quantum speed limits,
  coherence, and asymmetry},\ }\href
  {https://doi.org/10.1103/PhysRevA.93.052331} {\bibfield  {journal} {\bibinfo
  {journal} {Phys. Rev. A}\ }\textbf {\bibinfo {volume} {93}},\ \bibinfo
  {pages} {052331} (\bibinfo {year} {2016})}\BibitemShut {NoStop}%
\bibitem [{\citenamefont {Takagi}(2019)}]{takagi2019skew}%
  \BibitemOpen
  \bibfield  {author} {\bibinfo {author} {\bibfnamefont {R.}~\bibnamefont
  {Takagi}},\ }\bibfield  {title} {\bibinfo {title} {Skew informations from an
  operational view via resource theory of asymmetry},\ }\href
  {https://doi.org/10.1038/s41598-019-50279-w} {\bibfield  {journal} {\bibinfo
  {journal} {Sci. Rep.}\ }\textbf {\bibinfo {volume} {9}},\ \bibinfo {pages}
  {14562} (\bibinfo {year} {2019})}\BibitemShut {NoStop}%
\bibitem [{\citenamefont {Yamaguchi}\ and\ \citenamefont
  {Tajima}(2023)}]{yamaguchi2023smooth}%
  \BibitemOpen
  \bibfield  {author} {\bibinfo {author} {\bibfnamefont {K.}~\bibnamefont
  {Yamaguchi}}\ and\ \bibinfo {author} {\bibfnamefont {H.}~\bibnamefont
  {Tajima}},\ }\bibfield  {title} {\bibinfo {title} {Smooth metric adjusted
  skew information rates},\ }\href {https://doi.org/10.22331/q-2023-05-22-1012}
  {\bibfield  {journal} {\bibinfo  {journal} {Quantum}\ }\textbf {\bibinfo
  {volume} {7}},\ \bibinfo {pages} {1012} (\bibinfo {year} {2023})}\BibitemShut
  {NoStop}%
\bibitem [{\citenamefont {Ahmadi}\ \emph {et~al.}(2013)\citenamefont {Ahmadi},
  \citenamefont {Jennings},\ and\ \citenamefont {Rudolph}}]{ahmadi2013wigner}%
  \BibitemOpen
  \bibfield  {author} {\bibinfo {author} {\bibfnamefont {M.}~\bibnamefont
  {Ahmadi}}, \bibinfo {author} {\bibfnamefont {D.}~\bibnamefont {Jennings}},\
  and\ \bibinfo {author} {\bibfnamefont {T.}~\bibnamefont {Rudolph}},\
  }\bibfield  {title} {\bibinfo {title} {The {W}igner-{A}raki-{Y}anase theorem
  and the quantum resource theory of asymmetry},\ }\href
  {https://doi.org/10.1088/1367-2630/15/1/013057} {\bibfield  {journal}
  {\bibinfo  {journal} {New J. Phys.}\ }\textbf {\bibinfo {volume} {15}},\
  \bibinfo {pages} {013057} (\bibinfo {year} {2013})}\BibitemShut {NoStop}%
\bibitem [{\citenamefont {Gour}\ and\ \citenamefont
  {Spekkens}(2008)}]{gour2008resource}%
  \BibitemOpen
  \bibfield  {author} {\bibinfo {author} {\bibfnamefont {G.}~\bibnamefont
  {Gour}}\ and\ \bibinfo {author} {\bibfnamefont {R.~W.}\ \bibnamefont
  {Spekkens}},\ }\bibfield  {title} {\bibinfo {title} {The resource theory of
  quantum reference frames: {M}anipulations and monotones},\ }\href
  {https://doi.org/10.1088/1367-2630/10/3/033023} {\bibfield  {journal}
  {\bibinfo  {journal} {New J. Phys.}\ }\textbf {\bibinfo {volume} {10}},\
  \bibinfo {pages} {033023} (\bibinfo {year} {2008})}\BibitemShut {NoStop}%
\bibitem [{\citenamefont {Marvian}\ and\ \citenamefont
  {Spekkens}(2013)}]{marvian2013theory}%
  \BibitemOpen
  \bibfield  {author} {\bibinfo {author} {\bibfnamefont {I.}~\bibnamefont
  {Marvian}}\ and\ \bibinfo {author} {\bibfnamefont {R.~W.}\ \bibnamefont
  {Spekkens}},\ }\bibfield  {title} {\bibinfo {title} {The theory of
  manipulations of pure state asymmetry: I. {B}asic tools, equivalence classes
  and single copy transformations},\ }\href
  {https://doi.org/10.1088/1367-2630/15/3/033001} {\bibfield  {journal}
  {\bibinfo  {journal} {New J. Phys.}\ }\textbf {\bibinfo {volume} {15}},\
  \bibinfo {pages} {033001} (\bibinfo {year} {2013})}\BibitemShut {NoStop}%
\bibitem [{\citenamefont {Marvian}\ and\ \citenamefont
  {Spekkens}(2014{\natexlab{b}})}]{marvian2014extending}%
  \BibitemOpen
  \bibfield  {author} {\bibinfo {author} {\bibfnamefont {I.}~\bibnamefont
  {Marvian}}\ and\ \bibinfo {author} {\bibfnamefont {R.~W.}\ \bibnamefont
  {Spekkens}},\ }\bibfield  {title} {\bibinfo {title} {Extending {N}oether’s
  theorem by quantifying the asymmetry of quantum states},\ }\href
  {https://doi.org/10.1038/ncomms4821} {\bibfield  {journal} {\bibinfo
  {journal} {Nat. Commun.}\ }\textbf {\bibinfo {volume} {5}},\ \bibinfo {pages}
  {3821} (\bibinfo {year} {2014}{\natexlab{b}})}\BibitemShut {NoStop}%
\bibitem [{\citenamefont {Gour}\ \emph {et~al.}(2009)\citenamefont {Gour},
  \citenamefont {Marvian},\ and\ \citenamefont {Spekkens}}]{gour2009measuring}%
  \BibitemOpen
  \bibfield  {author} {\bibinfo {author} {\bibfnamefont {G.}~\bibnamefont
  {Gour}}, \bibinfo {author} {\bibfnamefont {I.}~\bibnamefont {Marvian}},\ and\
  \bibinfo {author} {\bibfnamefont {R.~W.}\ \bibnamefont {Spekkens}},\
  }\bibfield  {title} {\bibinfo {title} {Measuring the quality of a quantum
  reference frame: The relative entropy of frameness},\ }\href
  {https://doi.org/10.1103/PhysRevA.80.012307} {\bibfield  {journal} {\bibinfo
  {journal} {Phys. Rev. A}\ }\textbf {\bibinfo {volume} {80}},\ \bibinfo
  {pages} {012307} (\bibinfo {year} {2009})}\BibitemShut {NoStop}%
\bibitem [{\citenamefont {Marvian}\ and\ \citenamefont
  {Lloyd}()}]{marvian2016clocks}%
  \BibitemOpen
  \bibfield  {author} {\bibinfo {author} {\bibfnamefont {I.}~\bibnamefont
  {Marvian}}\ and\ \bibinfo {author} {\bibfnamefont {S.}~\bibnamefont
  {Lloyd}},\ }\bibfield  {title} {\bibinfo {title} {From clocks to cloners:
  Catalytic transformations under covariant operations and recoverability},\
  }\href {https://arxiv.org/abs/1608.07325} {\bibinfo  {journal} {arXiv
  preprint arXiv:1608.07325}\ }\BibitemShut {NoStop}%
\bibitem [{\citenamefont {Vaccaro}\ \emph {et~al.}(2008)\citenamefont
  {Vaccaro}, \citenamefont {Anselmi}, \citenamefont {Wiseman},\ and\
  \citenamefont {Jacobs}}]{vaccaro2008tradeoff}%
  \BibitemOpen
\bibfield  {journal} {  }\bibfield  {author} {\bibinfo {author} {\bibfnamefont
  {J.~A.}\ \bibnamefont {Vaccaro}}, \bibinfo {author} {\bibfnamefont
  {F.}~\bibnamefont {Anselmi}}, \bibinfo {author} {\bibfnamefont {H.~M.}\
  \bibnamefont {Wiseman}},\ and\ \bibinfo {author} {\bibfnamefont
  {K.}~\bibnamefont {Jacobs}},\ }\bibfield  {title} {\bibinfo {title} {Tradeoff
  between extractable mechanical work, accessible entanglement, and ability to
  act as a reference system, under arbitrary superselection rules},\ }\href
  {https://doi.org/10.1103/PhysRevA.77.032114} {\bibfield  {journal} {\bibinfo
  {journal} {Phys. Rev. A}\ }\textbf {\bibinfo {volume} {77}},\ \bibinfo
  {pages} {032114} (\bibinfo {year} {2008})}\BibitemShut {NoStop}%
\bibitem [{\citenamefont {Bartlett}\ and\ \citenamefont
  {Wiseman}(2003)}]{bartlett2003entanglement}%
  \BibitemOpen
  \bibfield  {author} {\bibinfo {author} {\bibfnamefont {S.~D.}\ \bibnamefont
  {Bartlett}}\ and\ \bibinfo {author} {\bibfnamefont {H.~M.}\ \bibnamefont
  {Wiseman}},\ }\bibfield  {title} {\bibinfo {title} {Entanglement constrained
  by superselection rules},\ }\href
  {https://doi.org/10.1103/PhysRevLett.91.097903} {\bibfield  {journal}
  {\bibinfo  {journal} {Phys. Rev. Lett.}\ }\textbf {\bibinfo {volume} {91}},\
  \bibinfo {pages} {097903} (\bibinfo {year} {2003})}\BibitemShut {NoStop}%
\bibitem [{Note1()}]{Note1}%
  \BibitemOpen
  \bibinfo {note} {For a finite group, we have $\protect \mathcal {G}(\rho
  )\equiv \protect \frac {1}{\protect \abs {G}}\DOTSB \sum@ \slimits@ _{g\in
  G}U_g\rho U_g^\dagger $, where $\protect \abs {G}$ denotes the order of the
  group.}\BibitemShut {Stop}%
\bibitem [{\citenamefont {{\AA}berg}()}]{aberg2006quantifying}%
  \BibitemOpen
  \bibfield  {author} {\bibinfo {author} {\bibfnamefont {J.}~\bibnamefont
  {{\AA}berg}},\ }\bibfield  {title} {\bibinfo {title} {Quantifying
  superposition},\ }\href {https://arxiv.org/abs/quant-ph/0612146} {\bibinfo
  {journal} {arXiv preprint quant-ph/0612146}\ }\BibitemShut {NoStop}%
\bibitem [{\citenamefont {Baumgratz}\ \emph {et~al.}(2014)\citenamefont
  {Baumgratz}, \citenamefont {Cramer},\ and\ \citenamefont
  {Plenio}}]{baumgratz2014quantifying}%
  \BibitemOpen
\bibfield  {journal} {  }\bibfield  {author} {\bibinfo {author} {\bibfnamefont
  {T.}~\bibnamefont {Baumgratz}}, \bibinfo {author} {\bibfnamefont
  {M.}~\bibnamefont {Cramer}},\ and\ \bibinfo {author} {\bibfnamefont {M.~B.}\
  \bibnamefont {Plenio}},\ }\bibfield  {title} {\bibinfo {title} {Quantifying
  coherence},\ }\href {https://doi.org/10.1103/PhysRevLett.113.140401}
  {\bibfield  {journal} {\bibinfo  {journal} {Phys. Rev. Lett.}\ }\textbf
  {\bibinfo {volume} {113}},\ \bibinfo {pages} {140401} (\bibinfo {year}
  {2014})}\BibitemShut {NoStop}%
\bibitem [{Note2()}]{Note2}%
  \BibitemOpen
  \bibinfo {note} {{For the origin of $\gamma $ regarding the Holevo
  information, refer to Eqs.~(11)\protect \,--\protect \,(14) in Ref.~\cite
  {kafri2012holevo}.}}\BibitemShut {Stop}%
\bibitem [{\citenamefont {Ruskai}(2002)}]{ruskai2002inequalities}%
  \BibitemOpen
  \bibfield  {author} {\bibinfo {author} {\bibfnamefont {M.~B.}\ \bibnamefont
  {Ruskai}},\ }\bibfield  {title} {\bibinfo {title} {Inequalities for quantum
  entropy: A review with conditions for equality},\ }\href
  {https://doi.org/10.1063/1.1497701} {\bibfield  {journal} {\bibinfo
  {journal} {J. Math. Phys.}\ }\textbf {\bibinfo {volume} {43}},\ \bibinfo
  {pages} {4358} (\bibinfo {year} {2002})}\BibitemShut {NoStop}%
\bibitem [{\citenamefont {Pesah}\ \emph {et~al.}(2021)\citenamefont {Pesah},
  \citenamefont {Cerezo}, \citenamefont {Wang}, \citenamefont {Volkoff},
  \citenamefont {Sornborger},\ and\ \citenamefont {Coles}}]{pesah2021absence}%
  \BibitemOpen
  \bibfield  {author} {\bibinfo {author} {\bibfnamefont {A.}~\bibnamefont
  {Pesah}}, \bibinfo {author} {\bibfnamefont {M.}~\bibnamefont {Cerezo}},
  \bibinfo {author} {\bibfnamefont {S.}~\bibnamefont {Wang}}, \bibinfo {author}
  {\bibfnamefont {T.}~\bibnamefont {Volkoff}}, \bibinfo {author} {\bibfnamefont
  {A.~T.}\ \bibnamefont {Sornborger}},\ and\ \bibinfo {author} {\bibfnamefont
  {P.~J.}\ \bibnamefont {Coles}},\ }\bibfield  {title} {\bibinfo {title}
  {Absence of barren plateaus in quantum convolutional neural networks},\
  }\href {https://doi.org/10.1103/PhysRevX.11.041011} {\bibfield  {journal}
  {\bibinfo  {journal} {Phys. Rev. X}\ }\textbf {\bibinfo {volume} {11}},\
  \bibinfo {pages} {041011} (\bibinfo {year} {2021})}\BibitemShut {NoStop}%
\bibitem [{\citenamefont {Uvarov}\ and\ \citenamefont
  {Biamonte}(2021)}]{uvarov2020barren}%
  \BibitemOpen
  \bibfield  {author} {\bibinfo {author} {\bibfnamefont {A.}~\bibnamefont
  {Uvarov}}\ and\ \bibinfo {author} {\bibfnamefont {J.~D.}\ \bibnamefont
  {Biamonte}},\ }\bibfield  {title} {\bibinfo {title} {On barren plateaus and
  cost function locality in variational quantum algorithms},\ }\href
  {https://doi.org/10.1088/1751-8121/abfac7} {\bibfield  {journal} {\bibinfo
  {journal} {J. Phys. A Math. Theor.}\ }\textbf {\bibinfo {volume} {54}},\
  \bibinfo {pages} {245301} (\bibinfo {year} {2021})}\BibitemShut {NoStop}%
\bibitem [{\citenamefont {McClean}\ \emph {et~al.}(2018)\citenamefont
  {McClean}, \citenamefont {Boixo}, \citenamefont {Smelyanskiy}, \citenamefont
  {Babbush},\ and\ \citenamefont {Neven}}]{mcclean2018barren}%
  \BibitemOpen
  \bibfield  {author} {\bibinfo {author} {\bibfnamefont {J.~R.}\ \bibnamefont
  {McClean}}, \bibinfo {author} {\bibfnamefont {S.}~\bibnamefont {Boixo}},
  \bibinfo {author} {\bibfnamefont {V.~N.}\ \bibnamefont {Smelyanskiy}},
  \bibinfo {author} {\bibfnamefont {R.}~\bibnamefont {Babbush}},\ and\ \bibinfo
  {author} {\bibfnamefont {H.}~\bibnamefont {Neven}},\ }\bibfield  {title}
  {\bibinfo {title} {Barren plateaus in quantum neural network training
  landscapes},\ }\href {https://doi.org/10.1038/s41467-018-07090-4} {\bibfield
  {journal} {\bibinfo  {journal} {Nat. Commun.}\ }\textbf {\bibinfo {volume}
  {9}},\ \bibinfo {pages} {1} (\bibinfo {year} {2018})}\BibitemShut {NoStop}%
\bibitem [{\citenamefont {Cerezo}\ \emph {et~al.}(2021)\citenamefont {Cerezo},
  \citenamefont {Sone}, \citenamefont {Volkoff}, \citenamefont {Cincio},\ and\
  \citenamefont {Coles}}]{cerezo2021barren}%
  \BibitemOpen
  \bibfield  {author} {\bibinfo {author} {\bibfnamefont {M.}~\bibnamefont
  {Cerezo}}, \bibinfo {author} {\bibfnamefont {A.}~\bibnamefont {Sone}},
  \bibinfo {author} {\bibfnamefont {T.}~\bibnamefont {Volkoff}}, \bibinfo
  {author} {\bibfnamefont {L.}~\bibnamefont {Cincio}},\ and\ \bibinfo {author}
  {\bibfnamefont {P.~J.}\ \bibnamefont {Coles}},\ }\bibfield  {title} {\bibinfo
  {title} {Cost-function-dependent barren plateaus in shallow quantum neural
  networks},\ }\href {https://doi.org/10.1038/s41467-021-21728-w} {\bibfield
  {journal} {\bibinfo  {journal} {Nat. Commun.}\ }\textbf {\bibinfo {volume}
  {12}},\ \bibinfo {pages} {1791} (\bibinfo {year} {2021})}\BibitemShut
  {NoStop}%
\bibitem [{\citenamefont {Wang}\ \emph {et~al.}(2021)\citenamefont {Wang},
  \citenamefont {Fontana}, \citenamefont {Cerezo}, \citenamefont {Sharma},
  \citenamefont {Sone}, \citenamefont {Cincio},\ and\ \citenamefont
  {Coles}}]{wang2020noise}%
  \BibitemOpen
  \bibfield  {author} {\bibinfo {author} {\bibfnamefont {S.}~\bibnamefont
  {Wang}}, \bibinfo {author} {\bibfnamefont {E.}~\bibnamefont {Fontana}},
  \bibinfo {author} {\bibfnamefont {M.}~\bibnamefont {Cerezo}}, \bibinfo
  {author} {\bibfnamefont {K.}~\bibnamefont {Sharma}}, \bibinfo {author}
  {\bibfnamefont {A.}~\bibnamefont {Sone}}, \bibinfo {author} {\bibfnamefont
  {L.}~\bibnamefont {Cincio}},\ and\ \bibinfo {author} {\bibfnamefont {P.~J.}\
  \bibnamefont {Coles}},\ }\bibfield  {title} {\bibinfo {title} {Noise-induced
  barren plateaus in variational quantum algorithms},\ }\href
  {https://doi.org/10.1038/s41467-021-27045-6} {\bibfield  {journal} {\bibinfo
  {journal} {Nat. Commun.}\ }\textbf {\bibinfo {volume} {12}},\ \bibinfo
  {pages} {1} (\bibinfo {year} {2021})}\BibitemShut {NoStop}%
\bibitem [{\citenamefont {Zhang}\ \emph
  {et~al.}(2022{\natexlab{b}})\citenamefont {Zhang}, \citenamefont {Sone},\
  and\ \citenamefont {Zhuang}}]{zhang2022quantum}%
  \BibitemOpen
  \bibfield  {author} {\bibinfo {author} {\bibfnamefont {B.}~\bibnamefont
  {Zhang}}, \bibinfo {author} {\bibfnamefont {A.}~\bibnamefont {Sone}},\ and\
  \bibinfo {author} {\bibfnamefont {Q.}~\bibnamefont {Zhuang}},\ }\bibfield
  {title} {\bibinfo {title} {Quantum computational phase transition in
  combinatorial problems},\ }\href {https://doi.org/10.1038/s41534-022-00596-2}
  {\bibfield  {journal} {\bibinfo  {journal} {npj Quantum Inf.}\ }\textbf
  {\bibinfo {volume} {8}},\ \bibinfo {pages} {87} (\bibinfo {year}
  {2022}{\natexlab{b}})}\BibitemShut {NoStop}%
\bibitem [{\citenamefont {Jeyaretnam}\ \emph {et~al.}(2021)\citenamefont
  {Jeyaretnam}, \citenamefont {Richter},\ and\ \citenamefont
  {Pal}}]{jeyaretnam2021quantum}%
  \BibitemOpen
  \bibfield  {author} {\bibinfo {author} {\bibfnamefont {J.}~\bibnamefont
  {Jeyaretnam}}, \bibinfo {author} {\bibfnamefont {J.}~\bibnamefont
  {Richter}},\ and\ \bibinfo {author} {\bibfnamefont {A.}~\bibnamefont {Pal}},\
  }\bibfield  {title} {\bibinfo {title} {Quantum scars and bulk coherence in a
  symmetry-protected topological phase},\ }\href
  {https://doi.org/10.1103/PhysRevB.104.014424} {\bibfield  {journal} {\bibinfo
   {journal} {Phys. Rev. B}\ }\textbf {\bibinfo {volume} {104}},\ \bibinfo
  {pages} {014424} (\bibinfo {year} {2021})}\BibitemShut {NoStop}%
\bibitem [{\citenamefont {Chen}\ \emph {et~al.}(2014)\citenamefont {Chen},
  \citenamefont {Lu},\ and\ \citenamefont {Vishwanath}}]{chen2014symmetry}%
  \BibitemOpen
  \bibfield  {author} {\bibinfo {author} {\bibfnamefont {X.}~\bibnamefont
  {Chen}}, \bibinfo {author} {\bibfnamefont {Y.-M.}\ \bibnamefont {Lu}},\ and\
  \bibinfo {author} {\bibfnamefont {A.}~\bibnamefont {Vishwanath}},\ }\bibfield
   {title} {\bibinfo {title} {Symmetry-protected topological phases from
  decorated domain walls},\ }\href {https://doi.org/10.1038/ncomms4507}
  {\bibfield  {journal} {\bibinfo  {journal} {Nat. Commun.}\ }\textbf {\bibinfo
  {volume} {5}},\ \bibinfo {pages} {3507} (\bibinfo {year} {2014})}\BibitemShut
  {NoStop}%
\bibitem [{\citenamefont {Haldane}(1983)}]{haldane1983nonlinear}%
  \BibitemOpen
  \bibfield  {author} {\bibinfo {author} {\bibfnamefont {F.~D.~M.}\
  \bibnamefont {Haldane}},\ }\bibfield  {title} {\bibinfo {title} {Nonlinear
  {F}ield {T}heory of {L}arge-{S}pin {H}eisenberg {A}ntiferromagnets:
  {S}emiclassically {Q}uantized {S}olitons of the {O}ne-{D}imensional
  {E}asy-{A}xis {N}\'eel {S}tate},\ }\href
  {https://doi.org/10.1103/PhysRevLett.50.1153} {\bibfield  {journal} {\bibinfo
   {journal} {Phys. Rev. Lett.}\ }\textbf {\bibinfo {volume} {50}},\ \bibinfo
  {pages} {1153} (\bibinfo {year} {1983})}\BibitemShut {NoStop}%
\bibitem [{\citenamefont {Gu}\ and\ \citenamefont {Wen}(2009)}]{gu2009tensor}%
  \BibitemOpen
  \bibfield  {author} {\bibinfo {author} {\bibfnamefont {Z.-C.}\ \bibnamefont
  {Gu}}\ and\ \bibinfo {author} {\bibfnamefont {X.-G.}\ \bibnamefont {Wen}},\
  }\bibfield  {title} {\bibinfo {title} {Tensor-entanglement-filtering
  renormalization approach and symmetry-protected topological order},\ }\href
  {https://doi.org/10.1103/PhysRevB.80.155131} {\bibfield  {journal} {\bibinfo
  {journal} {Phys. Rev. B}\ }\textbf {\bibinfo {volume} {80}},\ \bibinfo
  {pages} {155131} (\bibinfo {year} {2009})}\BibitemShut {NoStop}%
\bibitem [{\citenamefont {Pollmann}\ \emph {et~al.}(2012)\citenamefont
  {Pollmann}, \citenamefont {Berg}, \citenamefont {Turner},\ and\ \citenamefont
  {Oshikawa}}]{pollmann2012symmetry}%
  \BibitemOpen
  \bibfield  {author} {\bibinfo {author} {\bibfnamefont {F.}~\bibnamefont
  {Pollmann}}, \bibinfo {author} {\bibfnamefont {E.}~\bibnamefont {Berg}},
  \bibinfo {author} {\bibfnamefont {A.~M.}\ \bibnamefont {Turner}},\ and\
  \bibinfo {author} {\bibfnamefont {M.}~\bibnamefont {Oshikawa}},\ }\bibfield
  {title} {\bibinfo {title} {Symmetry protection of topological phases in
  one-dimensional quantum spin systems},\ }\href
  {https://doi.org/10.1103/PhysRevB.85.075125} {\bibfield  {journal} {\bibinfo
  {journal} {Phys. Rev. B}\ }\textbf {\bibinfo {volume} {85}},\ \bibinfo
  {pages} {075125} (\bibinfo {year} {2012})}\BibitemShut {NoStop}%
\bibitem [{\citenamefont {Spall}(1998)}]{spall1998overview}%
  \BibitemOpen
  \bibfield  {author} {\bibinfo {author} {\bibfnamefont {J.~C.}\ \bibnamefont
  {Spall}},\ }\bibfield  {title} {\bibinfo {title} {An overview of the
  simultaneous perturbation method for efficient optimization},\ }\href
  {https://secwww.jhuapl.edu/techdigest/content/techdigest/pdf/V19-N04/19-04-Spall.pdf}
  {\bibfield  {journal} {\bibinfo  {journal} {Johns {H}opkins {APL} {T}ech.
  {D}ig.}\ }\textbf {\bibinfo {volume} {19}},\ \bibinfo {pages} {482 }
  (\bibinfo {year} {1998})}\BibitemShut {NoStop}%
\bibitem [{\citenamefont {Spall}(1999)}]{spall1999stochastic}%
  \BibitemOpen
  \bibfield  {author} {\bibinfo {author} {\bibfnamefont {J.~C.}\ \bibnamefont
  {Spall}},\ }\bibfield  {title} {\bibinfo {title} {Stochastic optimization:
  stochastic approximation and simulated annealing},\ }\href
  {https://onlinelibrary.wiley.com/doi/abs/10.1002/047134608X.W1044} {\bibfield
   {journal} {\bibinfo  {journal} {Wiley Encyclopedia of electrical and
  electronics engineering}\ }\textbf {\bibinfo {volume} {20}},\ \bibinfo
  {pages} {529 } (\bibinfo {year} {1999})}\BibitemShut {NoStop}%
\bibitem [{Note3()}]{Note3}%
  \BibitemOpen
  \bibinfo {note} {Specifically, we measure the output of the trained QCNN with
  some projectors and assign the label that corresponds to the highest
  probability of the measurement results.}\BibitemShut {Stop}%
\end{thebibliography}%

%\onecolumngrid

\end{document}